\documentclass[article]{IEEEtran}
\usepackage{epsfig,latexsym,amsmath,amssymb,setspace,cite,color,graphicx,algorithm,algorithmic,cases,subcaption}
\usepackage[percent]{overpic}

\usepackage{amsmath}
\usepackage{amsfonts}
\usepackage{amssymb}
\usepackage{dsfont}
\usepackage{bm}
\usepackage{amsthm}
\usepackage{newlfont}
\usepackage{float}
\usepackage{hyperref}
\usepackage{algorithm}
\usepackage{algorithmic}
\usepackage{enumerate}
\usepackage{chngcntr}
\usepackage{mathtools}
\usepackage{breqn}
\usepackage{stmaryrd}
\usepackage{cite}
\usepackage[yyyymmdd]{datetime}
\hypersetup{
    colorlinks=true, 
    linkcolor= blue, 
    citecolor=blue 
}
\usepackage{titlesec}

\begin{document}
\sloppy
\allowdisplaybreaks[1]

\newcommand\numberthis{\addtocounter{equation}{1}\tag{\theequation}}

\newtheorem{thm}{Theorem} 
\newtheorem{lem}{Lemma}
\newtheorem{prop}{Proposition}
\newtheorem{cor}{Corollary}
\newtheorem{defn}{Definition}
\newcommand{\remarkend}{\IEEEQEDopen}
\newtheorem{remark}{Remark}
\newtheorem{rem}{Remark}
\newtheorem{ex}{Example}
\newtheorem{pro}{Property}

\newenvironment{example}[1][Example]{\begin{trivlist}
\item[\hskip \labelsep {\bfseries #1}]}{\end{trivlist}}
  
\renewcommand{\qedsymbol}{ \begin{tiny}$\blacksquare$ \end{tiny} }

\renewcommand{\leq}{\leqslant}
\renewcommand{\geq}{\geqslant}

\title {Distributed Secret Sharing over a Public Channel from Correlated Random Variables
}
 
 \author{\IEEEauthorblockN{R\'emi A. Chou}

\thanks{
R. Chou is with the Department of Computer Science and Engineering, The University of Texas at Arlington, Arlington, TX 76019. E-mail: remi.chou@uta.edu. Part of this work has been presented at the 2018 IEEE International Symposium on Information Theory (ISIT) \cite{chou2018secret}. This work was supported in part by NSF grants CCF-1850227 and CCF-2047913.}
}
\maketitle
\begin{abstract}
We consider a secret-sharing model where a dealer distributes the shares of a secret among a set of participants with the constraint that only predetermined subsets of participants must be able to reconstruct the secret by pooling their shares. Our study generalizes Shamir's secret-sharing model in three directions. First, we allow a joint design of the protocols for the creation of the shares  and the distribution of the shares, instead of  constraining the model to independent designs. Second, instead of assuming that the participants and the dealer have access to information-theoretically secure channels at no cost, we assume that they have access to a public channel and correlated randomness. Third, motivated by a wireless network setting where the correlated randomness is obtained from channel gain measurements, we explore a distributed setting where the dealer is an entity made of multiple sub-dealers.
Our main results are inner and outer regions for the achievable secret rates that the dealer and the participants can obtain in this model. To this end, we develop two new achievability techniques, a first one to successively handle reliability and security constraints in a distributed setting, and a second one to reduce a multi-dealer setting to multiple single-user dealer settings. Our results yield the capacity region for threshold access structures when the correlated randomness corresponds to pairwise secret keys shared between each sub-dealer and each participant, and the capacity for the all-or-nothing access structure in the presence of a single dealer and arbitrarily correlated randomness.

 \end{abstract} 
 \section{Introduction}
Consider a dealer who distributes to $L$ participants the shares of a secret $S$ with the requirements that any $t$ participants are able to reconstruct $S$ by pooling their shares, and any subsets of participants with cardinality strictly smaller than $t$ must be unable to learn anything about the secret, in an information-theoretic sense.
More specifically, the dealer forms $L$ shares $(M_1, \ldots, M_L)$ from the secret $S$, and transmits each share $S_l$ to Participant $l \in \{1,\ldots,L\}$ via an individual and information-theoretically secure channel. The setting is depicted in Figure \ref{fig1} for the case $(L,t)=(3,2)$. 

\begin{figure}[t!]
    \centering
    \begin{subfigure}[t]{0.5\textwidth}
        \centering
   \includegraphics[width=5.04cm]{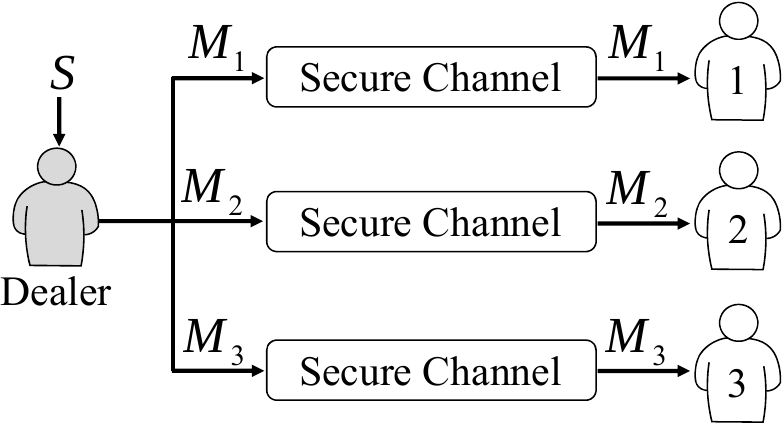}
    \subcaption{Distribution of the shares}
    \label{fig:5a}  
    \end{subfigure}
    ~~~~~~ 
    \begin{subfigure}[t]{0.5\textwidth}
        \centering
    \includegraphics[width=2.45cm]{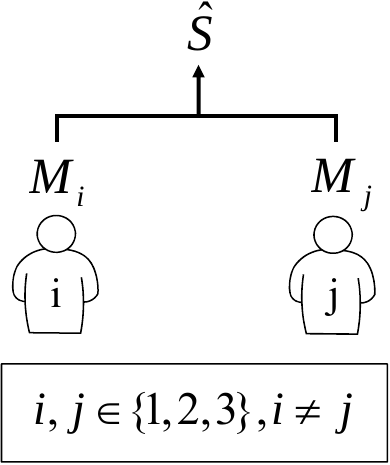}
    \subcaption{Reconstruction of the secret}
    \label{fig:5b}
    \end{subfigure}
\caption{Traditional secret sharing with $L=3$ participants and $t=2$.} \label{fig1}
\end{figure}

\begin{figure}[t!]
    \centering
    \begin{subfigure}[t]{0.5\textwidth}
        \centering
   \includegraphics[width=6.05cm]{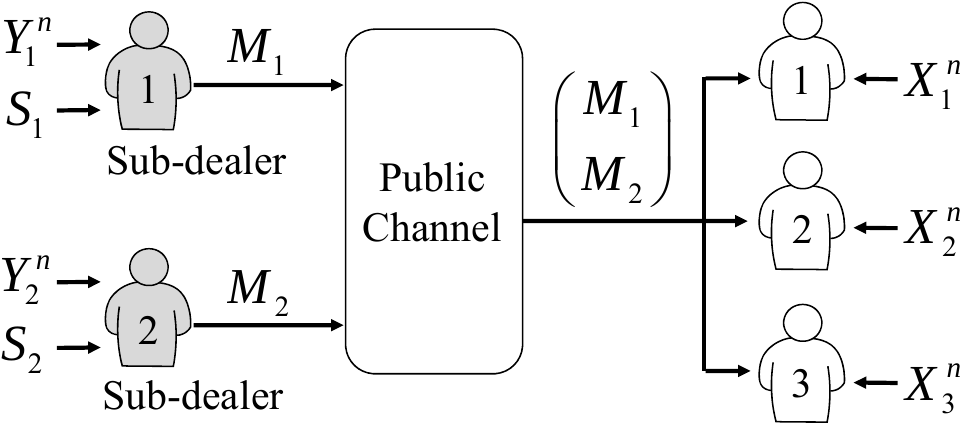}
    \subcaption{Distribution of the shares} 
    \end{subfigure}
    ~~~~~~ 
    \begin{subfigure}[t]{0.5\textwidth}
        \centering
    \includegraphics[width=3.89cm]{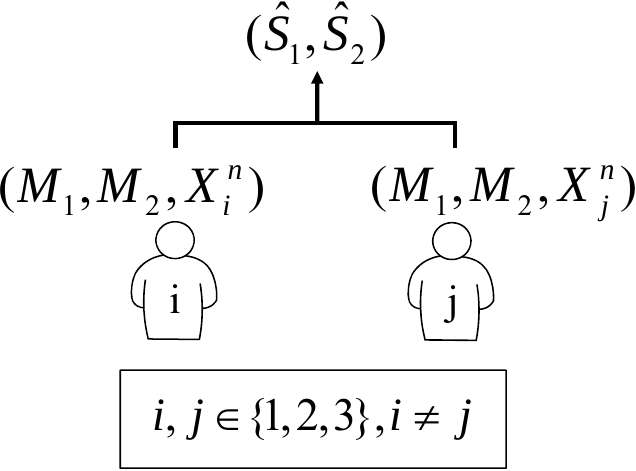}
    \subcaption{Reconstruction of the secret}
    \end{subfigure}
\caption{Proposed secret sharing model with two sub-dealers, three participants, and a reconstruction threshold $t=2$.} \label{fig1b}
\end{figure}
This secret-sharing problem was first introduced by Shamir in~\cite{shamir1979share}  and, independently, by Blakley in \cite{blakley}.  Subsequently, numerous variants have been extensively studied in the computer science literature, see, for instance, \cite{stinson2005cryptography,beimel} and references therein. In these studies of information-theoretically secure secret sharing, it is assumed that the dealer can distribute to the participants the shares of the secret through information-theoretically secure channels that are available at no cost, as in the setting depicted in Figure \ref{fig:5a}.  
Recently, to avoid this assumption, an information-theoretic treatment of secret sharing over noisy channels has  been proposed in~\cite{zou2015information} by leveraging information-theoretic security results at the physical~layer. Specifically, in \cite{zou2015information}, the information-theoretically secure channels of  traditional secret-sharing models are replaced by a noisy broadcast  channel from the dealer to the participants so that secret sharing reduces to physical-layer security for a coumpound wiretap channel~\cite{Liang09}.

In this paper, as illustrated in Figure \ref{fig1b}, we formulate a secret-sharing model that generalizes Shamir's secret-sharing model in three directions. First, our model allows a joint design of the creation of the shares by the dealer and the distribution of the shares by the dealer to the participants. This contrasts with Shamir's model which considers these two phases independently,  since information-theoretically secure channels are available between the dealer and each participant for the distribution phase. Second, while Shamir's model assumes that the participants and the dealer  have access to information-theoretically secure channels  for the distribution of shares, we, instead, only rely on a public channel and correlated randomness in the form of realizations of independently and identically distributed random variables.\footnote{Note that the latter resources are more general than the former resources since they can be used to implement information-theoretically secure channels via a one-time pad when the correlated randomness corresponds to secret keys shared between the dealer and the participants.} Third, motivated by a wireless network setting discussed next, we further explore the problem of secret sharing in a distributed setting where the dealer is an entity made of multiple sub-dealers.

 Our setting is formally described in Section \ref{sec:pre} and can be explained at a high level as follows. Assume that the participants and the dealer, made of multiple sub-dealers distributed in space, observe independently and identically distributed realizations of correlated random variables, that have, for instance, been obtained in a wireless communication network from channel gain measurements  after  appropriate  manipulations~\cite{wilson2007channel,wallace2010automatic,ye2010information,pierrot2013experimental}. The dealer wishes to share a secret with the participants with the requirement that only predefined subsets of participants are able to reconstruct the secret, while any other subsets of participants that pool all their knowledge must remain ignorant, in an information-theoretic sense, of the secret.  
  We are interested in characterizing the set of all achievable secret rates that the dealer can obtain via its sub-dealers when those are allowed to communicate with the participants over a  public channel. Note that a potential limitation in the presence of a single dealer is that the correlated random variables available at the dealer and the participants could be such that no positive secret rates are achievable. 
It is precisely to mitigate this eventuality that we consider a dealer that can deploy in space sub-dealers using, for instance, mobile stations or drones whose positions are adjusted to change the statistics of the channel gains and avoid sterile spatial configurations.

We summarize the main features of our work as follows:
\begin{enumerate} [(i)]
\item  We study a secret-sharing model that extends Shamir's secret-sharing model in three  directions by considering (1) a joint, instead of independent, design of the share creation and distribution phases, (2) a public channel and correlated randomness, instead of secure channels, and (3) a dealer made of distributed sub-dealers, instead of a single dealer. Specifically, we derive inner and outer regions for the achievable secret rates that the dealer and the participants can obtain    in this model. We obtain capacity results in the case of threshold access structures when the correlated randomness corresponds to pairwise secret keys shared between each sub-dealer and each participant, and in the case of a single-dealer setting for the all-or-nothing access structure and arbitrarily correlated randomness.  
\item  In all our achievabilitiy results, the length of each share  always scales linearly with the size of the secret for any access structures. This comes from the fact that the size of the secret is linear with the number of source observations $n$, and  a share corresponds to the public communication plus $n$ source observations, whose lengths are both linear with $n$. Indeed, the public communication corresponds to a compressed version of the $n$ source observations of all the sub-dealers. The length of the public communication does not depend on the number of participants but does depend on the access structure, in particular, the public communication must allow the secret reconstruction for the group of authorized participants that has the least amount of information in their source observations about the secret. This contrasts with Shamir's secret-sharing model, for which the best known coding schemes require the share size to depend exponentially on the number of participants for some access structures~\cite{beimel}. 
\item  As a by-product of independent interest, for distributed settings, we develop two novel achievability techniques to simultaneously satisfy  reliability  and   security constraints. The first one consists in successively handling the reliability and security constraints. This is done by deriving a new variant of the distributed leftover hash lemma and developing a new coding scheme for distributed reconciliation that can be combined with it.  The second one consists in reducing a distributed setting to multiple single-user~settings.
\end{enumerate}

\subsection{Related work} Our work is related to secret-key generation from correlated random variables and public communication \cite{Maurer93,Ahlswede93}, as correlated randomness and public communication are also the main resources considered in our setting.  However, the analysis of our proposed secret-sharing model does not follow from known results for the secret-key generation models in~\cite{Maurer93,Ahlswede93}, as these models only consider a key exchange between two parties, whereas our setting considers a secret exchange between multipe dealers and multiple participants. The analysis of our proposed secret-sharing model does not follow either from  subsequent multiuser secret-key generation models, e.g.,~\cite{Csiszar04,tavangaran2016secret,zhang2017multi,gohari2010information}, that either do not consider multiple reliability and security constraints simultaneously (and are thus unable to support  access structures as in our secret-sharing model) or do not consider distributed settings (and are thus unable to support our distributed dealer setting). The main technical difficulties in our study precisely come from having to simultaneously deal with (i) a distributed setting due to the presence of multiple sub-dealers, and (ii) information-theoretic security constraints able to support an access structure, i.e., able to ensure that \emph{all} the unauthorized subsets of participants cannot learn information about the secret. 
More specifically, for the all-or-nothing access structure, i.e., when all the participants are needed to reconstruct the secret, we develop a new achievability technique that \emph{successively} handles the reliability and security constraints in the presence of distributed sub-dealers. Perhaps surprisingly, we show that this achievability technique is superior to a random binning strategy that \emph{simultaneously} handles the reliability and security constraints, in the sense that no elimination of  auxiliary rates in the obtained achievability region is necessary. To this end, we derive a new variant of 
the distributed leftover hash lemma~\cite{wullschleger2007oblivious,nascimento2008oblivious,chou2017secret}. While the standard leftover hash lemma~\cite{haastad1999pseudorandom} has been extensively used for secret key generation, e.g., \cite{Bennett95,cachin1997linking,Maurer00},  known proofs techniques to study the \emph{distributed} leftover hash lemma in our problem do not seem optimal. Specifically,   at least two new technical challenges arise in our study: 
(i)~while in a non-distributed setting only one min-entropy appears in the leftover hash lemma, the presence of multiple min-entropies (defined from the marginals of the same joint probability distribution) for a distributed setting complexifies the task of finding good approximations of theses min-entropies, further, (ii)~usual techniques for non-distributed settings, e.g.,~\cite[Lemma 10]{Maurer00}, to study the impact of public communication on the leaked information to an eavesdropper do not lead to tight results in a distributed setting.   
Additionally, we also develop for the all-or-nothing access structure another new achievability technique to reduce the task of coding for a distributed-dealer setting to the task of coding for multiple separate single-dealer settings.  

As alluded to earlier, \cite{zou2015information} considers a channel model version of the model studied in this paper but in the presence of a single dealer. Note that subsequently to the preliminary version~\cite{chou2018secret} of this paper, Reference~\cite{rana2021information} investigated a similar model to the one in this study, but only in the presence of a \emph{single-dealer}, when the participants and the dealer observe realizations of correlated Gaussian variables. Note also that~\cite{csiszar2} investigated another secret-sharing problem from correlated random variables and public discussion in the absence of a designated dealer and for special kinds of access structures that are not monotone. By contrast, in this work, we consider arbitrary monotone access structures, as defined in~\cite{benaloh1988generalized}.

Finally, note that distributed secret sharing has also been studied from a different perspective in \cite{soleymani2020distributed,khalesi2021capacity}. In these references, a dealer stores information in multiple storage nodes such that each participant who has access to predefined storage nodes can reconstruct a secret but cannot learn information about the secrets of the other participants. The main difference, in terms of assumptions, between \cite{soleymani2020distributed,khalesi2021capacity} and our setting is that it is assumed in \cite{soleymani2020distributed,khalesi2021capacity} that the dealer can store information in multiple nodes and thus that there exist information-theoretically secure channels between the dealer and each node, similar to the standard assumption in Shamir's secret sharing. By contrast, we do not make this assumption in our setting and instead only rely on a public channel and correlated randomness. For this reason, in \cite{soleymani2020distributed,khalesi2021capacity} the nature of the problem studied is different, specifically, in \cite{soleymani2020distributed,khalesi2021capacity}, the minimization of communication rates is sought out, whereas in our setting, for given source statistics, the maximization of the secret length is sought out.

 \subsection{Paper organization}
The remainder of the paper is organized as follows.  We formally define the problem in Section~\ref{sec:pre} and state our main results in Section~\ref{sec:res}. We present our achievability proofs and converse proofs in Sections~\ref{sec:achievability} and \ref{appth2ga}, respectively. We prove the optimality of our results in some special cases in Section~\ref{secoptimal}. We propose an extension of all our results to the case of chosen (instead of random) secrets in Section \ref{secextension}. Finally, we provide concluding remarks in Section~\ref{sec:concl}.

\section{Notation} \label{secnot}
 
For any $a \in \mathbb{R}^*$, define $\llbracket 1,  a \rrbracket \triangleq [1,\lceil a \rceil] \cap \mathbb{N}$.  The indicator function is denoted by $\mathds{1}\{ \omega \}$, which is equal to~$1$ if the predicate $\omega$ is true and $0$ otherwise.    Let $\mathbb{V}(\cdot,\cdot)$ denote the variational distance. For a given set $\mathcal{S}$, let $2^{\mathcal{S}}$  denote the power set of~$\mathcal{S}$, and $|\mathcal{S}|$ denotes the cardinality of~$\mathcal{S}$.  Finally, let $\bigtimes$ denote the Cartesian product.

\section{Problem statement}\label{sec:pre}

For $L,D \in \mathbb{N}^*$, define the sets $\mathcal{L} \triangleq \llbracket 1, L \rrbracket $ and $\mathcal{D} \triangleq \llbracket 1, D \rrbracket $. Consider $L$ finite alphabets $(\mathcal{X}_l)_{l \in \mathcal{L}}$, $D$ finite alphabets $(\mathcal{Y}_d)_{d \in \mathcal{D}}$, and define $\mathcal{X}_{\mathcal{L}} \triangleq \bigtimes_{l \in \mathcal{L}} \mathcal{X}_l$ and  $\mathcal{Y}_{\mathcal{D}} \triangleq \bigtimes_{d \in \mathcal{D}} \mathcal{Y}_d$.  Then, consider a discrete memoryless source $(\mathcal{X}_{\mathcal{L}}\times \mathcal{Y}_{\mathcal{D}},p_{X_{\mathcal{L}}Y_{\mathcal{D}}})$, where $X_{\mathcal{L}} \triangleq (X_l)_{l\in \mathcal{L}}$ and $Y_{\mathcal{D}} \triangleq (Y_d)_{d\in \mathcal{D}}$. $n \in \mathbb{N}$ independent and identically distributed realizations of the source are denoted by $(X^n_{\mathcal{L}},Y^n_{\mathcal{D}})$, where $X^n_{\mathcal{L}} \triangleq (X^n_l)_{l\in \mathcal{L}}$ and $Y^n_{\mathcal{D}} \triangleq (Y^n_d)_{d\in \mathcal{D}}$.  In the following, for any subset $ \mathcal{T} \subseteq \mathcal{L}$, we use the notation $X^n_{\mathcal{T}} \triangleq (X^n_l)_{l \in \mathcal{T}}$.

 As formalized next, we consider $D$ sub-dealers and $L$ participants, who  each observes one component of the discrete memoryless source. Through public communication from the sub-dealers to the participants, their objective is to generate $D$ random secrets such that authorized subsets of participants can reconstruct the secrets, whereas unauthorized subsets of participants cannot learn any information about the secrets. We highlight that in the following definitions the secrets are random, however, in Section \ref{secextension}, we explain how to address the same setting when the values of the secrets are chosen by the sub-dealers.

\begin{defn}[Monotone access structure \cite{benaloh1988generalized}]
A set $\mathbb{A}$ of subsets of $\mathcal{L}$ is  a monotone access structure when for any $\mathcal{T} \subseteq \mathcal{L}$, if $\mathcal{T}$ contains a set that belongs to $\mathbb{A}$, then $\mathcal{T}$ also belongs to $\mathbb{A}$. We write the complement of $\mathbb{A}$ in $2^{\mathcal{L}}$ as $\mathbb{U} \triangleq 2^{\mathcal{L}} \backslash \mathbb{A}$.
\end{defn}

\begin{defn} \label{definition_modelg}
For $d \in \mathcal{D}$, define the alphabet  $\mathcal{S}_d \triangleq \llbracket 1 ,2^{nR_d} \rrbracket$ and  $\mathcal{S}_{\mathcal{D}}\triangleq \bigtimes_{d\in \mathcal{D}} \mathcal{S}_d$. A $( (2^{nR_d})_{d \in \mathcal{D}},\mathbb{A},\mathbb{U},n)$ secret-sharing strategy consists of:
\begin{itemize}
\item A monotone access structure $\mathbb{A}$.
\item $D$ sub-dealers indexed by the set $\mathcal{D}$.
\item $L$ participants indexed by the set $\mathcal{L}$.
\item $D$ encoding functions $(f_d)_{d\in\mathcal{D}}$, where $f_d:\mathcal{Y}^n_d \to  \mathcal{M}_d$, $d\in\mathcal{D}$, with $\mathcal{M}_d$ an arbitrary finite alphabet.
\item $D$ encoding functions $(g_d)_{d\in\mathcal{D}}$, where $g_d:\mathcal{Y}^n_d \to  \mathcal{S}_d$, $d\in\mathcal{D}$.
\item $|\mathbb{A}|\times D$ decoding functions $(h_{\mathcal{A},d})_{\mathcal{A}\in\mathbb{A}, d\in\mathcal{D}}$, where $h_{\mathcal{A},d}: \mathcal{X}_{\mathcal{A}}^n  \times \mathcal{M}_{\mathcal{D}} \to \mathcal{S}_{d}$, $\mathcal{A} \in \mathbb{A}$, with  $\mathcal{X}_{\mathcal{A}}^n \triangleq \bigtimes_{a\in \mathcal{A}} \mathcal{X}_a^n$ and $\mathcal{M}_{\mathcal{D}} \triangleq \bigtimes_{d\in \mathcal{D}} \mathcal{M}_d$.
\end{itemize}
and operates as follows:
\begin{itemize}
\item Sub-dealer $d \in \mathcal{D}$ observes $Y_d^{n}$.
\item Participant $l\in \mathcal{L}$ observes $X_l^{n}$.
\item Sub-dealer $d \in \mathcal{D}$ sends over a noiseless public authenticated channel the  public communication $M_d \triangleq f_d(Y_d^{n})$ to the participants. We write the global communication of all the sub-dealers as $M_{\mathcal{D}} \triangleq (M_d)_{d\in\mathcal{D}}$. 
\item Sub-dealer $d\in\mathcal{D}$ computes $S_d \triangleq g_d(Y_d^{n})$.
\item Any subset of participants $\mathcal{A} \in \mathbb{A}$ can compute for $d\in\mathcal{D}$, $\widehat{S}_d(\mathcal{A}) \triangleq h_{\mathcal{A},d} (X^n_{\mathcal{A}},M_{\mathcal{D}})$, and thus form
$\widehat{{S}}_{\mathcal{D}} (\mathcal{A})\triangleq(\widehat{S}_d(\mathcal{A}))_{d\in\mathcal{D}}$, an estimate of ${S}_{\mathcal{D}} \triangleq ({S}_d)_{d\in\mathcal{D}}$.
\end{itemize}
\end{defn}

\begin{defn} \label{def}
A secret rate-tuple $(R_d)_{d\in \mathcal{D}}$ is achievable if there exists a sequence of $((2^{nR_d})_{d \in \mathcal{D}}$,$\mathbb{A}$, $\mathbb{U},n)$ secret-sharing strategies such that
\begin{align}
\lim_{n \to \infty}	\displaystyle\max_{ \mathcal{A} \in \mathbb{A}}    \mathbb{P}\left[\widehat{{S}}_{\mathcal{D}} (\mathcal{A}) \neq {S}_{\mathcal{D}} \right] &= 0 \text{ (Reliability),} \label{eqrel} \\
\lim_{n \to \infty}	\displaystyle\max_{ \mathcal{U} \in \mathbb{U}}    I\left({S}_{\mathcal{D}} ; M_{\mathcal{D}}, X^n_{\mathcal{U}}  \right) &= 0 \text{ (Strong Security),}\label{eqSeca} \\
	\lim_{n \to \infty} \log |\mathcal{S}_{\mathcal{D}} | - H({S}_{\mathcal{D}}) &= 0 \text{ (Secret Uniformity)}.\label{eqU}
\end{align}
Let $\mathcal{C}(\mathbb{A})$ denote the set of all achievable secret rate-tuples. When $D=1$,  ${C}(\mathbb{A})$ denotes the supremum   of all achievable secret rates and is called the secret capacity. 
\end{defn}
\eqref{eqrel} means that any subset of participants in $\mathbb{A}$ is able to recover the secret, while \eqref{eqSeca} means that any subset of participants in $\mathbb{U}$ cannot learn any information about the secret even if they pool their observations and the  public communication sent by all the sub-dealers. \eqref{eqU}  means that
the secret is nearly uniform, i.e., the entropy of the secret is nearly equal to its length. In other words, \eqref{eqU} means that we seek secret-sharing strategies that \emph{maximize the entropy of the secret}.

\begin{ex}
Suppose that there are $L=3$ participants who observe $(X_1^n,X_2^n,X_3^n)$ and $D=2$ sub-dealers who observe $(Y_1^n,Y_2^n)$  as depicted in Figure \ref{fig:secreta}.  In a first phase, depicted in Figure~\ref{fig:secreta}, Sub-dealer $i \in \{1,2\}$ computes $S_i \triangleq g_i(Y^n_i)$ and $M_i \triangleq f_i(Y_i^n)$, and publicly shares $M_i$ with all the participants. In this example, suppose that  $M_1$, $M_2$, $S_1$, $S_2$ are created such that any two participants must be able to recover $(S_1,S_2)$, i.e., the access structure is defined as $\mathbb{A} \triangleq \{ \{1,2\}, \{1,3\}, \{2,3\}, \{1,2,3\} \}$, but a single participant must not learn information about $(S_1,S_2)$ as described by Equation \eqref{eqSeca} with $\mathbb{U} \triangleq 2^{\mathcal{L}} \backslash \mathbb{A} = \{ \{1\},\{2\},\{3\}\}$. Hence, in a second phase,  depicted in Figure~\ref{fig:secretb}, any two participants $i\in\{1,2,3\}$ and $j\in \{1,2,3\} \backslash \{i\}$ who pool their information, i.e., $(M_1,M_2,X_i^n,X_j^n)$, can estimate  $(S_1,S_2)$ as $(\hat{S}_1,\hat{S}_2)\triangleq (h_{\mathcal{A},1}(X^n_{\mathcal{A}},M_{\mathcal{D}} ),h_{\mathcal{A},2}(X^n_{\mathcal{A}},M_{\mathcal{D}} ))$, where $\mathcal{A} \triangleq \{i,j\}\in \mathbb{A}$ and $M_{\mathcal{D}} \triangleq (M_1,M_2)$.
\end{ex}

\begin{figure}[t!]
    \centering
    \begin{subfigure}[t]{0.5\textwidth}
        \centering
    \includegraphics[width=6.8cm]{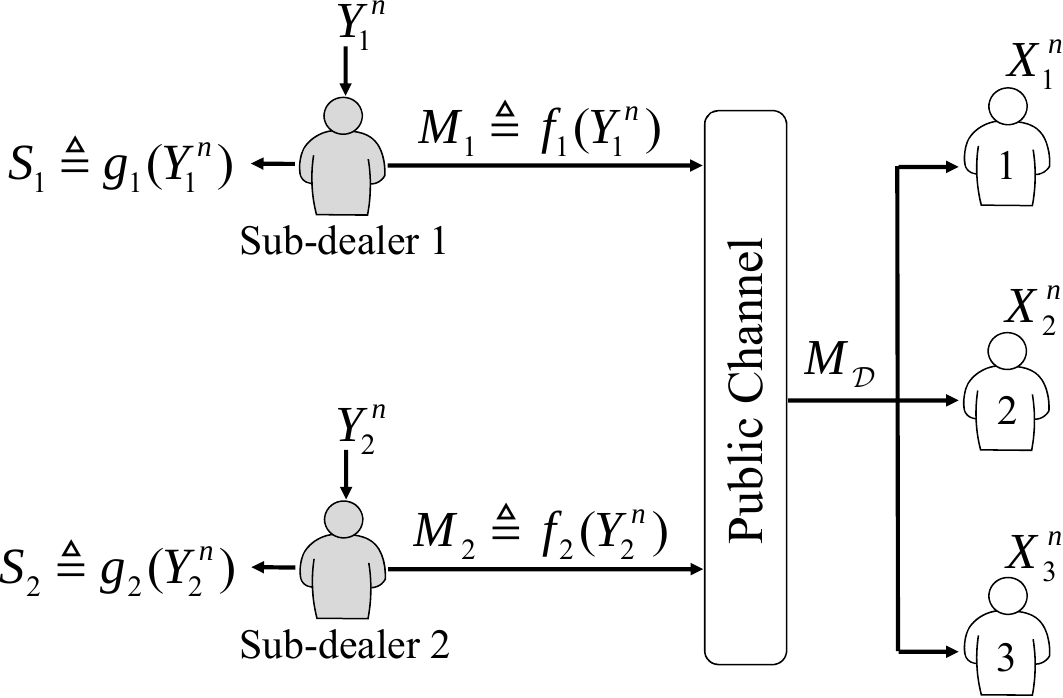}
        \caption{Formation and distribution of shares} \label{fig:secreta}
    \end{subfigure}
    ~~~~~~~~~~ 
    \begin{subfigure}[t]{0.5\textwidth}
        \centering
    \includegraphics[width=4.8cm]{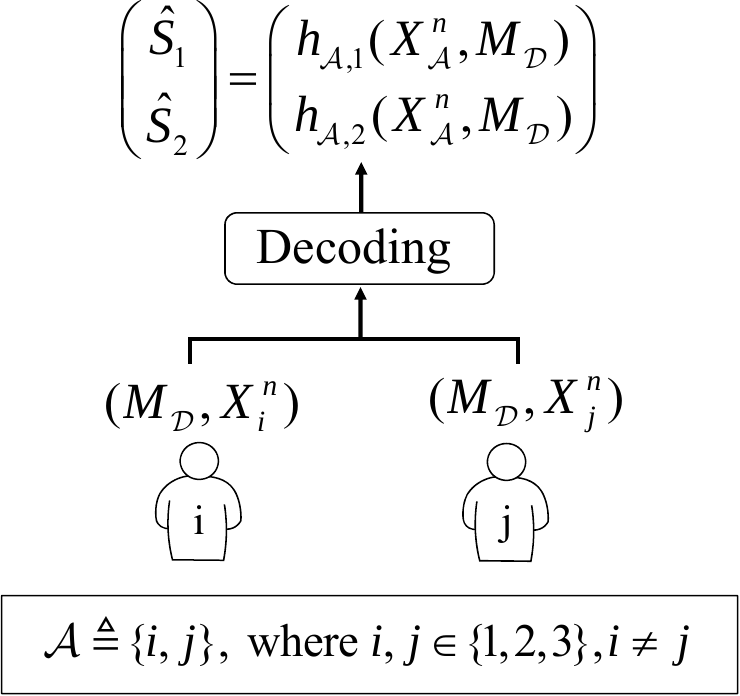}
        \caption{Reconstruction of the secret~$(S_1,S_2)$} \label{fig:secretb}
    \end{subfigure}
\caption{Secret sharing with $D=2$ sub-dealers, $L=3$ users, and the access structure $\mathbb{A} \triangleq \{ \{1,2\}, \{1,3\}, \{2,3\}, \{1,2,3\} \}$.}
\end{figure}

\section{Results} \label{sec:res}
In the following, for a rate-tuple $(R_d)_{d\in\mathcal{D}} \in \mathbb{R}^D_+$ and $\mathcal{S} \subseteq \mathcal{D}$, we use the notation $R_{\mathcal{S }} \triangleq \sum_{i \in \mathcal{S}} R_i$.

\subsection{General access structures} \label{sec:resGEN}
\subsubsection{Results for an arbitrary number $D$ of sub-dealers}

The achievability scheme to derive Theorem \ref{th1ga} relies on  random binning designed to simultaneously satisfy the reliability condition \eqref{eqrel} and the security condition \eqref{eqSeca}.

 \begin{thm}[Inner bound]  \label{th1ga}
We have $\mathcal{R}^{(\textup{in})}(\mathbb{A}) \subseteq \mathcal{C}(\mathbb{A})$, where
\begin{align*}
&\mathcal{R}^{(\textup{in})}(\mathbb{A})\\
&\triangleq \textup{Proj}_{(R_d)_{d\in\mathcal{D}}} \left\{ \left(R_d,R_d'\right)_{d\in\mathcal{D}} : \right.\\
& \left. \phantom{-----}\begin{array}{rl}
R_{\mathcal{S}}' &\geq \displaystyle\max_{\mathcal{A} \in \mathbb{A}} H(Y_{\mathcal{S}}|Y_{\mathcal{S}^c}X_{\mathcal{A}}), \forall \mathcal{S} \subseteq \mathcal{D}\\
R_{\mathcal{S}}' + R_{\mathcal{S}} &\leq \displaystyle\min_{\mathcal{U} \in \mathbb{U}} H(Y_{\mathcal{S}}|X_{\mathcal{U}}), \forall \mathcal{S} \subseteq \mathcal{D}
\end{array}      \right\}\!,
\end{align*}
where $\textup{Proj}_{(R_d)_{d\in\mathcal{D}}}$ denotes the projection on the space defined by the rates $(R_d)_{d\in\mathcal{D}}$.
\end{thm}
\begin{proof}
See Section \ref{appth1ga}.
\end{proof}

\begin{thm}[Outer bound]  \label{th2ga}
We have $ \mathcal{C}(\mathbb{A}) \subseteq \mathcal{R}^{(\textup{out})}(\mathbb{A})$, where
\begin{align*}
&\mathcal{R}^{(\textup{out})}(\mathbb{A}) \triangleq\\
& \left\{(R_d)_{d\in\mathcal{D}} : 
R_{\mathcal{S}} \leq \displaystyle\min_{\mathcal{A} \in \mathbb{A}} \min_{\mathcal{U} \in \mathbb{U}} I(Y_{\mathcal{S}};X_{\mathcal{\mathcal{A}}} Y_{\mathcal{S}^c}| X_{\mathcal{U}}) , \forall \mathcal{S} \subseteq \mathcal{D}    \right\}.
\end{align*} 
\end{thm}
\begin{proof}
See Section \ref{secconvproof2}.
\end{proof}
 
 Note that it is challenging to simplify the inner bound $\mathcal{R}^{(\textup{in})}(\mathbb{A})$ in  Theorem \ref{th1ga} because the set functions  $\mathcal{S} \mapsto \displaystyle\max_{\mathcal{A} \in \mathbb{A}} H(Y_{\mathcal{S}}|Y_{\mathcal{S}^c}X_{\mathcal{A}})$ and $\mathcal{S} \mapsto \displaystyle\min_{\mathcal{U} \in \mathbb{U}} H(Y_{\mathcal{S}}|X_{\mathcal{U}})$  are not necessarily submodular or supermodular and, consequently, Fourier-Motzkin elimination is not easily applicable for a large number of sub-dealers~$D$. As described next, one can, however, obtain simplified bounds when $D=1$ and $D=2$, and a capacity result for threshold access structures when the source of randomness corresponds to pairwise secret keys.

\subsubsection{Results for a two-sub-dealer setting, i.e., $D=2$}

\begin{cor}[Inner bound] \label{cor1}
Assume that $D=2$. We have $\mathcal{R}^{(\textup{in})}(\mathbb{A}) \subseteq \mathcal{C}(\mathbb{A})$, where $\mathcal{R}^{(\textup{in})}(\mathbb{A})$ is defined in \eqref{eqcor1}.
\end{cor}
\begin{figure*}[t!]
\begin{align*}
\mathcal{R}^{(\textup{in})}(\mathbb{A}) \triangleq  \left\{\!(R_1, R_2) :\! 
\begin{array}{rl}
R_{1} \!\!&\leq \displaystyle\min_{\mathcal{A} \in \mathbb{A}} \min_{\mathcal{U} \in \mathbb{U}} \left( I(Y_{1};Y_2X_{\mathcal{A}}) - I(Y_{1};X_{\mathcal{U}}) \right)   \\
R_{2} \!\!&\leq  \displaystyle\min_{\mathcal{A} \in \mathbb{A}} \min_{\mathcal{U} \in \mathbb{U}} \left( I(Y_{2};Y_1 X_{\mathcal{A}}) - I(Y_{2};X_{\mathcal{U}}) \right) \\
R_1 + R_2 \!\!&\leq  \min \!\!\left(\!\!\!\!\!\!\! \begin{array}{rl}
& \displaystyle\min_{\mathcal{A} \in \mathbb{A}}    I(Y_{\mathcal{D}};X_{\mathcal{A}}) - \max_{\mathcal{U} \in \mathbb{U}} I(Y_{\mathcal{D}};X_{\mathcal{U}})  ,\\
&  \displaystyle\min_{\mathcal{A} \in \mathbb{A}}    I(Y_{1};Y_2 X_{\mathcal{A}}) + \displaystyle\min_{\mathcal{A} \in \mathbb{A}}    I(Y_{2}; X_{\mathcal{A}}|Y_1) - \displaystyle\max_{\mathcal{U} \in \mathbb{U}}  I(Y_{\mathcal{D}};X_{\mathcal{U}}),\\
&  \displaystyle\min_{\mathcal{A} \in \mathbb{A}}    I(Y_{\mathcal{D}}; X_{\mathcal{A}}) \!-\!  \displaystyle\max_{\mathcal{U} \in \mathbb{U}}  I(Y_{1};X_{\mathcal{U}}) \!  -\!  \displaystyle\max_{\mathcal{U} \in \mathbb{U}}  I(Y_{2};X_{\mathcal{U}}) + I(Y_1;Y_2) \end{array} \!\! \!\right)  \end{array}     
\!\!\!\!\!\right\} \numberthis \label{eqcor1}
\end{align*} 
\hrulefill
\end{figure*}
\begin{cor}[Outer bound]  \label{cor2}
Assume that $D=2$. We have $ \mathcal{R}^{(\textup{out})}(\mathbb{A}) \supseteq \mathcal{C}(\mathbb{A}) $, where
\begin{align*}
&\mathcal{R}^{(\textup{out})}(\mathbb{A})\\& \triangleq \left\{(R_1, R_2) : \!\!
\begin{array}{rl}
R_{1} &\leq \displaystyle\min_{\mathcal{A} \in \mathbb{A}} \min_{\mathcal{U} \in \mathbb{U}} I(Y_1;X_{\mathcal{\mathcal{A}}} Y_{2}| X_{\mathcal{U}}) \\
R_{2} &\leq \displaystyle\min_{\mathcal{A} \in \mathbb{A}} \min_{\mathcal{U} \in \mathbb{U}} I(Y_{2};X_{\mathcal{\mathcal{A}}} Y_{1}| X_{\mathcal{U}}) \\
R_1 + R_2 &\leq \displaystyle\min_{\mathcal{A} \in \mathbb{A}} \min_{\mathcal{U} \in \mathbb{U}} I(Y_{\mathcal{D}};X_{\mathcal{\mathcal{A}}} | X_{\mathcal{U}}) 
\end{array}    \!\! \right\}.
\end{align*} 
\end{cor}

 Corollary \ref{cor1} is obtained from Theorem \ref{th1ga} by using Fourier-Motzkin elimination. Corollary \ref{cor2} is a consequence of Theorem~\ref{th2ga}.
\subsubsection{Results for a single-dealer setting, i.e., $D=1$}

\begin{cor} \label{cor3}
Assume that $D=1$. We have the following lower and upper bounds for the secret capacity ${C}(\mathbb{A})$
\begin{multline*}
\displaystyle\min_{\mathcal{A} \in \mathbb{A}} \min_{\mathcal{U} \in \mathbb{U}} \left( I(Y_1;X_{\mathcal{A}}) - I(Y_1;X_{\mathcal{U}}) \right) \\ \leq {C}(\mathbb{A}) \\  \leq \displaystyle\min_{\mathcal{A} \in \mathbb{A}} \min_{\mathcal{U} \in \mathbb{U}} I(Y_1;X_{\mathcal{\mathcal{A}}} | X_{\mathcal{U}}).
\end{multline*}
\end{cor}
Corollary \ref{cor3} is a  consequence of Theorem~\ref{th1ga} and Theorem~\ref{th2ga}.
\subsection{All-or-nothing access structure} \label{sec:resAON}
In this section, we consider the all-or-nothing access structure denoted by $\mathbb{A}^{\star} \triangleq \{ \mathcal{L} \}$. This setting corresponds to the case where all the participants are needed to reconstruct the secret.

\subsubsection{Results for an arbitrary number $D$ of sub-dealers}

The achievability proof technique for Theorem \ref{prop1} is different than the proof technique for a general access structure in Theorem \ref{th1ga}. Specifically, we successively, instead of simultaneously, handle the reliability constraint \eqref{eqrel} and the security constraint~\eqref{eqSeca}. This strategy is, for instance, used for secret-key generation~\cite{Bennett95,cachin1997linking,Maurer00}. However, in our \emph{distributed setting}, the application of this strategy is not straightforward and we discuss in the proof the main technical challenges that needs to be overcome to obtain this extension. The first step of our coding strategy, to handle the reliability constraint,  involves  a careful design of an exponential number (with respect to $D$) of nested binnings. The second step of our coding strategy, to handle the security constraints, involves a new variant of the distributed leftover hash lemma (Lemma \ref{lemloh} in Appendix \ref{sec:analy}).  Note that the proof technique used to prove Theorem~\ref{prop1} has at least two  advantages compared to a joint random binning approach as in Theorem~\ref{th1ga}. First, no auxiliary rate appears in the achievability region of Theorem~\ref{prop1}, second, it provides insight for the design of explicit secret-sharing schemes by showing that a two-layer design approach that separates the reliability constraint from the security constraints can be used.

\begin{thm}[Inner bound] \label{prop1}
We have $\mathcal{R}_{1}^{(\textup{in})}\subseteq \mathcal{C}(\mathbb{A}^{\star})$, with 
\begin{align*} 
\mathcal{R}_{1}^{(\textup{in})}  \triangleq \!\! \left\{ \! (R_d)_{d \in \mathcal{D}} : R_{\mathcal{S }} \leq   \min_{ \mathcal{T} \subsetneq \mathcal{L} }  I(Y_{\mathcal{S}};   X_{\mathcal{L}}  |X_{\mathcal{T}}) 
, \forall \mathcal{S} \subseteq \mathcal{D} \right\}.
\end{align*}
\end{thm}
\begin{proof}
See Section \ref{sec:cs}.
\end{proof}

\begin{thm}[Outer bound] \label{th3}
We have $\mathcal{R}^{(\textup{out})}(\mathbb{A}^{\star}) \supseteq \mathcal{C}(\mathbb{A}^{\star})$, where 
\begin{align*}
&\mathcal{R}^{(\textup{out})}(\mathbb{A}^{\star})\\
&\triangleq \! \left\{ \! (R_d)_{d \in \mathcal{D}} : R_{\mathcal{S }} \leq   \min_{ \mathcal{T} \subsetneq \mathcal{L} }  I(Y_{\mathcal{S}};   X_{\mathcal{L}} Y_{\mathcal{S}^c} |X_{\mathcal{T}}) 
, \forall \mathcal{S} \subseteq \mathcal{D} \right\}.
\end{align*}
\end{thm}
\begin{proof}
See Section \ref{secconvproof2b}.
\end{proof}

\subsubsection{Results for a two-sub-dealer setting, i.e., $D=2$}

The achievability proof strategy of Theorem \ref{th2} is different than the achievability proof strategy of Theorem \ref{prop1}. Note that in the proof of Theorem \ref{prop1}, we deal with the security constraint \eqref{eqSeca} by jointly considering all the sub-dealers. By contrast,  our achievability proof strategy in Theorem \ref{th2} considers the sub-dealers individually when ensuring \eqref{eqSeca}. Specifically, when $D=2$, one can first realize a secret-sharing scheme between Sub-dealer 1 and the participants  with the requirement $\lim_{n \to \infty} \max_{ \mathcal{U} \in \mathbb{U}}   I\left({S}_{1} ; M_{1}, X^n_{\mathcal{U}}  \right) = 0$, and then realize a secret-sharing scheme between Sub-dealer~2 and the participants with the requirement $\lim_{n \to \infty} \max_{ \mathcal{U} \in \mathbb{U}}  I\left({S}_{2} ; M_{2}, X^n_{\mathcal{U}}, Y_1^n  \right) = 0$, as illustrated in Figure \ref{fig:sj}. As described next, one can show that such an approach is sufficient to ensure the security constraint~\eqref{eqSeca}. However,  the proof is not trivial as we need to modify the reconciliation protocol of Theorem \ref{prop1} described in Section~\ref{sec:cs}, and as an initialization phase is also required, during which  Sub-dealer 2 shares a secret with negligible rate with all the participants. Note also that one could exchange the role of the two sub-dealers in the  protocol to potentially enlarge the achievablity region via this method. This idea leads to Theorem \ref{th2}.

\begin{figure}[t!]
    \centering
    \begin{subfigure}[t]{0.5 \textwidth}
        \centering
    \includegraphics[width=5.36cm]{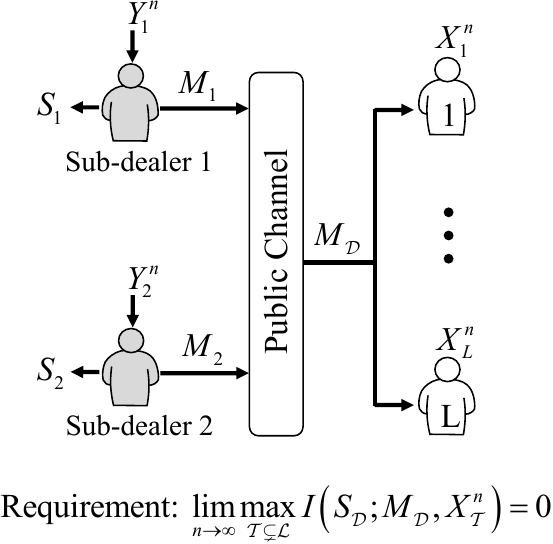}
        \caption{Joint security design strategy}
    \end{subfigure}
\vspace{1em}
    \begin{subfigure}[t]{0.5 \textwidth}
        \centering
    \includegraphics[width=5.36cm]{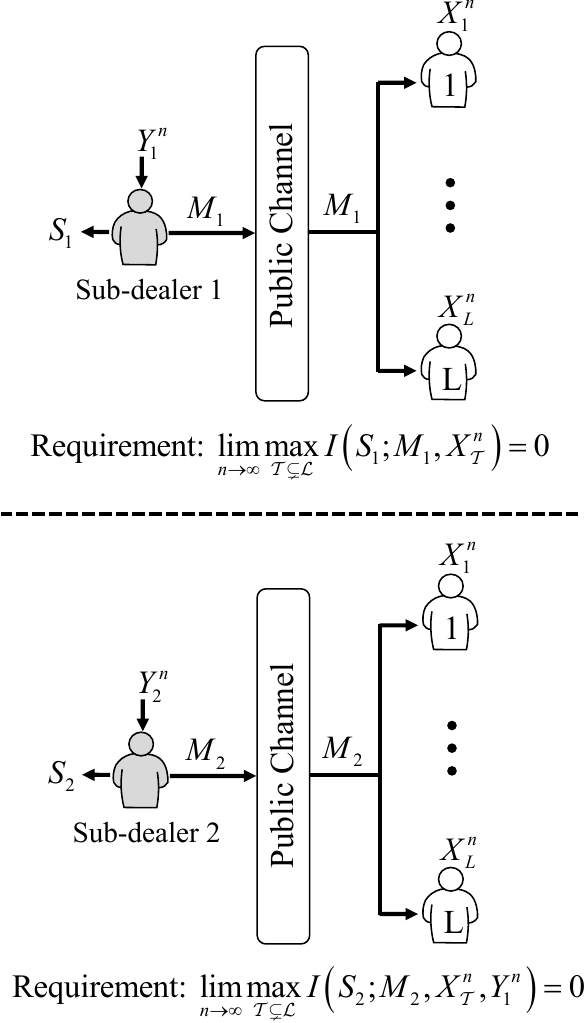}
        \caption{Successive security design strategy}
    \end{subfigure}
\caption{A joint security design strategy for $(S_1,S_2)$ is used in Theorem \ref{prop1}, whereas a successive security design strategy for $(S_1,S_2)$ is used in Theorem \ref{th2}.}\label{fig:sj}
\end{figure}

\begin{thm}[Inner bound] \label{th2}
Assume that $D=2$. If $\displaystyle\min_{d \in \{1,2\}}\displaystyle\min_{ \mathcal{T} \subsetneq \mathcal{L}} I(Y_{d};X_{\mathcal{L}}|X_{\mathcal{T}}) >0$, then $\mathcal{R}_{1}^{(\textup{in})} \subseteq \mathcal{R}_{2}^{(\textup{in})} \subseteq \mathcal{C}(\mathbb{A}^{\star})$, with
\begin{align*}  
\mathcal{R}_{2}^{(\textup{in})} &\triangleq  \left[  \mathcal{R}(\{1\})\times \mathcal{R}(\{2\}|\{1\}) \right]\\
&\phantom{--}\cup \left[  \mathcal{R}(\{1\}|\{2\}) \times \mathcal{R}(\{2\})\right] \cup \mathcal{R}(\{1,2\}),
\end{align*}
where we have defined for any $\mathcal{S}, \mathcal{V} \subseteq \mathcal{D}$,
\begin{align*}
&\mathcal{R}(\mathcal{S}|\mathcal{V})\\
& \triangleq \! \left\{ (R_d)_{d \in \mathcal{S}} \!  : \! R_{\mathcal{B}}  \leq \!  \min_{ \mathcal{T} \subsetneq \mathcal{L} }  I(Y_{\mathcal{B}};     X_{\mathcal{L}}| Y_{\mathcal{V}} X_{\mathcal{T}}) 
  ,\forall \mathcal{B} \subseteq \mathcal{S} \right\} ,
\end{align*}
and $\mathcal{R}(\mathcal{S}) \triangleq \mathcal{R}(\mathcal{S}|\emptyset)$.
\end{thm}
\begin{proof}
See Section~\ref{App:th2}. 
\end{proof}

From Theorem \ref{th2}, we deduce the following sum-rate achievability result.

\begin{cor}[Sum-rate achievability] \label{ex1}
Assume that $D=2$ and $\displaystyle\min_{d \in \{1,2\}}\displaystyle\min_{ \mathcal{T} \subsetneq \mathcal{L}} I(Y_{d};X_{\mathcal{L}}|X_{\mathcal{T}}) >0$. Define for any $\mathcal{S}, \mathcal{V} \subseteq \mathcal{D}$, $$R(\mathcal{S} | \mathcal{V}) \triangleq \min_{ \mathcal{T} \subsetneq \mathcal{L} }  I(Y_{\mathcal{S}};   X_{\mathcal{L}}| Y_{\mathcal{V}} X_{\mathcal{T}}) .$$
For convenience, we also define for $\mathcal{S} \subseteq \mathcal{D}$, $R(\mathcal{S}) \triangleq R(\mathcal{S}|\emptyset)$.  Theorem~\ref{prop1} shows the achievability of the secret sum-rate $R^{\textup{sum}}_1$, while Theorem~\ref{th2} shows the achievability of the secret sum-rate $\max (R^{\textup{sum}}_{1},  R^{\textup{sum}}_{2}, R^{\textup{sum}}_{3})$, where
\begin{align*}
R^{\textup{sum}}_{1} & \triangleq	\min \left(  R(\{1,2\})  ; R(\{1\}) +R(\{2\})   \right), \\
R^{\textup{sum}}_{2} & \triangleq	 [ R(\{1\})  +  R(\{2\} | \{1\}) ] , \\
R^{\textup{sum}}_{3} & \triangleq	 [R(\{2\})  +  R(\{1\} | \{ 2 \})].
\end{align*}
\end{cor}
From Theorem \ref{th3}, we will also have the following outer bound.
\begin{cor}[Outer bound]  \label{cor5}
Assume that $D=2$. We have $ \mathcal{R}^{(\textup{out})}(\mathbb{A}^{\star}) \supseteq \mathcal{C}(\mathbb{A}^{\star}) $, where
\begin{align*}
&\mathcal{R}^{(\textup{out})}(\mathbb{A}^{\star})\\
& \triangleq \left\{(R_1, R_2) : 
\begin{array}{rl}
R_{1} &\leq \displaystyle  \min_{\mathcal{T} \subsetneq \mathcal{L}} I(Y_1;X_{\mathcal{\mathcal{L}}} Y_{2}| X_{\mathcal{T}}) \\
R_{2} &\leq \displaystyle  \min_{\mathcal{T} \subsetneq \mathcal{L}} I(Y_{2};X_{\mathcal{\mathcal{L}}} Y_{1}| X_{\mathcal{T}}) \\
R_1 + R_2 &\leq \displaystyle  \min_{\mathcal{T} \subsetneq \mathcal{L}} I(Y_{\mathcal{D}};X_{\mathcal{\mathcal{L}}} | X_{\mathcal{T}}) 
\end{array}     \right\}.
\end{align*} 
\end{cor}

Next, we provide a sufficient condition for having found the optimal secret sum-rate in Corollary~\ref{ex1}.
\begin{cor}
We use the same notation as in Corollary \ref{cor5}. If $R(\{1,2\}) \leq R(\{1\}) +R(\{2\})$, then the secret sum-rate $R^{\textup{sum}}_{1}$ in Corollary \ref{ex1} is optimal by Corollary \ref{cor5}.
\end{cor}

\subsubsection{Result for a single-dealer setting, i.e., $D=1$}

In the presence of a single dealer, i.e., when $D=1$, we have the following capacity result.

\begin{thm} \label{cor4}
Assume that $D=1$. The secret capacity ${C}(\mathbb{A}^{\star})$ is given by 
$$
 {C}(\mathbb{A}^{\star}) = \displaystyle\min_{ \mathcal{T} \subsetneq \mathcal{L} }I(Y_{\mathcal{D}};X_{\mathcal{L}}|X_{\mathcal{T}}).
$$ 
\end{thm}

\begin{proof}
See Section \ref{appth6}.
\end{proof}

Theorem \ref{cor4} can be seen as a counterpart to the result  for a channel model in \cite{zou2015information}.
 
 \begin{ex}
Suppose that $D=1$ and $L=2$. Then, by Theorem \ref{cor4}, we have 
\begin{align*}
C(\mathbb{A}^{\star})
& = \min[ I(Y_{\mathcal{D}};X_{1}|X_2),I(Y_{\mathcal{D}};X_2|X_1) ].
\end{align*}
 \end{ex}

 \begin{ex}
Suppose that $D=1$ and consider $L$ identical and independent channels $C_l = (\mathcal{Y},p_{X|Y},\mathcal{X})$ with  $\mathcal{X}$ and $\mathcal{Y}$ two finite alphabets. Suppose  that, for any $l\in\mathcal{L}$, $\mathcal{X}_l = \mathcal{X}$, and $X_l$ is the output of the channel $C_l$ when $Y_{\mathcal{D}}$, distributed according to $p_Y$, is the input. Then, by Theorem \ref{cor4}, we have 
\begin{align*}
C(\mathbb{A}^{\star}) = I(Y_{\mathcal{D}};X_{1}|X_{\llbracket 2, L \rrbracket}).
 \end{align*}
 \end{ex}

\subsection{Threshold access structures when the source of randomness corresponds to pairwise secret keys}
We define threshold access structures as follows. Let $t \in \llbracket 1 , L \rrbracket$ and $z \in \llbracket 1, t-1 \rrbracket$. Define the access structure
$$\mathbb{A}_t \triangleq \{ \mathcal{S} \subseteq \mathcal{L} : |\mathcal{S}| \geq t \},$$
 and the set of non-authorized participants as  
 $$\mathbb{U}_z \triangleq \{ \mathcal{S} \subseteq \mathcal{L} : |\mathcal{S}| \leq z \},$$
 and consider Definition \ref{def} with the substitution $\mathbb{A} \leftarrow \mathbb{A}_t$ and $\mathbb{U} \leftarrow \mathbb{U}_t$. We denote the capacity region by $\mathcal{C}(\mathbb{A}_t,\mathbb{U}_z)$ instead of $\mathcal{C}(\mathbb{A})$, and the secret capacity by $C(t,z)$ instead of ${C}(\mathbb{A})$ when $D=1$. This setting means that any set of participants of size larger than or equal to~$t$ must be able to recover the secrets, and any set of participants of size smaller than or equal to $z$ must be unable to learn any information about the secrets.
 
Clearly, for arbitrarily correlated source of randomness, the results of Section \ref{sec:resGEN} apply for any $t \in \llbracket 1 , L \rrbracket$ and $z \in \llbracket 1, t-1 \rrbracket$, and the results of Section \ref{sec:resAON} apply for $(t,z)=(L,L-1)$. 
 We then have the following capacity result when the source of randomness corresponds to pairwise secret keys.
\begin{thm}[Capacity region] \label{thm_thr}
Suppose that Participant $l\in \mathcal{L}$ and Sub-dealer $d\in\mathcal{D}$ share a secret key $K_{l,d}^n$ uniformly distributed over $\{0,1\}^n$, and that all the keys are jointly independent. With the notation of Section \ref{sec:pre}, we thus have $X^n_l = (K_{l,d}^n)_{d \in \mathcal{D}} $ for User $l\in \mathcal{L}$ and  $Y^n_d = (K_{l,d}^n)_{l \in \mathcal{L}}$ for Sub-dealer $d\in \mathcal{D}$.  Let $t \in \llbracket 1 , L \rrbracket$ and $z \in \llbracket 1, t-1 \rrbracket$. Then, we have 
\begin{align*}
\mathcal{C}(\mathbb{A}_t,\mathbb{U}_z) &= \left\{(R_d)_{d\in\mathcal{D}} : 
R_{\mathcal{S}} \leq  |\mathcal{S}| (t-z), \forall \mathcal{S} \subseteq \mathcal{D} \right\}\\&= \left\{(R_d)_{d\in\mathcal{D}} : 
R_d \leq   t-z, \forall d\in\mathcal{D} \right\},
\end{align*}
moreover, the rate-tuple $(R^{\star}_d)_{d\in\mathcal{D}}$ is achievable with $R^{\star}_d \triangleq  t-z$.
\end{thm}
\begin{proof}
See Section~\ref{App_th7}. 
\end{proof}

Note that Theorem \ref{thm_thr} is consistent with known results for Shamir's secret sharing model. Indeed, suppose that $D=1$ and $z=t-1$ in Theorem \ref{thm_thr}. Using Shamir's secret sharing, the dealer can first form $L$ shares of $n$ bits for a secret $S$ with entropy $H(S)=n$, and then secretly transmit each share to a participants via a one-time pad over the public channel by using the secret keys of length $n$. The dealer has thus shared a secret with rate 
$\frac{n}{n} =1$. Now, since $C(t,z) = t-z = 1$ by Theorem \ref{thm_thr}, we also conclude in this example that there is no loss of optimality in \emph{independently} handling the share generation phase and the secure share distribution~phase. 

\begin{ex}
Suppose that $D=1$ and $L=10$. Then, $C(t,z) = t-z$ is depicted in Figure \ref{figext}.
\begin{figure}[t!]
        \centering
    \includegraphics[width=6.8cm]{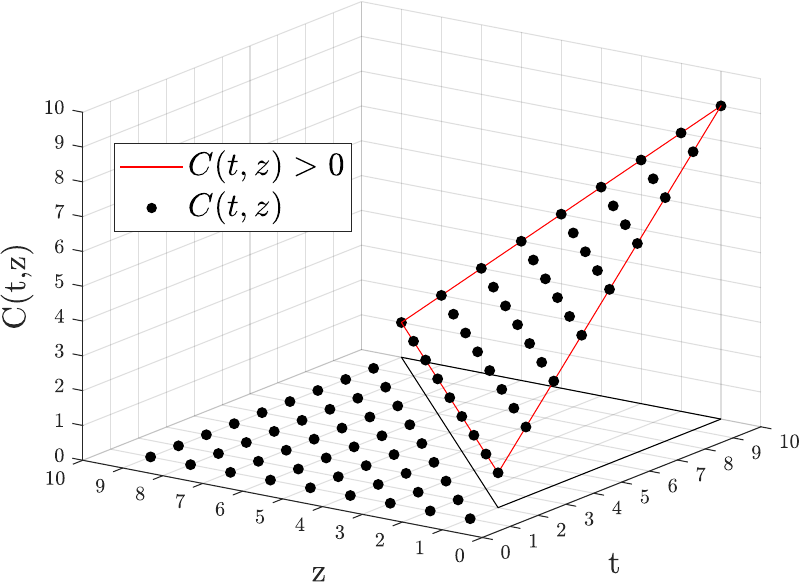}
        \caption{Secret capacity for threshold access structures when $D=1$ and $L=10$.}\label{figext}
\end{figure}

\end{ex}

\section{Achievability proofs} \label{sec:achievability}

Sections \ref{appth1ga}, \ref{sec:cs}, \ref{App:th2} contain the achievability proofs of Theorems \ref{th1ga}, \ref{prop1}, and  \ref{th2}, respectively. In the following, we will use the following notation. For a pair of discrete random variables $(X,Y)$ distributed according to $p_{XY}$ over a finite alphabet $\mathcal{X} \times \mathcal{Y}$, let $\mathcal{T}^n_{\epsilon}(X) \triangleq \left\{ x^n \in \mathcal{X}^n \!:\! \left|\tfrac{ \sum_{i=1}^n \mathds{1}\{x_i\! =\!x\}}{n}-p_{X}(x) \right| \leq \epsilon p_{X}(x), \forall x \in \mathcal{X} \! \right\}$ denote the $\epsilon$-letter-typical set associated with $p_X$ for sequences of length $n$, e.g.,~\cite{Orlitsky01}, and define $\mu_X \triangleq \min_{x \in \mathcal{X} \text{ s.t. } p_X(x)>0 } p_X(x)$. Let also $\mathcal{T}^n_{\epsilon}(XY|x^n) \triangleq \left\{ y^n \in \mathcal{Y}^n : (x^n,y^n) \in \mathcal{T}_{\epsilon}^n(XY) \right\}$ be the conditional $\epsilon$-letter-typical set associated with $p_{XY}$ with respect to $x^n \in \mathcal{X}^n$. 

\subsection{Proof of Theorem \ref{th1ga}} \label{appth1ga}

Theorem \ref{th1ga} relies on random binning. The coding scheme and its analysis are described in Sections \ref{csth1} and \ref{csath1}, respectively. 

\subsubsection{Coding scheme} \label{csth1}

\emph{Binnings}: Fix $i\in\mathcal{D}$. 
Define the  functions $g_{i} : \mathcal{Y}_i^n \to \llbracket 1,2^{nR'_{i}} \rrbracket$ and $h_{i} : \mathcal{Y}_i^n \to \llbracket 1,2^{nR_{i}} \rrbracket$,  where, for $y_i^n \in \mathcal{Y}_i^n$, $g_{i}(y_i^n)$ is drawn uniformly at random in the set $\llbracket 1 , 2^{nR_i'} \rrbracket$, and  $h_{i}(y_i^n)$ is drawn uniformly at random in the set $\llbracket 1 , 2^{nR_i} \rrbracket$. 

Then, the encoding at the sub-dealers and the decoding at the participants are  as follows:

\emph{Encoding at Sub-dealer ${i} \in\mathcal{D}$}: Given $y_i^{n}$, Sub-dealer ${i} \in\mathcal{D}$ computes  $m_{i} \triangleq   g_{i}(y_i^n)$ and $s_{i} \triangleq   h_{i}(y_i^n)$.

\emph{Decoding for a set of participants $\mathbf{\mathcal{A} \in \mathbb{A}}$}: Given $m_{\mathcal{D}} \triangleq (m_{d})_{d\in \mathcal{D}}$ and  $x_{\mathcal{A}}^n$, the set of participants $\mathcal{A}$ returns $\hat{y}_{\mathcal{D}}^n (\mathcal{A}) = (\hat{y}_{i}^n)_{ i\in \mathcal{D}}$ if it is the unique sequence such that  $(\hat{y}_{\mathcal{D}}^n(\mathcal{A}),  x_{\mathcal{A}}^n) \in \mathcal{T}_{\epsilon}^n({Y}_{\mathcal{D}} {X}_{\mathcal{A}})$ and $  \left(g_{i}(\hat{y}_{i}^n) \right)_{i\in \mathcal{D}} = m_{\mathcal{D}}$, otherwise it returns an error.

Next, we determine how to choose $R_i$ and $R_i'$, $i\in \mathcal{D}$, to ensure the reliability, security, and uniformity conditions as described in Definition \ref{def}.

\subsubsection{Coding scheme analysis} \label{csath1}

\paragraph{Reliability analysis}
Fix $\mathcal{A} \in \mathbb{A}$. Define for any $\mathcal{S} \subseteq \mathcal{D}$, $\mathcal{S} \neq  \emptyset $,
\begin{align*}
\mathcal{E}_0 & \triangleq \{ (X_{\mathcal{A}}^n,Y_{\mathcal{D}}^n) \notin \mathcal{T}_{\epsilon }^n(X_{\mathcal{A}}  Y_{\mathcal{D}}) \},\\
\mathcal{E}_{\mathcal{S}} & \triangleq \left\{ \forall i \in \mathcal{S}, \exists \hat{y}_{i}^n \neq Y_{i}^n,  g_i (\hat{y}_i^n) = g_i (Y_i^n)  \right.\\ & \phantom{-----}\left.
 \text{ and } ( X^n_{\mathcal{A}},\hat{y}_{\mathcal{S}}^n,Y_{\mathcal{D} \backslash \mathcal{S}}^n) \in \mathcal{T}_{\epsilon }^n (X_{\mathcal{A}}  Y_{\mathcal{D}}) \right\},
\end{align*}
so that by the union bound, \begin{align}\mathbb{E}  [ \mathbb{P} [ \hat{Y}_{\mathcal D}^n(\mathcal{A}) \neq {Y}_{\mathcal D}^n  ] ] \leq \mathbb{P} [\mathcal{E}_0 ] + \sum_{\mathcal{S} \subseteq \mathcal{D},\mathcal{S} \neq  \emptyset  }\mathbb{P} [\mathcal{E}_{\mathcal{S}}],\label{eqcemi}
\end{align}
where the expectation is over the random choice of the binnings. 

\begin{lem} \label{lemappth1a}

For any $\mathcal{S} \subseteq \mathcal{D}$, $\mathcal{S} \neq  \emptyset   $, we have
\begin{align}
\mathbb{P} [\mathcal{E}_{\mathcal{S}} ] &   \leq 2^{  n(1 + \epsilon ) \max_{\mathcal{A} \in \mathbb{A}} {H}(Y_{\mathcal{S}}|Y_{ \mathcal{S}^c } X_{\mathcal{A}}) -nR'_{\mathcal{S}} }, \label{2demi2} \\
\mathbb{P} [ \mathcal{E}_0] & \leq   2 |\mathcal{X}_{\mathcal{L}}| |\mathcal{Y}_{\mathcal{D}}| e^{-n \epsilon^2 \mu_{X_{\mathcal{L}}Y_{\mathcal{D}}}}. \label{eqdemi}
\end{align}
\end{lem}

\begin{proof}
See Appendix \ref{appth1a}.
\end{proof}

Hence, by \eqref{eqcemi}, \eqref{2demi2}, and \eqref{eqdemi}, we have 
\begin{align*}
& \mathbb{E}  \left[ \max_{\mathcal{A} \in \mathbb{A}} \mathbb{P} [ \hat{Y}_{\mathcal D}^n(\mathcal{A}) \neq {Y}_{\mathcal D}^n  ] \right] \\
& \leq \mathbb{E}  \left[ \sum_{\mathcal{A} \in \mathbb{A}} \mathbb{P} [ \hat{Y}_{\mathcal D}^n(\mathcal{A}) \neq {Y}_{\mathcal D}^n  ] \right] 
\\
& = \sum_{\mathcal{A} \in \mathbb{A}} \mathbb{E}  \left[  \mathbb{P} [ \hat{Y}_{\mathcal D}^n(\mathcal{A}) \neq {Y}_{\mathcal D}^n  ] \right] 
\\
& \leq 2 | \mathbb{A}| |\mathcal{X}_{\mathcal{L}}| |\mathcal{Y}_{\mathcal{D}}| e^{-n \epsilon^2 \mu_{X_{\mathcal{L}}Y_{\mathcal{D}}}} \\
&\phantom{-} + | \mathbb{A}| \sum_{\mathcal{S} \subseteq \mathcal{D},\mathcal{S} \neq   \emptyset   }\!\!\! 2^{  n(1 + \epsilon ) \max_{\mathcal{A} \in \mathbb{A}} {H}(Y_{\mathcal{S}}|Y_{ \mathcal{S}^c } X_{\mathcal{A}}) -nR'_{\mathcal{S}} }. \numberthis \label{ratechoicere}
\end{align*}

\paragraph{Security and uniformity analysis}
Fix $\mathcal{U} \in \mathbb{U}$. For all $m_{\mathcal{D}}$, $s_{\mathcal{D}}$, $x_{\mathcal{U}}^n$, we have
\begin{align*} 
&p_{M_{\mathcal{D}} S_{\mathcal{D}} X_{\mathcal{U}}^n }(m_{\mathcal{D}},s_{\mathcal{D}},x_{\mathcal{U}}^n) \\
&\phantom{-}= \sum_{y_{\mathcal D}^n} p(y_{\mathcal{D}}^n,x_{\mathcal{U}}^n) \prod_{i\in \mathcal{D}} \mathds{1} \{g_i(y_i^n) = m_{i} \}\mathds{1} \{h_i(y_i^n) = s_{i} \}.
\end{align*}
Hence, on average over the random choice of the binnings, for all $m_{\mathcal{D}}$, $s_{\mathcal{D}}$, $x_{\mathcal{U}}^n$, we have
$$
 \mathbb{E} \left[p_{M_{\mathcal{D}} S_{\mathcal{D}} X_{\mathcal{U}}^n }(m_{\mathcal{D}},s_{\mathcal{D}},x_{\mathcal{U}}^n) \right] = p(x_{\mathcal{U}}^n) 2^{-n(R_{\mathcal{D}}+R'_{\mathcal{D}})}, 
$$
which allows us to write
\begin{align}
& \mathbb{E} [ \mathbb{V} ( p_{M_{\mathcal{D}} S_{\mathcal{D}} X_{\mathcal{U}}^n}, p^{\textup{(unif)}}_{M_{\mathcal{D}} S_{\mathcal{D}}} p_{ X_{\mathcal{U}}^n}) ] \nonumber  \\
& = \mathbb{E} \left[ \textstyle\sum_{m_{\mathcal{D}},s_{\mathcal{D}},x_{\mathcal{U}}^n} \left| p_{M_{\mathcal{D}} S_{\mathcal{D}} X_{\mathcal{U}}^n }(m_{\mathcal{D}},s_{\mathcal{D}},x_{\mathcal{U}}^n)\right.\right. \nonumber \\ & \phantom{----------}\left.\left. -  \mathbb{E} \left[ p_{M_{\mathcal{D}} S_{\mathcal{D}} X_{\mathcal{U}}^n }(m_{\mathcal{D}},s_{\mathcal{D}},x_{\mathcal{U}}^n) \right] \right| \right] \nonumber \\
& \leq  \sum_{k=1}^2 \mathbb{E} \left[ \textstyle\sum_{m_{\mathcal{D}},s_{\mathcal{D}},x_{\mathcal{U}}^n} \left| p^{(k)}_{M_{\mathcal{D}} S_{\mathcal{D}} X_{\mathcal{U}}^n }(m_{\mathcal{D}},s_{\mathcal{D}},x_{\mathcal{U}}^n) \right.\right. \nonumber \\ & \phantom{----------}\left.\left.-  \mathbb{E} \left[ p^{(k)}_{M_{\mathcal{D}} S_{\mathcal{D}} X_{\mathcal{U}}^n }(m_{\mathcal{D}},s_{\mathcal{D}},x_{\mathcal{U}}^n) \right] \right| \right] \label{twotermsa}, 
\end{align}
where $p^{\textup{(unif)}}_{M_{\mathcal{D}} S_{\mathcal{D}}}$ is the uniform distribution over the sample space of $p_{M_{\mathcal{D}} S_{\mathcal{D}}}$, and $\forall m_{\mathcal{D}}, \forall s_{\mathcal{D}}, \forall x_{\mathcal{U}}^n$,
\begin{align*}
& p^{(1)}_{M_{\mathcal{D}} S_{\mathcal{D}} X_{\mathcal{U}}^n }(m_{\mathcal{D}},s_{\mathcal{D}},x_{\mathcal{U}}^n) \\
&=  \sum_{y_{\mathcal D}^n \in \mathcal{T}_{\epsilon}^n({Y}_{\mathcal{D}}  {X}_{\mathcal{U}}|x_{\mathcal{U}}^n)} p(y_{\mathcal{D}}^n,x_{\mathcal{U}}^n) \prod_{i\in \mathcal{D}} \mathds{1} \{g_i(y_i^n) = m_{i} \}\\ &\phantom{----------------} \times\mathds{1} \{h_i(y_i^n) = s_{i} \},\\
& p^{(2)}_{M_{\mathcal{D}} S_{\mathcal{D}} X_{\mathcal{U}}^n }(m_{\mathcal{D}},s_{\mathcal{D}},x_{\mathcal{U}}^n) \\
&=  \sum_{y_{\mathcal D}^n \notin \mathcal{T}_{\epsilon}^n({Y}_{\mathcal{D}}  {X}_{\mathcal{U}}|x_{\mathcal{U}}^n)} p(y_{\mathcal{D}}^n,x_{\mathcal{U}}^n) \prod_{i\in \mathcal{D}} \mathds{1} \{g_i(y_i^n) = m_{i} \} \\ &\phantom{----------------} \times \mathds{1} \{h_i(y_i^n) = s_{i} \}.
\end{align*}
\begin{lem} \label{lemappth1b}
We have
\begin{align*}
&\mathbb{E} \left[ \textstyle\sum_{m_{\mathcal{D}},s_{\mathcal{D}},x_{\mathcal{U}}^n} \left| p^{(2)}_{M_{\mathcal{D}} S_{\mathcal{D}} X_{\mathcal{U}}^n }(m_{\mathcal{D}},s_{\mathcal{D}},x_{\mathcal{U}}^n) \right. \right. \\ & \phantom{-------}\left. \left. -  \mathbb{E} \left[ p^{(2)}_{M_{\mathcal{D}} S_{\mathcal{D}} X_{\mathcal{U}}^n }(m_{\mathcal{D}},s_{\mathcal{D}},x_{\mathcal{U}}^n) \right] \right| \right] \nonumber \\
&\leq 2  |\mathcal{Y}_{\mathcal{D}}| |\mathcal{X}_{\mathcal{L}}| e^{-n\epsilon^2 \mu_{Y_{\mathcal{D}} X_{\mathcal{L}}}} \numberthis \label{ademia},
\end{align*}
and
\begin{align*}
&\mathbb{E} \left[ \textstyle\sum_{m_{\mathcal{D}},s_{\mathcal{D}},x_{\mathcal{U}}^n} \left| p^{(1)}_{M_{\mathcal{D}} S_{\mathcal{D}} X_{\mathcal{U}}^n }(m_{\mathcal{D}},s_{\mathcal{D}},x_{\mathcal{U}}^n) \right. \right. \\ & \phantom{-------}\left. \left.-  \mathbb{E} \left[ p^{(1)}_{M_{\mathcal{D}} S_{\mathcal{D}} X_{\mathcal{U}}^n }(m_{\mathcal{D}},s_{\mathcal{D}},x_{\mathcal{U}}^n) \right] \right| \right] \nonumber \\
&\leq  \sum_{\mathcal{S} \subseteq \mathcal{D}, \mathcal{S}\neq \emptyset}     2^{-\tfrac{n}{2}(1-\epsilon)\min_{\mathcal{U} \in \mathbb{U}}H(Y_{\mathcal{S}}|X_{\mathcal{U}})}  2^{\tfrac{n}{2}(R_{\mathcal{S}}+R_{\mathcal{S}}')}. \numberthis \label{ademia2}
\end{align*}
\end{lem}

\begin{proof}
See Appendix \ref{appth1b}.
\end{proof}

Finally, by (\ref{twotermsa}), \eqref{ademia}, and \eqref{ademia2}, we obtain
\begin{align*}
&\mathbb{E} \left[ \max_{\mathcal{U} \in \mathbb{U}} \mathbb{V} ( p_{M_{\mathcal{D}} S_{\mathcal{D}} X_{\mathcal{U}}^n}, p^{\textup{(unif)}}_{M_{\mathcal{D}} S_{\mathcal{D}}} p_{ X_{\mathcal{U}}^n}) \right] \nonumber \\
& \leq \mathbb{E}  \left[ \sum_{\mathcal{U} \in \mathbb{U}} \mathbb{V} ( p_{M_{\mathcal{D}} S_{\mathcal{D}} X_{\mathcal{U}}^n}, p^{\textup{(unif)}}_{M_{\mathcal{D}} S_{\mathcal{D}}} p_{ X_{\mathcal{U}}^n}) \right] 
\\
& = \sum_{\mathcal{U} \in \mathbb{U}} \mathbb{E}  \left[  \mathbb{V} ( p_{M_{\mathcal{D}} S_{\mathcal{D}} X_{\mathcal{U}}^n}, p^{\textup{(unif)}}_{M_{\mathcal{D}} S_{\mathcal{D}}} p_{ X_{\mathcal{U}}^n})  ] \right] 
\\
& \leq 2 | \mathbb{U}| |\mathcal{Y}_{\mathcal{D}}| |\mathcal{X}_{\mathcal{L}}| e^{-n\epsilon^2 \mu_{Y_{\mathcal{D}} X_{\mathcal{L}}}} \\
&\phantom{-}+ | \mathbb{U}|\sum_{\mathcal{S} \subseteq \mathcal{D}, \mathcal{S}\neq \emptyset}     2^{\tfrac{n}{2}[R_{\mathcal{S}}+R_{\mathcal{S}}'-(1-\epsilon)\min_{\mathcal{U} \in \mathbb{U}}H(Y_{\mathcal{S}}|X_{\mathcal{U}})]}  . \numberthis \label{eqratechoicesu}
\end{align*}

\subsubsection{Rate choices}
By Markov's inequality, \eqref{ratechoicere}, and \eqref{eqratechoicesu}, there exists a random binning choice and a constant $a>0$ such that $\max_{\mathcal{A} \in \mathbb{A}}  \mathbb{P} [ \hat{Y}_{\mathcal D}^n(\mathcal{A}) \neq {Y}_{\mathcal D}^n  ]  + \max_{\mathcal{U} \in \mathbb{U}} \mathbb{V} ( p_{M_{\mathcal{D}} S_{\mathcal{D}} X_{\mathcal{U}}^n}, p^{\textup{(unif)}}_{M_{\mathcal{D}} S_{\mathcal{D}}} p_{ X_{\mathcal{U}}^n}) = o(e^{-na})$ provided that for any $\mathcal{S} \subseteq \mathcal{D}$, $(1 + \epsilon ) \max_{\mathcal{A} \in \mathbb{A}} {H}(Y_{\mathcal{S}}|Y_{ \mathcal{S}^c } X_{\mathcal{A}})< R'_{\mathcal{S}}$ and $R_{\mathcal{S}}+R_{\mathcal{S}}' < (1-\epsilon)\min_{\mathcal{U} \in \mathbb{U}}H(Y_{\mathcal{S}}|X_{\mathcal{U}})$. Finally, we remark that $ \mathbb{V} ( p_{M_{\mathcal{D}} S_{\mathcal{D}} X_{\mathcal{U}}^n}, p^{\textup{(unif)}}_{M_{\mathcal{D}} S_{\mathcal{D}}} p_{ X_{\mathcal{U}}^n}) = o(e^{-na})$ implies~\eqref{eqSeca} and \eqref{eqU} by \cite[Lemma~2.7]{bookCsizar}. 

\subsection{Proof of Theorem \ref{prop1}} \label{sec:cs}
Our coding scheme operates in two steps to successively deal with reliability and secrecy by means of reconciliation and privacy amplification. The main difficulty compared to the case $D=1$ is the analysis of privacy amplification because of the distributed setting induced by the multiple sub-dealers. Additionally, our analysis of the privacy amplification step requires a modified reconciliation protocol with additional properties compared to the case $D=1$. We describe our coding scheme in Section \ref{sec:subcs} and provide its analysis in Section \ref{sec:analy}. We use the same notation as in Appendix \ref{appth1ga}. 
\subsubsection{Coding scheme} \label{sec:subcs}
\paragraph{Reconciliation} \label{sec:rec}
 We define the encoding and decoding procedures for reconciliation through $2^{{D} } +1 $ nested random binnings as follows.

\emph{Binnings}: Fix $i \in \mathcal{D}$. For $y_i^n \in \mathcal{Y}_i^n$, for $j \in \llbracket 1, 2^{{D} } +1 \rrbracket$, draw uniformly at random an index in the set $\llbracket 1,2^{n {R}_{i,j}} \rrbracket$ and let this index assignment define the  function $b_{i,j} : \mathcal{Y}_i^n \to \llbracket 1,2^{nR_{i,j}} \rrbracket$. The value of $R_{i,j}$ will be chosen later. For any subset $\mathcal{S} \subseteq \llbracket 1, 2^{{D} } +1 \rrbracket$, we define ${R}_{i,\mathcal{S}} \triangleq \sum_{j \in \mathcal{S}} {R}_{i,j}  $.

\emph{Encoding at Sub-dealer $i \in \mathcal{D}$}: Given $y_i^{n}$, Sub-dealer $i \in \mathcal{D}$ computes $(m_{i,j})_{j \in \llbracket 1, 2^{{D} } +1 \rrbracket} \triangleq \left(  b_{i,j}(y_i^n) \right)_{j \in \llbracket 1, 2^{{D} } +1 \rrbracket}$.

\emph{Decoding at the participants}: For $i \in \mathcal{D}$, given $m_i \triangleq (m_{i,j})_{j \in \llbracket 1, 2^{{D} } +1 \rrbracket}$, $y_{1:i-1}^{n} \triangleq (y_{j}^{n})_{j \in \llbracket1 ,i-1 \rrbracket}$, and  $x_{\mathcal{L}}^n$, output $\hat{y}_i^n$ if it is the unique sequence such that  $(\hat{y}_{i}^n, y_{1:i-1}^{n}, x_{\mathcal{L}}^n) \in \mathcal{T}_{\epsilon}^n( {Y}_{1:i} {X}_{\mathcal{L}})$ and $ \left(  b_{i,j}(\hat{y}_i^n) \right)_{j \in \llbracket 1, 2^{{D} } +1 \rrbracket} = (m_{i,j})_{j \in \llbracket 1, 2^{{D} } +1 \rrbracket}$, otherwise output $1$.

\emph{Design properties of the reconciliation protocol}: Fix $i \in \mathcal{D}$. We first introduce additional definitions. Let $\delta>0$. Define for $\mathcal{S} \subseteq \mathcal{D}$, $\bar{R}_{i,\mathcal{S}} \triangleq H(Y_i|Y_{1:i-1}Y_{\mathcal{S}} X_{\mathcal{L}}) - \delta$ if $H(Y_i|Y_{1:i-1}Y_{\mathcal{S}} X_{\mathcal{L}}) \neq 0$ and $\bar{R}_{i,\mathcal{S}} \triangleq 0$ otherwise. We sort the sequence $(\bar{R}_{i,\mathcal{S}})_{\mathcal{S} \subseteq \mathcal{D}}$ in increasing order and denote the result by $(\bar{R}_{i,j})_{j \in \llbracket 1, 2^{{D}}\rrbracket}$. For notation convenience, we denote by $\mathcal{S}_j$, $j\in \llbracket 1 , 2^D \rrbracket $, the subset of $\mathcal{D}$ such that $\bar{R}_{i,j} = H(Y_i|Y_{1:i-1}Y_{\mathcal{S}_j} X_{\mathcal{L}}) - \delta$. Observe that $\bar{R}_{i,1} = 0$ and $\bar{R}_{i,2^{D}} = H(Y_i|Y_{1:i-1} X_{\mathcal{L}}) - \delta$.

\begin{enumerate}[(i)]
\item We will design the reconciliation such that, for any $i\in\mathcal{D}$, the participants in $\mathcal{L}$ can form an approximation $\hat{Y}^{n}_i$  of $Y_i^{n}$, from $(M_{i,j})_{j \in \llbracket 1, 2^{{D} } +1 \rrbracket}$ and $(Y_{1:i-1}^{n}, X_{\mathcal{L}}^n)$, such that $\mathbb{P}\left[  \hat{Y}^{n}_i \neq Y^{n}_i \right] \xrightarrow{n \to \infty} 0 $. 
\item For $j \in \llbracket 1, 2^{D}\rrbracket$ such that $H(Y_i|Y_{1:i-1}Y_{\mathcal{S}_j} X_{\mathcal{L}}) \neq 0$, we will design the reconciliation such that  almost independence holds between  $M_{i,1:j} \triangleq (M_{i,k})_{k \in \llbracket 1,j \rrbracket}$ and $(Y_{\llbracket 1,i-1 \rrbracket \cup \mathcal{S}_j}^n,X^n_{\mathcal{L}})$, in the sense that $n\mathbb{V} ( p_{M_{i,1:j} Y_{\llbracket 1,i-1 \rrbracket \cup \mathcal{S}_j}^n X^n_{\mathcal{L}}}, p^{\textup{(unif)}}_{M_{i,1:j}} p_{ Y_{\llbracket 1,i-1 \rrbracket \cup \mathcal{S}_j}^n X^n_{\mathcal{L}}})\xrightarrow{n \to \infty} 0 $, where $p^{\textup{(unif)}}_{M_{i,1:j}}$ is the uniform distribution over the sample space of $p_{M_{i,1:j}}$.
\end{enumerate}
 
Note that the second property is crucial in our analysis of privacy amplification, and is not necessary in the treatment of the case $D=1$.

\paragraph{Privacy Amplification}
We rely on two-universal hash functions as defined next. 
 \begin{defn}[{\!\cite{Carter79}}]
 A family $\mathcal{F}$ of two-universal hash functions $\mathcal{F} = \{f:\{0,1 \}^n \to \{0,1\}^r\}$ is such~that
$
 \forall x,x' \in \{0,1 \}^n, x \neq x' \implies \mathbb{P} [F(x)=F(x')] \leq 2^{-r},
 $
 where $F$ is a function uniformly chosen in $\mathcal{F}$.
 \end{defn}

Suppose that the reconciliation step in Section \ref{sec:subcs} is independently repeated $B$ times. Let  $\widehat{Y}^{nB}_d$, $d\in\mathcal{D}$, be the estimate of $Y^{nB}_d$. 
 For $d \in\mathcal{D}$, let $F_d : \{0,1\}^{nB} \rightarrow \{0,1\}^{r_d}$, be  uniformly chosen in a family $\mathcal{F}_d$ of two-universal hash functions. We leave the quantities $(r_d)_{d \in \mathcal{D}}$ unspecified in this section, and will specify them in~Section~\ref{sec:analy}.
 The privacy amplification step operates as follows. Sub-dealer $d\in\mathcal{D}$ computes $S_d \triangleq F_d(Y^{nB}_d)$, while the participants in $\mathcal{L}$ compute for $d\in\mathcal{D}$, $\widehat{S}_d \triangleq F_d(\widehat{Y}^{nB}_d)$, where $\widehat{Y}^{nB}_d$ has been obtained in the reconciliation step.

\subsubsection{Coding scheme analysis} \label{sec:analy}

We now show that any rate-tuple $(R_d)_{d \in \mathcal{D}}$ in $\mathcal{R}_{1}^{(\textup{in})}$, defined in Theorem \ref{prop1}, is achievable.

\paragraph{Analysis of reconciliation} \label{sec:reconcil}
We first prove that Property (i) of Section \ref{sec:rec} holds. 
The probability of error averaged over  the random choice of the binnings $(b_{i,j})_{j \in \llbracket 1, 2^{{D} } +1 \rrbracket }$ is upper bounded~as
\begin{align*}
&\mathbb{E} \left[\mathbb{P}\left[  \hat{Y}^{n}_i \neq Y^{n}_i \right]\right] \leq \mathbb{P}\left[ \mathcal{E}_{i,1} \right] + \mathbb{P}\left[ \mathcal{E}_{i,2} \right], 
\end{align*}
where 
\begin{align*}
\mathcal{E}_{i,1} &\triangleq \left\{ (Y_{1:i}^n,X_{\mathcal{L}}^n) \notin \mathcal{T}_{\epsilon}^n({Y}_{1:i}  {X}_{\mathcal{L}}) \right\},  \\
\mathcal{E}_{i,2} &\triangleq \!\left\{ \exists \hat{y}_i^n \neq Y_i^n, \right. \\
& \phantom{--}\left. \left(  b_{i,j}(\hat{y}_i^n) \right)_{j \in \llbracket 1, 2^{{D} } +1 \rrbracket} \!=\! \left(  b_{i,j}(Y_i^n) \right)_{j \in \llbracket 1, 2^{{D} } +1 \rrbracket}\! \right. \\
& \phantom{--}\left.\text{ and }\! (\hat{y}_i^n,Y_{1:i-1}^n,X_{\mathcal{L}}^n) \in \mathcal{T}_{\epsilon}^n({Y}_{1:i}  {X}_{\mathcal{L}}) \!\right\}\!.
\end{align*}
Similar to the proof of \eqref{2demi2} and \eqref{eqdemi}, one can show that
\begin{align}
&\mathbb{E} \left[\mathbb{P}\left[  \hat{Y}^{n}_i \neq Y^{n}_i \right]\right] \leq 2  |\mathcal{Y}_{1:i}| |\mathcal{X}_{\mathcal{L}}| e^{-n\epsilon^2 \mu_{Y_{1:i}X_{\mathcal{L}}}} \nonumber
\\
&\phantom{------}+ 2^{-n(  R_{i,\llbracket 1, 2^{{D} } +1 \rrbracket} - H(Y_i|Y_{1:i-1}X_{\mathcal{L}})(1+\epsilon))}. \label{eqchoicerates}
\end{align}

We next prove that Property (ii) of Section \ref{sec:rec} holds.  Let $i \in \mathcal{D}$ and $j \in \llbracket 1, 2^{{D}}\rrbracket$ such that $H(Y_i|Y_{1:i-1}Y_{\mathcal{S}_j} X_{\mathcal{L}}) \neq 0$. In the following, for notation convenience, we define  $Z_{i,j} \triangleq (Y_{\llbracket 1, i-1 \rrbracket \cup \mathcal{S}_j},  X_{\mathcal{L}})$.
We have
\begin{align*}
&p_{M_{i,1:j} Z_{i,j}^n }(m_{i,1:j},z_{i,j}^n) \\
&= \sum_{y_i^n} p(y_{i}^n,z_{i,j}^n) \mathds{1} \{b_{i,1:j}(y_i^n)= m_{i,1:j} \}, \forall m_{i,1:j}, \forall z_{i,j}^n,
\end{align*}
where $ b_{i,1:j}(y_i^n) \triangleq (b_{i,k}(y_i^n))_{k \in \llbracket 1 ,j \rrbracket}$,
hence, on average over $(b_{i,k})_{k \in \llbracket 1, j \rrbracket }$,  
\begin{align*}
 &\mathbb{E} \left[p_{M_{i,1:j} Z_{i,j}^n}(m_{i,1:j},z_{i,j}^n) \right] \\
 &= p(z_{i,j}^n) 2^{-nR_{i,\llbracket 1, j \rrbracket}}, \forall m_{i,1:j}, \forall z_{i,j}^n. 
\end{align*}
Then, similar to the proof of \eqref{ademia} and \eqref{ademia2}, one can show that
\begin{align}
& \mathbb{E} [ \mathbb{V} ( p_{M_{i,1:j} Z_{i,j}^n}, p^{\textup{(unif)}}_{M_{i,1:j}} p_{ Z_{i,j}^n}) ]  \nonumber\\
& \leq 2  |\mathcal{Y}_{i}| |\mathcal{Z}_{i,j}| e^{-n\epsilon^2 \mu_{Y_{i}Z_{i,j}}}  \nonumber\\
&\phantom{-} +  2^{-\frac{n}{2} \left[ (1-3\epsilon) H(Y_i |Y_{1:i-1}Y_{\mathcal{S}_j} X_{\mathcal{L}}) - R_{i,\llbracket 1 , j \rrbracket}\right]}. \label{eqchoicerates2}
\end{align}

Finally, we choose the rates as follows.
  Let $i\in \mathcal{D}$. We define for $j \in \llbracket 2, 2^{D}  \rrbracket$, ${R}_{i,j} \triangleq \bar{R}_{i,j} - \bar{R}_{i,j-1}$, and ${R}_{i,1} \triangleq \bar{R}_{i,1}$. We thus have for any $j \in \llbracket 1, 2^{D} \rrbracket$, $R_{i,\llbracket 1 , j \rrbracket} = \bar{R}_{i,j}$. We then choose $\delta \triangleq 3\epsilon H(Y_i |Y_{1:i-1}  X_{\mathcal{L}}) + \epsilon$ and $R_{i,2^D +1} \triangleq \delta + \epsilon H(Y_i|Y_{1:i-1} X_{\mathcal{L}}) + \epsilon$.

  Hence, we have  $R_{i,\llbracket 1, 2^{{D} } +1 \rrbracket} = \bar{R}_{i, 2^D} + R_{i,2^D +1} = (1+ \epsilon) H(Y_i|Y_{1:i-1} X_{\mathcal{L}}) + \epsilon$ and $(1-3\epsilon) H(Y_i |Y_{1:i-1} Y_{\mathcal{S}_j}X_{\mathcal{L}}) - R_{i,\llbracket 1, j \rrbracket} = (1-3\epsilon) H(Y_i |Y_{1:i-1} Y_{\mathcal{S}_j} X_{\mathcal{L}}) - H(Y_i|Y_{1:i-1}Y_{\mathcal{S}_j} X_{\mathcal{L}}) + \delta \geq \epsilon$, which ensures, by \eqref{eqchoicerates} and \eqref{eqchoicerates2}, that $\mathbb{E} \left[ \sum_{i \in \mathcal{D}} \sum_{j \in \llbracket 1, 2^D\rrbracket} \mathbb{V} ( p_{M_{i,1:j} Z_{i,j}^n}, p^{\textup{(unif)}}_{M_{i,1:j}} p_{ Z_{i,j}^n}) \right.$ $\left.+\sum_{i \in \mathcal{D}}\mathbb{P}\left[  \hat{Y}^{n}_i \neq Y^{n}_i \right]\right] \xrightarrow{n \to \infty} 0.$ Then, by Markov's Lemma, there exist binnings $(b_{i,j})_{i\in \mathcal{D}, j\in \llbracket 1, 2^D +1\rrbracket}$ such that for any $i \in \mathcal{D}$, $\mathbb{P}\left[  \hat{Y}^{n}_i \neq Y^{n}_i \right] \xrightarrow{n \to \infty} 0$ and for any $i \in \mathcal{D}$, $j \in \llbracket 1, 2^D \rrbracket$, $\mathbb{V} ( p_{M_{i,1:j} Z_{i,j}^n}, p^{\textup{(unif)}}_{M_{i,1:j}} p_{ Z_{i,j}^n})\xrightarrow{n \to \infty} 0$. 
  
\paragraph{Analysis of privacy amplification}

We use the following version of the leftover hash lemma~\cite{haastad1999pseudorandom,dodis2008fuzzy} to analyze the privacy amplification step. The lemma is of independent interest as related versions of this lemma \cite{wullschleger2007oblivious,nascimento2008oblivious,chou2017secret,chou2021private}  had found a wide variety of applications including oblivious transfer~\cite{nascimento2008oblivious,wullschleger2007oblivious,Chou21}, commitment \cite{chou2022bc}, secret generation \cite{chou2017secret,chou2019biometric},  multiple-access channel resolvability \cite{sultana2022multiple}, and private classical communication over quantum multiple-access channels \cite{chou2021private}. 

     \begin{lem}[Distributed leftover hash lemma] \label{lemloh}
 Consider a sub-normalized non-negative function  $ p_{ X_{\mathcal L}Z}$ defined over $\bigtimes_{l \in \mathcal{L}}\mathcal{X}_{l}\times \mathcal{Z}$, where $X_{\mathcal{L}} \triangleq (X_l)_{l \in \mathcal{L}}$ and, $\mathcal{Z}$, $\mathcal{X}_{l}$, $l \in\mathcal{L}$, are finite alphabets.    
 For $l \in \mathcal{L}$, let  $F_l:\{0,1\}^{n_l} \longrightarrow \{0,1\}^{r_l}$, be uniformly chosen in a family $\mathcal{F}_l$ of two-universal hash functions. Define $s_{\mathcal L} \triangleq \prod_{l \in \mathcal{L}} s_l$, where $s_l\triangleq |\mathcal{F}_l|$, $l \in \mathcal{L}$, and for any $\mathcal{S} \subseteq \mathcal{L}$, define $r_{\mathcal{S}}\triangleq \sum_{i \in \mathcal{S}}r_i$. Define also ${F}_{\mathcal{L}}\triangleq (F_l)_{l \in \mathcal{L}}$ and 
   $
        F_{\mathcal{L}}( X_{\mathcal{L}})\triangleq \left( F_l(X_l)\right)_{l\in\mathcal{L}}$. Then, for any $q_Z$ defined over $\mathcal{Z}$ such that $\textup{supp}(q_Z) \subseteq \textup{supp}(p_Z)$, we have
   \begin{align}
      \mathbb{V}({{p}_{F_{\mathcal{L}}( X_{\mathcal{L}})  F_{\mathcal{L}}Z}}, p_{U_{\mathcal K}} p_{U_{\mathcal F}} p_Z) \leq   {{{\sqrt{{ \displaystyle\sum_{\substack{{\mathcal S\subseteq\mathcal L}, {\mathcal S \neq \emptyset}}}}2^{r_{\mathcal S}-H_{\infty}(p_{X_{\mathcal S}Z}|q_Z)}}}}}, \label{eq:lohl}
       \end{align}
      where  $p_{U_{\mathcal K}}$ and  $p_{U_{\mathcal F}}$ are the uniform distributions over $\llbracket 1,2^{r_{{\mathcal{L}}}} \rrbracket$ and $\llbracket 1,{s_{{\mathcal{L}}}} \rrbracket$, respectively, and  the min-entropies are defined as in \cite{renner2008security}, i.e., for any $\mathcal S\subseteq\mathcal L, \mathcal S \neq \emptyset$, 
$$H_{\infty}(p_{X_{\mathcal{S}} Z}|q_{Z})\triangleq -\log \displaystyle \max_{\substack{{x_{\mathcal{S}} \in \mathcal{X}_{\mathcal{S}}}\\ z \in \textup{supp}(q_{Z}) }}\frac{p_{X_{\mathcal{S}} Z}(x_{\mathcal{S}}, z)}{q_{Z}(z)}.$$
          \end{lem}
          \begin{proof}
See          Appendix \ref{app_loh}.
          \end{proof}
A challenge with using Lemma \ref{lemloh} is the evaluation of  the min-entropies in \eqref{eq:lohl}. A possible solution is to use the method in   \cite{Maurer00} to lower bound a min-entropy in terms of a Shannon entropy. However, one drawback of this method is that an extra round of reconciliation is needed, as in~\cite{chou2014separation}, which complexifies the coding scheme. Another solution could be to rely on the notion of smooth min-entropy, as in \cite{renner2008security}. However, this technique is challenging to apply here because one would need to \emph{simultaneously} smooth all the min-entropies in \eqref{eq:lohl}. Instead, we propose  to lower bound the min-entropies in~\eqref{eq:lohl} by relying on the following lemma.

 \begin{lem} \label{lems1}
        Let $(\mathcal{Y}_d)_{d \in \mathcal{D}}$ be $D$ finite alphabets and define for $\mathcal{S}\subseteq \mathcal{D}$, $\mathcal{Y}_{\mathcal{S}}\triangleq \bigtimes_{d \in \mathcal{S}} \mathcal{Y}_d$. Consider the random variables  $Y^{n}_{\mathcal{D}}\triangleq ({Y}^{n}_d)_{d \in \mathcal{D}}$ and $Z^{n}$ defined over $\mathcal{Y}_{\mathcal{D}}^n  \times\mathcal{Z}^n$ with probability distribution $q_{Y^{n}_{\mathcal{D}} Z^{n}}\triangleq \prod_{i=1}^n q_{Y_{\mathcal{D}} Z}$. For any $\epsilon>0$, there exists a subnormalized non-negative function $w_{Y^{n}_{\mathcal{D}} Z^{n}}$ defined over $\mathcal{Y}^n_{\mathcal{D}} \times\mathcal{Z}^n$ such that $\mathbb{V}(q_{Y^{n}_{\mathcal{D}} Z^{n}},w_{Y^{n}_{\mathcal{D}} Z^{n}})\leq\epsilon$ and
       \begin{align*}
           \forall \mathcal{S}\subseteq \mathcal{D}, H_{\infty}(w_{Y^{n}_{\mathcal{S}} Z^{n}}|q_{Z^{n}})\geq n H({Y_{\mathcal{S}}}|Z)-n \delta_{\mathcal{S}}(n),
       \end{align*}
       where $\delta_{\mathcal{S}}(n)\triangleq (\log (\lvert\mathcal{Y}_{\mathcal{S}}\rvert+3))\sqrt{\frac{2}{n}(D+\log(\frac{1}{\epsilon}))}$. 
\label{lem2}
       \end{lem}
 \begin{proof}
 See Appendix \ref{App_lemma2}.
 \end{proof}

We now combine Lemma \ref{lemloh} and Lemma \ref{lems1} as follows.

\begin{lem} \label{lemamp}
For any $\mathcal{U} \in \mathbb{U}$, we have
\begin{align}
 &\mathbb{V} ( p_{F_{\mathcal{D}}(Y^{nB}_{\mathcal{D}})F_{\mathcal{D}} M^B_{\mathcal{D}} X^{nB}_{\mathcal{U}} }, p_{U_{\mathcal{D}}} p_{U_{\mathcal{F}}}p_{M^B_{\mathcal{D}}X^{nB}_{\mathcal{U}}  }) \nonumber \\
 & \leq    2 \epsilon+  \sqrt{ \sum_{ \substack{ \mathcal{S} \subseteq {\mathcal{D}} \\ \mathcal{S}\neq \emptyset }} 2^{ r_{\mathcal{S}} - BH \left( Y^n_{\mathcal{S}}|M_{\mathcal{D}}  X^n_{\mathcal{U}}   \right)+ B \delta_{\mathcal{S}}(n,B)} }, \label{eqlem3}
\end{align}
where $\delta_{\mathcal{S}}(n,B)\triangleq (\log (\lvert\mathcal{Y}_{\mathcal{S}}\rvert^n+3))\sqrt{\frac{2}{B}(D+\log(\frac{1}{\epsilon}))}$.
\end{lem}

\begin{proof}
See Appendix \ref{Applemamp}.
\end{proof}
Note that in the case $D=1$, a standard technique could be used \cite[Lemma 10]{Maurer00} to lower-bound the min-entropy appearing in the leftover hash lemma and study the effect of the public communication on the information leaked to unauthorized participants. However, using \cite[Lemma 10]{Maurer00} in the case $D>1$ to lower-bound the min-entropies in \eqref{eqremark1} would result in the achievability of 
\begin{align*}
 & \left\{  (R_d)_{d \in \mathcal{D}} : \right.\\
 & \left. R_{\mathcal{S }} \leq   \min_{ \mathcal{T} \subsetneq \mathcal{L} }  \left[I(Y_{\mathcal{S}};   X_{\mathcal{L}}  |X_{\mathcal{T}}) 
 - H(Y_{\mathcal{S}^c}|Y_{\mathcal{S}}X_{\mathcal{L}}) \right]^+ ,   \forall \mathcal{S} \subseteq \mathcal{D} \right\} 
\end{align*}
which is always contained in the region $\mathcal{R}_{1}^{(\textup{in})} $ of Theorem~\ref{prop1}. For this reason, we did not study the effect of the public communication on the information leaked to unauthorized participants in Lemma \ref{lemamp}. Instead, we do it by lower-bounding the Shannon entropies that appear in \eqref{eqlem3} as follows. Note that Property (ii) in the reconciliation protocol described in Section~\ref{sec:rec} plays a key role in the proof of Lemma~\ref{lemamp2}. 
\begin{lem} \label{lemamp2}
For any $\mathcal{S} \subseteq \mathcal{D}$, $\mathcal{S}\neq \emptyset$, we have
\begin{align*}
&H(Y_{\mathcal{S}}^n|M_{\mathcal{D}} X^n_{\mathcal{U}}) \geq  n \left[ I(Y_{\mathcal{S}};X_{\mathcal{L}}| X_{\mathcal{U}} )   - \delta(\epsilon)\right]   - \delta(n), 
\end{align*}
where $\delta(n)$ is such that $\lim_{n \to \infty} \delta(n) = 0$ and $\delta(\epsilon)$ is such that $\lim_{\epsilon \to 0} \delta(\epsilon) = 0$.
\end{lem}

\begin{proof}
See Appendix \ref{Applemamp2}.

\end{proof}
We are now equipped to prove that \eqref{eqSeca} and \eqref{eqU} hold. For any $\mathcal{U} \in \mathbb{U}$ and $\xi>0$, we have
\begin{align*}
&\mathbb{V} ( p_{F_{\mathcal{D}}(Y^{nB}_{\mathcal{D}})F_{\mathcal{D}} M_{\mathcal{D}}^B X^{nB}_{\mathcal{U}} }, p_{U_{\mathcal{D}}} p_{U_{\mathcal{F}}}p_{M^B_{\mathcal{D}}X^{nB}_{\mathcal{U}}  })    \\
&\stackrel{(a)} \leq      2 \epsilon+  \sqrt{ \sum_{ \substack{ \mathcal{S} \subseteq {\mathcal{D}} \\ \mathcal{S}\neq \emptyset }} 2^{ r_{\mathcal{S}} -   nB  I(Y_{\mathcal{S}};X_{\mathcal{L}}| X_{\mathcal{U}} )   + nB \delta(\epsilon)   +B \delta(n) + B \delta_{\mathcal{S}}(n,B)} }   \displaybreak[0] \\
& \stackrel{(b)} \leq 2 \epsilon+  \sqrt{ \sum_{ \substack{ \mathcal{S} \subseteq {\mathcal{D}} \\ \mathcal{S}\neq \emptyset }} 2^{ - n\xi  } }  \\
& \leq 2 \epsilon+ 2^{D/2}  2^{ - n\xi/2  }, \numberthis \label{eq45c}
\end{align*}
where $(a)$ holds by Lemmas \ref{lemamp} and \ref{lemamp2}, in $(b)$ we have chosen $r_{\mathcal{L}}$ such that for any $\mathcal{S} \subseteq \mathcal{D}$, 
\begin{align*}
r_{\mathcal{S}} & \leq \min_{\mathcal{U} \in \mathbb{U}} nB  I(Y_{\mathcal{S}};X_{\mathcal{L}}| X_{\mathcal{U}} ) \\
&\phantom{-}  - nB \delta(\epsilon)   -B \delta(n)- B \delta_{\mathcal{S}}(n,B) - n\xi .
\end{align*}
We conclude that \eqref{eqSeca} and \eqref{eqU} hold by \eqref{eq45c} and \cite[Lemma~2.7]{bookCsizar}.

\subsection{Proof of Theorem \ref{th2}} \label{App:th2}

By successively, rather than jointly (as in Theorem \ref{prop1}), considering the security constraints for the two sub-dealers, we prove Theorem \ref{th2}. The coding scheme and its analysis are described in Sections~\ref{csth5} and \ref{csth5a}, respectively. Note that $\mathcal{R}(\{1,2\}) = \mathcal{R}_{1}^{(\textup{in})}$, where the achievability of $\mathcal{R}_{1}^{(\textup{in})}$ follows from Theorem \ref{prop1} with $D=2$. Note also that if one can show the achievability of  $\left[  \mathcal{R}(\{1\})\times \mathcal{R}(\{2\}|\{1\}) \right]$, then one has the achievability of $\left[  \mathcal{R}(\{2\})\times \mathcal{R}(\{1\}|\{2\}) \right]$ by exchanging the roles of the two dealers. Hence, it is sufficient to prove the achievability of  $\left[  \mathcal{R}(\{1\})\times \mathcal{R}(\{2\}|\{1\}) \right]$.

\subsubsection{Coding scheme} \label{csth5}

In this section, we use the notation $\delta(n)$ to denote a generic function of $n$ that vanishes to $0$ as $n$ goes to infinity.  Our achievability scheme operates in two phases as follows.

\paragraph{Initialization phase} By using $n_2'$ source observations, Sub-dealer $2$ shares a secret $K_2$ with non-zero rate with the requirement $\lim_{n_2' \to \infty} \displaystyle\max_{ \mathcal{T} \subsetneq \mathcal{L}}   I\left({K}_{2} ; M_{2, \textup{init}}, X^{n_2'}_{\mathcal{T}}  \right) = \delta(n_2')$, where $M_{2, \textup{init}}$ corresponds to the public communication sends by Sub-dealer $2$. This is possible by Theorem~\ref{cor4} because we assumed that $\min_{d \in \{1,2\}}\min_{ \mathcal{T} \subsetneq \mathcal{L}} I(Y_{d};X_{\mathcal{L}}|X_{\mathcal{T}}) >0$. Define  for $\mathcal{U} \subsetneq \mathcal{L}$, $I_{2}(\mathcal{U}) \triangleq (M_{2, \textup{init}}, X_{\mathcal{U}}^{n_2'} )$.

\paragraph{Successive secret distribution phase} This phase requires $n$ source observations. Sub-dealer~1 performs the coding scheme in the proof of Theorem \ref{prop1} for the case $D=1$  with the requirement 
\begin{align} \label{eqreq1}
\lim_{n \to \infty} \displaystyle\max_{\mathcal{T} \subsetneq \mathcal{L}}  I\left({S}_{1} ; M_{1}, X^n_{\mathcal{T}}  \right) = 0.
\end{align}

Sub-dealer~2 performs the coding scheme in the proof of Theorem \ref{prop1} assuming that all the participants have access to $Y^n_1$  with the requirement
\begin{align} \label{eqreq2}
\lim_{n \to \infty} \displaystyle\max_{\mathcal{T} \subsetneq \mathcal{L}}  I\left({S}_{2} ; M_{2}, X^n_{\mathcal{T}}, Y^n_{1}  \right) = 0,
\end{align}
for the case $D=1$ with the following modification: Using the same notation as in the proof of Theorem \ref{prop1}, instead of defining $M_2 \triangleq M_{2,1:3}$, define  $M_2$ as  $M_2 \triangleq (M_2', M_2'')$ with $M_2' \triangleq K_2 \oplus M_{2,3}$ and $M_2'' \triangleq M_{2,1:2}$. By Property (ii) in the reconciliation step of the proof of Theorem \ref{prop1}, we have \begin{align} \label{eqreq2bbb}
I(M_2'';Y_1^n X^n_{\mathcal{L}} ) = \delta(n).\end{align}
Then, the proof of Theorem \ref{prop1} is still valid  because  $K_2$ is known by the participants (by the initialization phase provided that $n'_2$ is such that $|K_2| = |M_{2,3}|$), and the secrecy rates $R_1 = R(\{1\} )$ and $R_2 = R(\{2\} | \{ 1\})$ are achievable for Requirements \eqref{eqreq1} and \eqref{eqreq2}. Note that $|K_2|= |M_{2,3}|$ is negligible compared to $n$. More specifically, by inspecting the proof of Theorem~\ref{prop1}, one can choose $|M_{2,3}|$, on the order of $n^{1/2 -\xi}$, $\xi>0$, similar to \cite{chou2016coding,chou2013data}. Hence, it only remains to show that Requirements \eqref{eqreq1} and~\eqref{eqreq2} imply Requirements~\eqref{eqSeca} and~\eqref{eqU}.

\subsubsection{Coding scheme analysis} \label{csth5a}

We first prove that \eqref{eqU} holds. We have
\begin{align*}
&\log(|\mathcal{S}_1||\mathcal{S}_2|) - H(S_1,S_2)\\
& = \log(|\mathcal{S}_1||\mathcal{S}_2|) - H(S_1) - H(S_2) + I(S_2;S_1) \\
& \leq \log(|\mathcal{S}_1||\mathcal{S}_2|) - H(S_1) - H(S_2) + I(S_2;Y_1^n) \\
& \xrightarrow{n\to \infty} 0,
\end{align*}
where the limit holds by almost uniformity of $S_1$ and $S_2$, and by \eqref{eqreq2}. 

We now prove that \eqref{eqSeca} holds. We first ignore the initialization phase and upper bound the quantity $\displaystyle\max_{\mathcal{T} \subsetneq \mathcal{L}}I(S_{1},S_{2}; M_{1},M_{2}, X_{\mathcal{T}}^n)$. 

\begin{lem} \label{lemsc1}
For any $\mathcal{T} \subsetneq \mathcal{L}$, we have
\begin{align*}
I(S_{1},S_{2}; M_{1},M_{2}, X_{\mathcal{T}}^n)
&\leq \delta(n)  + \delta(n_2'), \numberthis \label{eqf1}
\end{align*}
\end{lem}

\begin{proof}
See Appendix \ref{Applemsc1}.

\end{proof}

Next, we jointly consider the initialization phase and the successive secret distribution phase.   

\begin{lem} \label{lemsc2}
We have for any $\mathcal{U}, \mathcal{T} \subsetneq \mathcal{L}$,
\begin{align*}
&I(S_{1},S_{2}; I_2(\mathcal{U}), M_{1},M_{2}, X_{\mathcal{T}}^n)  
\leq \delta(n)+  \delta(n'_2).
\end{align*}
\end{lem}

\begin{proof}
See Appendix \ref{Applemsc2}.

\end{proof}

\section{Converse proofs} \label{appth2ga}

\subsection{Proof of Theorem \ref{th2ga}} \label{secconvproof2}
Consider a secret-sharing strategy, as in Definition \ref{definition_modelg}, that satisfies the constraints~\eqref{eqrel}, \eqref{eqSeca}, and~\eqref{eqU}. For any $\mathcal{T} \subseteq \mathcal{D}$, $\mathcal{A} \in \mathbb{A}$,  $\mathcal{U} \in \mathbb{U}$, we have 
\begin{align*}
nR_{\mathcal{T}} 
& = \log |\mathcal{S}_{\mathcal{T}} |\\
& \stackrel{(a)} \leq H(S_{\mathcal{T}}) + o(n)\\
& \stackrel{(b)} \leq H(S_{\mathcal{T}} |M_{\mathcal{D}} X_{\mathcal{U}}^n) + o(n)\\
& \stackrel{(c)} \leq I (S_{\mathcal{T}}; \widehat{S}_{\mathcal{D}}(\mathcal{A}) | M_{\mathcal{D}} X_{\mathcal{U}}^n) + o(n)\\
& \stackrel{(d)}\leq I (Y^n_{\mathcal{T}}; X_{\mathcal{A}}^nM_{\mathcal{D}} | M_{\mathcal{D}} X_{\mathcal{U}}^n) + o(n)\\
& = I (Y^n_{\mathcal{T}}; X_{\mathcal{A}}^n | M_{\mathcal{D}} X_{\mathcal{U}}^n) + o(n)\\
& \leq  I (Y^n_{\mathcal{T}} M_{\mathcal{T}}; X_{\mathcal{A}}^n M_{\mathcal{T}^c}|  X_{\mathcal{U}}^n) + o(n)\\
&  \stackrel{(e)} \leq  I (Y^n_{\mathcal{T}}; X_{\mathcal{A}}^n Y^n_{\mathcal{T}^c}|  X_{\mathcal{U}}^n) + o(n) \\
& = n I (Y_{\mathcal{T}}; X_{\mathcal{A}}Y_{\mathcal{T}^c} |  X_{\mathcal{U}}) + o(n), \numberthis \label{eqreconv1}
\end{align*}
where $(a)$ holds by \eqref{eqU}, $(b)$ holds by \eqref{eqSeca}, $(c)$ holds by Fano's inequality and \eqref{eqrel}, $(d)$ holds because $\widehat{S}_{\mathcal{D}}(\mathcal{A})$ is a function of $(X_{\mathcal{A}}^n,M_{\mathcal{D}})$ and $S_{\mathcal{T}}$ is a function of $Y^n_{\mathcal{T}}$, $(e)$ holds because $M_{\mathcal{S}}$ is a function of $Y^n_{\mathcal{S}}$ for any $\mathcal{S} \subseteq \mathcal{D}$.

Then, since \eqref{eqreconv1} is valid for any $\mathcal{A} \in \mathbb{A}$,  $\mathcal{U} \in \mathbb{U}$, an upper-bound on the sum-rate $R_{\mathcal{S}} = \sum_{d\in\mathcal{S}} R_d$, $\mathcal{S} \subseteq \mathcal{D}$,    is $\displaystyle\min_{\mathcal{A} \in \mathbb{A}} \min_{\mathcal{U} \in \mathbb{U}}I(Y_{\mathcal{S}}; X_{\mathcal{A}} Y_{\mathcal{S}^c} |X_{\mathcal{U}}) + o(1)$.

\subsection{Proof of Theorem \ref{th3}} \label{secconvproof2b}

The proof of Theorem~\ref{th3} follows from the proof of Theorem~\ref{th2ga}, since for the all-or-nothing access structure we have for any $\mathcal{S} \subseteq \mathcal{D} $
\begin{align*}
\min_{\mathcal{A} \in \mathbb{A}} \min_{\mathcal{U} \in \mathbb{U}} I(Y_{\mathcal{S}};X_{\mathcal{\mathcal{A}}} Y_{\mathcal{S}^c}| X_{\mathcal{U}}) 
& =   \min_{\mathcal{U} \in \mathbb{U}} I(Y_{\mathcal{S}};X_{\mathcal{\mathcal{L}}} Y_{\mathcal{S}^c}| X_{\mathcal{U}}) \\
& =   \min_{\mathcal{U} \subsetneq \mathcal{L}} I(Y_{\mathcal{S}};X_{\mathcal{\mathcal{L}}} Y_{\mathcal{S}^c}| X_{\mathcal{U}}) .
\end{align*}

\section{Proof of capacity results in some special cases} \label{secoptimal}
 \subsection{Proof of Theorem \ref{cor4}} \label{appth6}
 The result holds by Corollary \ref{cor3} using the facts that for any $\mathcal{T} \subsetneq \mathcal{L}$, the Markov chain $Y_{\mathcal{D}} - X_{\mathcal{L}} - X_{\mathcal{T}}$ holds, and that $\mathbb{U}$ is the set of strict subsets of $\mathcal{L}$ for the all-or-nothing access structure. 

\subsection{Proof of Theorem \ref{thm_thr}} \label{App_th7}
In the following, for any $\mathcal{S} \subseteq \mathcal{L}$, $\mathcal{T} \subseteq \mathcal{D}$, we use the notation $K_{\mathcal{S},\mathcal{T}} \triangleq (K_{l,d})_{d \in \mathcal{S},d \in \mathcal{T}} $.

We first prove the achievability part. Let $t \in \llbracket 1 , L \rrbracket$ and $\mathcal{S} \subseteq \mathcal{D}$. We have
\begin{align*}
\displaystyle\max_{\mathcal{A} \in \mathbb{A}_t} H(Y_{\mathcal{S}}|Y_{\mathcal{S}^c}X_{\mathcal{A}}) 
& \stackrel{(a)}= \displaystyle\max_{\mathcal{A} \in \mathbb{A}_t} H(K_{\mathcal{L},\mathcal{S}}|K_{\mathcal{L},\mathcal{S}^c}K_{\mathcal{A},\mathcal{D}}) \\
& = \displaystyle\max_{\mathcal{A} \in \mathbb{A}_t} H(K_{\mathcal{L},\mathcal{S}}|K_{\mathcal{L},\mathcal{S}^c}K_{\mathcal{A},\mathcal{S}}K_{\mathcal{A},\mathcal{S}^c}) \\
& \stackrel{(b)}= \displaystyle\max_{\mathcal{A} \in \mathbb{A}_t} H(K_{\mathcal{L},\mathcal{S}}|K_{\mathcal{A},\mathcal{S}}) \\
& \stackrel{(c)}= \displaystyle\max_{\mathcal{A} \in \mathbb{A}_t} H(K_{\mathcal{A}^c,\mathcal{S}}) \\
&  \stackrel{(d)}= |\mathcal{S}| (L-t), \numberthis \label{eqamax001}
\end{align*}
where $(a)$ holds by definition of $Y_{\mathcal{S}}^n$, $Y_{\mathcal{S}^c}^n$, and $X_{\mathcal{A}}^n$, $(b)$ holds by independence between $(K_{\mathcal{A},\mathcal{S}^c},K_{\mathcal{L},\mathcal{S}^c})$ and $(K_{\mathcal{L},\mathcal{S}},K_{\mathcal{A},\mathcal{S}})$, $(c)$ holds by independence between $K_{\mathcal{A}^c,\mathcal{S}}$ and $K_{\mathcal{A},\mathcal{S}}$, $(d)$~holds by independence and uniformity of the keys.
Next, let $z \in \llbracket 1, t-1 \rrbracket$ and $\mathcal{S} \subseteq \mathcal{D}$. We have
\begin{align*}
\displaystyle\min_{\mathcal{U} \in \mathbb{U}_z} H(Y_{\mathcal{S}}|X_{\mathcal{U}})
& \stackrel{(a)} = \displaystyle\min_{\mathcal{U} \in \mathbb{U}_z} H(K_{\mathcal{L},\mathcal{S}}|K_{\mathcal{U},\mathcal{D}})\\
&  = \displaystyle\min_{\mathcal{U} \in \mathbb{U}_z} H(K_{\mathcal{L},\mathcal{S}}|K_{\mathcal{U},\mathcal{S}}K_{\mathcal{U},\mathcal{S}^c})\\
& \stackrel{(b)}= \displaystyle\min_{\mathcal{U} \in \mathbb{U}_z} H(K_{\mathcal{L},\mathcal{S}}|K_{\mathcal{U},\mathcal{S}})\\
& \stackrel{(c)}= \displaystyle\min_{\mathcal{U} \in \mathbb{U}_z} H(K_{\mathcal{U}^c,\mathcal{S}})\\
& \stackrel{(d)}=  |\mathcal{S}| (L-z), \numberthis \label{equmax001}
\end{align*}
where $(a)$ holds by definition of $Y_{\mathcal{S}}^n$ and $X_{\mathcal{U}}^n$, $(b)$ holds by independence between $K_{\mathcal{U},\mathcal{S}^c}$ and $(K_{\mathcal{L},\mathcal{S}},K_{\mathcal{U},\mathcal{S}})$, $(c)$ holds by independence between $K_{\mathcal{U}^c,\mathcal{S}}$ and $K_{\mathcal{U},\mathcal{S}}$, $(d)$ holds by independence and uniformity of the keys. %
Next, we have
\begin{align*}
&\mathcal{R}^{(\textup{in})}(\mathbb{A}_t,\mathbb{U}_z)\\
& \stackrel{(a)}= \textup{Proj}_{(R_d)_{d\in\mathcal{D}}} \left\{ \left(R_d,R_d'\right)_{d\in\mathcal{D}} : \right.\\
&\left. \phantom{-----} \begin{array}{rl}
R_{\mathcal{S}}' &\geq \displaystyle\max_{\mathcal{A} \in \mathbb{A}_t} H(Y_{\mathcal{S}}|Y_{\mathcal{S}^c}X_{\mathcal{A}}), \forall \mathcal{S} \subseteq \mathcal{D}\\
R_{\mathcal{S}}' + R_{\mathcal{S}} &\leq \displaystyle\min_{\mathcal{U} \in \mathbb{U}_z} H(Y_{\mathcal{S}}|X_{\mathcal{U}}), \forall \mathcal{S} \subseteq \mathcal{D}
\end{array}      \right\}\\
& \stackrel{(b)} = \textup{Proj}_{(R_d)_{d\in\mathcal{D}}} \left\{ \left(R_d,R_d'\right)_{d\in\mathcal{D}} :\right.\\
&\left. \phantom{-----} \begin{array}{rl}
R_{\mathcal{S}}' &\geq |\mathcal{S}| (L-t), \forall \mathcal{S} \subseteq \mathcal{D}\\
R_{\mathcal{S}}' + R_{\mathcal{S}} &\leq |\mathcal{S}| (L-z), \forall \mathcal{S} \subseteq \mathcal{D}
\end{array}      \right\} \\
& \stackrel{(c)}=   \left\{ \left(R_d \right)_{d\in\mathcal{D}} :  R_{\mathcal{S}} \leq |\mathcal{S}| (t-z), \forall \mathcal{S} \subseteq \mathcal{D}     \right\},
\end{align*}
where $(a)$ holds by Theorem \ref{th1ga}, $(b)$ holds by \eqref{eqamax001} and \eqref{equmax001}, $(c)$ holds as follows. 
First, consider the system 
\begin{align}
\begin{pmatrix}
	R_{\mathcal{S}}' \geq |\mathcal{S}| (L-t), \forall \mathcal{S} \subseteq \mathcal{D} \\
	R_{\mathcal{S}}' + R_{\mathcal{S}} \leq |\mathcal{S}| (L-z), \forall \mathcal{S} \subseteq \mathcal{D}
	\end{pmatrix} \label{eqSI},
\end{align}
and remark that the set functions $f: 2^{\mathcal{D}}\to \mathbb{R}, \mathcal{S} \mapsto |\mathcal{S}|(L-z)- R_{\mathcal{S}}$ and $g: 2^{\mathcal{D}}\to \mathbb{R}, \mathcal{S} \mapsto - |\mathcal{S}| (L-t) $ are submodular, i.e., $\forall \mathcal{S},\mathcal{T} \subseteq \mathcal{D}, g(\mathcal{S}) + g(\mathcal{T}) \geq g(\mathcal{S} \cup \mathcal{T})+g(\mathcal{S} \cap \mathcal{T})$ and $f(\mathcal{S}) + f(\mathcal{T}) \geq f(\mathcal{S} \cup \mathcal{T})+f(\mathcal{S} \cap \mathcal{T})$. Hence, by Lemma~\ref{lemsubm} below, we have that the system \eqref{eqSI} has a solution if and only if  
\begin{align*}
	|\mathcal{S}|(L-t) \leq  |\mathcal{S}| (L-z) - R_{\mathcal{S}} , \forall \mathcal{S} \subseteq \mathcal{D},
\end{align*}
which we rewrite as
\begin{align*}
	R_{\mathcal{S}} \leq |\mathcal{S}| (t-z), \forall \mathcal{S} \subseteq \mathcal{D}.
\end{align*}
\begin{lem}[{\cite[Lemma 2]{zhang2017multi}}] \label{lemsubm}
Consider two submodular functions $f:2^{\mathcal{D}} \to \mathbb{R}$ and $g:2^{\mathcal{D}} \to \mathbb{R}$. Then, the following system of equations for $(x_d)_{d\in\mathcal{D}} \in \mathbb{R}_+^D$
\begin{align*}
-g(\mathcal{S})\leq\sum_{s \in \mathcal{S}} x_s \leq f(\mathcal{S}) ,\forall \mathcal{S} \subseteq \mathcal{D},
\end{align*}
has a solution if and only if $-g(\mathcal{S}) \leq f(\mathcal{S}) ,\forall \mathcal{S} \subseteq \mathcal{D}$.
\end{lem}

We now prove the converse. Let $t \in \llbracket 1 , L \rrbracket$, $z \in \llbracket 1, t-1 \rrbracket$, and $\mathcal{S} \subseteq \mathcal{D}$. We have
\begin{align*}
&\displaystyle\min_{\mathcal{A} \in \mathbb{A}_t} \min_{\mathcal{U} \in \mathbb{U}_z} I(Y_{\mathcal{S}};X_{\mathcal{\mathcal{A}}} Y_{\mathcal{S}^c}| X_{\mathcal{U}})\\
& \stackrel{(a)}= \displaystyle\min_{\mathcal{A} \in \mathbb{A}_t} \min_{\mathcal{U} \in \mathbb{U}_z} I(K_{\mathcal{L},\mathcal{S}};K_{\mathcal{A},\mathcal{D}} K_{\mathcal{L},\mathcal{S}^c}| K_{\mathcal{U},\mathcal{D}})\\
& = \displaystyle\min_{\mathcal{A} \in \mathbb{A}_t} \min_{\mathcal{U} \in \mathbb{U}_z} [I(K_{\mathcal{L},\mathcal{S}};K_{\mathcal{A},\mathcal{D}} | K_{\mathcal{U},\mathcal{D}}) \\
& \phantom{--}+ I(K_{\mathcal{L},\mathcal{S}};K_{\mathcal{L},\mathcal{S}^c}| K_{\mathcal{U},\mathcal{D}} K_{\mathcal{A},\mathcal{D}} )]\\
& \stackrel{(b)}\leq \displaystyle\min_{\mathcal{A} \in \mathbb{A}_t} \min_{\mathcal{U} \in \mathbb{U}_z} [I(K_{\mathcal{L},\mathcal{S}};K_{\mathcal{A},\mathcal{D}} | K_{\mathcal{U},\mathcal{D}}) \\
& \phantom{--}+ I(K_{\mathcal{L},\mathcal{S}}K_{\mathcal{U},\mathcal{S}} K_{\mathcal{A},\mathcal{S}}  ;K_{\mathcal{L},\mathcal{S}^c} K_{\mathcal{U},\mathcal{S}^c}  K_{\mathcal{A},\mathcal{S}^c} )]\\
& \stackrel{(c)}= \displaystyle\min_{\mathcal{A} \in \mathbb{A}_t} \min_{\mathcal{U} \in \mathbb{U}_z} I(K_{\mathcal{L},\mathcal{S}};K_{\mathcal{A},\mathcal{D}} | K_{\mathcal{U},\mathcal{D}})\\
& \stackrel{(d)}\leq \displaystyle\min_{\mathcal{A} \in \mathbb{A}_t} \min_{\mathcal{U} \in \mathbb{U}_z} I(K_{\mathcal{L},\mathcal{S}};K_{\mathcal{A},\mathcal{D}} K_{\mathcal{U},\mathcal{S}^c} | K_{\mathcal{U},\mathcal{S}})\\
&  = \displaystyle\min_{\mathcal{A} \in \mathbb{A}_t} \min_{\mathcal{U} \in \mathbb{U}_z} [I(K_{\mathcal{L},\mathcal{S}};K_{\mathcal{A},\mathcal{S}} | K_{\mathcal{U},\mathcal{S}})\\
& \phantom{--}+I(K_{\mathcal{L},\mathcal{S}};K_{\mathcal{A},\mathcal{S}^c} K_{\mathcal{U},\mathcal{S}^c} | K_{\mathcal{A},\mathcal{S}} K_{\mathcal{U},\mathcal{S}})]\\
& \stackrel{(e)}\leq \displaystyle\min_{\mathcal{A} \in \mathbb{A}_t} \min_{\mathcal{U} \in \mathbb{U}_z} [I(K_{\mathcal{L},\mathcal{S}};K_{\mathcal{A},\mathcal{S}} | K_{\mathcal{U},\mathcal{S}})\\
& \phantom{--}+I(K_{\mathcal{L},\mathcal{S}} K_{\mathcal{A},\mathcal{S}} K_{\mathcal{U},\mathcal{S}} ;K_{\mathcal{A},\mathcal{S}^c} K_{\mathcal{U},\mathcal{S}^c})]\\
& \stackrel{(f)}= \displaystyle\min_{\mathcal{A} \in \mathbb{A}_t} \min_{\mathcal{U} \in \mathbb{U}_z} I(K_{\mathcal{L},\mathcal{S}};K_{\mathcal{A},\mathcal{S}} | K_{\mathcal{U},\mathcal{S}})\\
& \stackrel{(g)}= \displaystyle\min_{\mathcal{A} \in \mathbb{A}_t} \min_{\mathcal{U} \in \mathbb{U}_z}  H(K_{\mathcal{A},\mathcal{S}} | K_{\mathcal{U},\mathcal{S}}) \\
& = \displaystyle\min_{\mathcal{A} \in \mathbb{A}_t} \min_{\mathcal{U} \in \mathbb{U}_z}  H(K_{\mathcal{A}\backslash \mathcal{U},\mathcal{S}} ) \\
& \stackrel{(h)}= |\mathcal{S}|(t-z) , \numberthis \label{eqconvth}
\end{align*}
where $(a)$ holds by definition of $Y_{\mathcal{S}}^n$, $Y_{\mathcal{S}^c}^n$, $X_{\mathcal{A}}^n$, and $X_{\mathcal{U}}^n$, $(b)$ holds by the chain rule, $(c)$ holds by independence between $(K_{\mathcal{L},\mathcal{S}},K_{\mathcal{U},\mathcal{S}} ,K_{\mathcal{A},\mathcal{S}}  )$ and $(K_{\mathcal{L},\mathcal{S}^c} ,K_{\mathcal{U},\mathcal{S}^c} , K_{\mathcal{A},\mathcal{S}^c})$, $(d)$ and $(e)$ hold by the chain rule, $(f)$ holds by independence between $(K_{\mathcal{L},\mathcal{S}} ,K_{\mathcal{A},\mathcal{S}}, K_{\mathcal{U},\mathcal{S}})$ and $(K_{\mathcal{A},\mathcal{S}^c} ,K_{\mathcal{U},\mathcal{S}^c})$, $(g)$~holds because $K_{\mathcal{L},\mathcal{S}}$ contains $K_{\mathcal{A},\mathcal{S}}$, $(h)$ holds because the minimum is achieved for a choice of $\mathcal{A}$ and $\mathcal{U}$ that minimizes the cardinality of $\mathcal{A}\backslash \mathcal{U}$, which happens when $\mathcal{U} \subseteq \mathcal{A}$, $|\mathcal{A}|$ is as small as possible, i.e., $|\mathcal{A}|=t$, and  $|\mathcal{U}|$ is as large as possible, i.e., $|\mathcal{U}|=z$. Hence, \eqref{eqconvth} and Theorem~\ref{th2ga} proves the converse of Theorem \ref{thm_thr}.

\section{Extension to chosen secrets} \label{secextension}

Note that, similar to a secret-key generation problem, the secrets in the problem statement in Section \ref{sec:pre} are random. In this section, we prove that if, instead the secrets are chosen by the sub-dealers, then our results remain unchanged. We first formalize the problem statement for chosen secrets in Section~\ref{secchosens}. Then, in Section \ref{secchosensp}, we show how the results of Section~\ref{sec:res} for random secrets extend to the setting of Section \ref{secchosens}.

\subsection{Problem statement} \label{secchosens}

We modify Definitions \ref{definition_modelg} and \ref{def} of Section \ref{sec:pre} as follows. Additionally, Figure \ref{fig:secreta} of Section~\ref{sec:pre} now becomes Figure \ref{figext}.

\begin{defn} \label{definition_modelc}
For $d \in \mathcal{D}$, define the alphabet  $\mathcal{S}_d \triangleq \llbracket 1 ,2^{nR_d} \rrbracket$ and  $\mathcal{S}_{\mathcal{D}}\triangleq \bigtimes_{d\in \mathcal{D}} \mathcal{S}_d$. A $( (2^{nR_d})_{d \in \mathcal{D}},\mathbb{A},\mathbb{U},n)$ secret-sharing strategy consists of:
\begin{itemize}
\item A monotone access structure $\mathbb{A}$.
\item $D$ sub-dealers indexed by the set $\mathcal{D}$.
\item $D$ independent secrets $(S_d)_{d \in \mathcal{D}} \in \mathcal{S}_{\mathcal{D}}$, where $S_d$, $d\in\mathcal{D}$, is uniformly distributed over $\mathcal{S}_{d}$ and only known at Sub-dealer $d\in\mathcal{D}$. Moreover, the secrets are assumed to be independent from the source observations. 
\item $L$ participants indexed by the set $\mathcal{L}$.
\item $D$ encoding functions $(f_d)_{d\in\mathcal{D}}$, where $f_d:\mathcal{Y}^n_d \times \mathcal{S}_d \to  \mathcal{M}_d$, $d\in\mathcal{D}$, with $\mathcal{M}_d$ an arbitrary finite alphabet.
\item $|\mathbb{A}|\times D$ decoding functions $(h_{\mathcal{A},d})_{\mathcal{A}\in\mathbb{A}, d\in\mathcal{D}}$, where $h_{\mathcal{A},d}: \mathcal{X}_{\mathcal{A}}^n  \times \mathcal{M}_{\mathcal{D}} \to \mathcal{S}_{d}$, $\mathcal{A} \in \mathbb{A}$ with  $\mathcal{X}_{\mathcal{A}}^n \triangleq \bigtimes_{a\in \mathcal{A}} \mathcal{X}_a^n$ and $\mathcal{M}_{\mathcal{D}} \triangleq \bigtimes_{d\in \mathcal{D}} \mathcal{M}_d$.
\end{itemize}
and operates as follows:
\begin{itemize}
\item Sub-dealer $d \in \mathcal{D}$ observes $Y_d^{n}$.
\item Participant $l\in \mathcal{L}$ observes $X_l^{n}$.
\item Sub-dealer $d \in \mathcal{D}$ sends over a noiseless public authenticated channel the public communication $M_d \triangleq f_d(Y_d^{n},S_d)$ to the participants. We write the global communication of all the sub-dealers as $M_{\mathcal{D}} \triangleq (M_d)_{d\in\mathcal{D}}$. 
 \item Any subset of participants $\mathcal{A} \in \mathbb{A}$ can compute for $d\in\mathcal{D}$, $\widehat{S}_d(\mathcal{A}) \triangleq h_{\mathcal{A},d} (X^n_{\mathcal{A}},M_{\mathcal{D}})$, and thus form
$\widehat{{S}}_{\mathcal{D}} (\mathcal{A})\triangleq(\widehat{S}_d(\mathcal{A}))_{d\in\mathcal{D}}$, an estimate of ${S}_{\mathcal{D}} \triangleq ({S}_d)_{d\in\mathcal{D}}$.
\end{itemize}
\end{defn}

\begin{defn} \label{defc}
A secret rate-tuple $(R_d)_{d\in \mathcal{D}}$ is achievable if there exists a sequence of $((2^{nR_d})_{d \in \mathcal{D}}$,$\mathbb{A}$, $\mathbb{U},n)$ secret-sharing strategies such that
\begin{align}
\lim_{n \to \infty}	\displaystyle\max_{ \mathcal{A} \in \mathbb{A}}    \mathbb{P}\left[\widehat{{S}}_{\mathcal{D}} (\mathcal{A}) \neq {S}_{\mathcal{D}} \right] &= 0 \text{ (Reliability),} \label{eqrelc} \\
\lim_{n \to \infty}	\displaystyle\max_{ \mathcal{U} \in \mathbb{U}}    I\left({S}_{\mathcal{D}} ; M_{\mathcal{D}}, X^n_{\mathcal{U}}  \right) &= 0 \text{ (Strong Security)}\label{eqSecac}.  
\end{align}
Let $\mathcal{C}^{\textup{(chosen)}}(\mathbb{A})$ denote the set of all achievable secret rate-tuples. When $D=1$,  ${C}^{\textup{(chosen)}}(\mathbb{A})$ denotes the supremum   of all achievable secret rates and is called the secret capacity. 
\end{defn}

\begin{figure}[t!]
        \centering
    \includegraphics[width=5.6cm]{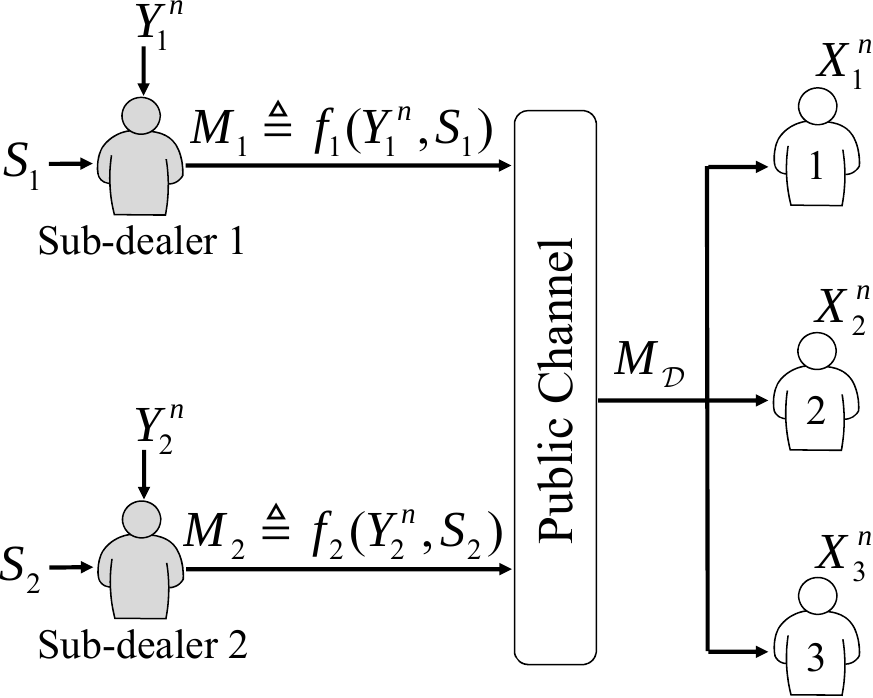}
        \caption{Secret sharing with $D=2$ sub-dealers, $L=3$ users. Formation and distribution of shares.} \label{figext}
\end{figure}
\subsection{Results} \label{secchosensp}

\begin{thm}
Fix $L,D \in \mathbb{N}^*$.
\begin{itemize}
\item For an arbitrary access structure $\mathbb{A} $, 
$$\mathcal{R}^{(\textup{in})}(\mathbb{A})\subseteq \mathcal{C}^{\textup{(chosen)}}(\mathbb{A}) \subseteq \mathcal{R}^{(\textup{out})}(\mathbb{A}),$$
where $\mathcal{R}^{(\textup{in})}(\mathbb{A})$ and $\mathcal{R}^{(\textup{out})}(\mathbb{A})$ are defined in Theorems \ref{th1ga} and \ref{th2ga}.
\item For the all-or-nothing access structure $\mathbb{A}^{\star} \triangleq \{ \mathcal{L}\} $, 
$$\mathcal{R}_1^{(\textup{in})} \subseteq \mathcal{C}^{\textup{(chosen)}}(\mathbb{A}^{\star}) \subseteq \mathcal{R}^{(\textup{out})}(\mathbb{A}^{\star}),$$
where $\mathcal{R}_1^{(\textup{in})}$ and $\mathcal{R}^{(\textup{out})}(\mathbb{A}^{\star})$ are defined in Theorems \ref{prop1} and \ref{th3}. And when $D=2$, if $\displaystyle\min_{d \in \{1,2\}}\displaystyle\min_{ \mathcal{T} \subsetneq \mathcal{L}} I(Y_{d};X_{\mathcal{L}}|X_{\mathcal{T}}) >0$, then we also have
$$\mathcal{R}_1^{(\textup{in})} \subseteq  \mathcal{R}_2^{(\textup{in})} \subseteq \mathcal{C}^{\textup{(chosen)}}(\mathbb{A}^{\star}) \subseteq \mathcal{R}^{(\textup{out})}(\mathbb{A}^{\star}),$$
where $\mathcal{R}_2^{(\textup{in})}$ is defined in Theorem \ref{th2}.
\end{itemize}
\end{thm}

\begin{proof}
The converse proof is obtained by modifying Equation~\eqref{eqreconv1} in Section \ref{secconvproof2} as follows.  For any $\mathcal{T} \subseteq \mathcal{D}$, $\mathcal{A} \in \mathbb{A}$,  $\mathcal{U} \in \mathbb{U}$, we have 
\begin{align*}
nR_{\mathcal{T}} 
& = \log |\mathcal{S}_{\mathcal{T}} |\\
& \stackrel{(a)} = H(S_{\mathcal{T}}) \\
& \stackrel{(b)} \leq H(S_{\mathcal{T}} |M_{\mathcal{D}} X_{\mathcal{U}}^n) + o(n)\\
& \stackrel{(c)} \leq I (S_{\mathcal{T}}; \widehat{S}_{\mathcal{D}}(\mathcal{A}) | M_{\mathcal{D}} X_{\mathcal{U}}^n) + o(n)\\
& \stackrel{(d)}\leq I (S_{\mathcal{T}}; X_{\mathcal{A}}^nM_{\mathcal{D}} | M_{\mathcal{D}} X_{\mathcal{U}}^n) + o(n)\\
& = I (S_{\mathcal{T}}; X_{\mathcal{A}}^n | M_{\mathcal{D}} X_{\mathcal{U}}^n) + o(n)\\
& \leq  I (S_{\mathcal{T}} M_{\mathcal{T}}; X_{\mathcal{A}}^n M_{\mathcal{T}^c}|  X_{\mathcal{U}}^n) + o(n)\\
&  \stackrel{(e)} \leq  I (S_{\mathcal{T}}Y^n_{\mathcal{T}} ; X_{\mathcal{A}}^n S_{\mathcal{T}^c}Y^n_{\mathcal{T}^c}|  X_{\mathcal{U}}^n) + o(n) \\
& \stackrel{(f)} = I (Y^n_{\mathcal{T}} ; X_{\mathcal{A}}^n S_{\mathcal{T}^c}Y^n_{\mathcal{T}^c}|  X_{\mathcal{U}}^n) + o(n) \\
&  \stackrel{(g)} =  I (Y^n_{\mathcal{T}} ; X_{\mathcal{A}}^n Y^n_{\mathcal{T}^c}|  X_{\mathcal{U}}^n) + o(n) \\
& = n I (Y_{\mathcal{T}}; X_{\mathcal{A}} Y_{\mathcal{T}^c}|  X_{\mathcal{U}}) + o(n), 
\end{align*}
where $(a)$ holds by the uniformity of the secrets, $(b)$ holds by \eqref{eqSecac}, $(c)$ holds by Fano's inequality and \eqref{eqrelc}, $(d)$ holds because $\widehat{S}_{\mathcal{D}}(\mathcal{A})$ is a function of $(X_{\mathcal{A}}^n,M_{\mathcal{D}})$, $(e)$ holds because $M_{\mathcal{T}}$ is a function of $(Y^n_{\mathcal{T}},S_{\mathcal{T}})$ for any $\mathcal{T} \subseteq \mathcal{D}$, $(f)$ holds by the chain rule and because $I (S_{\mathcal{T}} ; X_{\mathcal{A}}^n S_{\mathcal{T}^c}Y^n_{\mathcal{T}^c}|  X_{\mathcal{U}}^n Y^n_{\mathcal{T}})=0$, $(g)$ holds by the chain rule and because $I (Y^n_{\mathcal{T}} ;   S_{\mathcal{T}^c} | X_{\mathcal{A}}^n  Y^n_{\mathcal{T}^c} X_{\mathcal{U}}^n) = 0$.

The achievability proof consists in doing a one-time pad on top of the achievability proofs from Section \ref{sec:achievability}. More specifically, suppose that one has generated the secrets $(\tilde{S}_d)_{d\in\mathcal{D}}$ with rate $(R_d)_{d\in\mathcal{D}}$ with the achievability schemes of Section \ref{sec:achievability} such that
\begin{align}
\lim_{n \to \infty}	\displaystyle\max_{ \mathcal{U} \in \mathbb{U}}    I\left(\tilde{S}_{\mathcal{D}} ; \tilde{M}_{\mathcal{D}}, X^n_{\mathcal{U}}  \right)& = 0 ,\label{eqSecaprime} \\
\lim_{n \to \infty} \log |\mathcal{S}_{\mathcal{D}} | - H(\tilde{S}_{\mathcal{D}}) &= 0. \label{equprime}
\end{align}
 Then, Sub-dealer $d\in\mathcal{D}$ transmits over the public channel $\breve{M}_d \triangleq \tilde{S}_d\oplus S_d$  and the security requirement is satisfied because, for any $\mathcal{U} \in \mathbb{U}$ and by defining $\breve{M}_{\mathcal{D}} \triangleq (\breve{M}_d)_{d\in\mathcal{D}}$, we have
\begin{align*}
&I({S}_{\mathcal{D}} ; \tilde{M}_{\mathcal{D}} ,\breve{M}_{\mathcal{D}}, X^n_{\mathcal{U}}  )\\
& = I({S}_{\mathcal{D}} ; \breve{M}_{\mathcal{D}}  ) + I({S}_{\mathcal{D}} ; \tilde{M}_{\mathcal{D}}, X^n_{\mathcal{U}} | \breve{M}_{\mathcal{D}} ) \\
& \stackrel{(a)}\leq  \log |\mathcal{S}_{\mathcal{D}} | - H( \tilde{S}_{\mathcal{D}}) + I({S}_{\mathcal{D}} ; \tilde{M}_{\mathcal{D}}, X^n_{\mathcal{U}} | \breve{M}_{\mathcal{D}} ) \\
& \leq  \log |\mathcal{S}_{\mathcal{D}} | - H( \tilde{S}_{\mathcal{D}}) + I({S}_{\mathcal{D}} \breve{M}_{\mathcal{D}}; \tilde{M}_{\mathcal{D}}, X^n_{\mathcal{U}}   ) \\
& =  \log |\mathcal{S}_{\mathcal{D}} | - H( \tilde{S}_{\mathcal{D}}) + I({S}_{\mathcal{D}} ,\tilde{S}_{\mathcal{D}}; \tilde{M}_{\mathcal{D}}, X^n_{\mathcal{U}}   ) \\
& \stackrel{(b)} =  \log |\mathcal{S}_{\mathcal{D}} | - H( \tilde{S}_{\mathcal{D}}) + I(\tilde{S}_{\mathcal{D}} ; \tilde{M}_{\mathcal{D}}, X^n_{\mathcal{U}}   ) \\
& \xrightarrow{n\to\infty} 0,
\end{align*}
where $(a)$ holds because $H( \breve{M}_{\mathcal{D}}) \leq \log |\mathcal{S}_{\mathcal{D}} |$ and $H( \breve{M}_{\mathcal{D}}| {S}_{\mathcal{D}} ) = H( \tilde{S}_{\mathcal{D}}| {S}_{\mathcal{D}} ) = H( \tilde{S}_{\mathcal{D}})$, $(b)$ holds because $I({S}_{\mathcal{D}} ; \tilde{M}_{\mathcal{D}}, X^n_{\mathcal{U}}   | \tilde{S}_{\mathcal{D}}) \leq I({S}_{\mathcal{D}} ; \tilde{M}_{\mathcal{D}}, X^n_{\mathcal{U}}   , \tilde{S}_{\mathcal{D}}) =0$, and the limit holds by \eqref{eqSecaprime} and~\eqref{equprime}.
\end{proof}

\section{Concluding remarks} \label{sec:concl}

We defined a secret-sharing model between multiple participants and a dealer made of multiple sub-dealers, when each party observes the realizations of correlated random variables and each sub-dealer can communicate with the participants over a   public channel.
Our model extends Shamir's secret-sharing model in three directions. First, it allows a joint design of the creation of the shares and their distribution to the participants. This contrasts with Shamir's model which considers the creation of the shares and their distribution independently. Second, unlike Shamir's model, which assumes that the participants and the dealer  have access to information-theoretically secure channels,  our model   rely on more general resources, namely, a public channel and correlated randomness in the form of realizations of independently and identically distributed random variables.  Third, motivated by a wireless network setting, we explored the problem of secret sharing in a distributed setting where the dealer is an entity made of multiple sub-dealers. 
 
We derived inner and outer regions for the achievable secret rates that the dealer can obtain via its sub-dealers. To this end, we developed two new achievability techniques, a first one to successively handle reliability and security constraints in a distributed setting, and a second one to reduce a distributed setting to multiple single-user settings. We obtained  capacity results in the case of threshold access structures when the correlated randomness corresponds to pairwise secret keys shared between each sub-dealer and each participant, and in the case of a single-dealer setting for the all-or-nothing access structure and arbitrarily correlated randomness.   We highlight that in all our achievabilitiy results the length of each share  always scales linearly with the size of the secret for any access structures.

Note that constructive and low-complexity coding schemes for secret-sharing source model  and channel model have been proposed in the case of a single dealer in \cite{sultana2021} and \cite{chou2020,chou2018explicit,sultana2023}, respectively. While the question of providing  constructive and low-complexity coding schemes for distributed-dealer settings is not addressed in this paper and represents an open challenge, we expect that our proof technique that separates the reliability and security constraints for the all-or-nothing access structure can lead to such a constructive and low-complexity coding scheme for an arbitrary number of sub-dealers.

\appendices

\section{Proof of Lemma \ref{lemappth1a}} \label{appth1a}

By  \cite{kramerbook}, we have
\vspace*{-1em}
\begin{align*}
\mathbb{P} [ \mathcal{E}_0]  \leq 2 |\mathcal{X}_{\mathcal{A}}| |\mathcal{Y}_{\mathcal{D}}| e^{-n \epsilon^2 \mu_{X_{\mathcal{A}}Y_{\mathcal{D}}}} \leq 2 |\mathcal{X}_{\mathcal{L}}| |\mathcal{Y}_{\mathcal{D}}| e^{-n \epsilon^2 \mu_{X_{\mathcal{L}}Y_{\mathcal{D}}}}. 
\end{align*}

Then, for any $\mathcal{S} \subseteq \mathcal{D}$, $\mathcal{S} \neq  \emptyset   $, we have
\begin{align}
\mathbb{P} [\mathcal{E}_{\mathcal{S}} ]   \nonumber 
& = {\sum_{x^n_{\mathcal{A}},y^n_{\mathcal{D}}} } p(x^n_{\mathcal{A}},y^n_{\mathcal{D}}) \mathbb{P} \left[  \forall i \in \mathcal{S}, \exists \hat{y}_{i}^n \neq y_{i}^n,\right.\\ \nonumber
& \left. \phantom{--}  g_i (\hat{y}_i^n) = g_i (y_i^n) \text{ and } ( x^n_{\mathcal{A}},\hat{y}_{\mathcal{S}}^n,y_{  \mathcal{S}^c}^n) \in \mathcal{T}_{\epsilon }^n (X_{\mathcal{A}}  Y_{\mathcal{D}})   \right]\\ \nonumber
& \stackrel{(a)} \leq  {\sum_{x^n_{\mathcal{A}},y^n_{\mathcal{D}}} } p(x^n_{\mathcal{A}},y^n_{\mathcal{D}}) \\
& \phantom{--}\times \sum_{ \stackrel{\hat{y}_{\mathcal{S}}^n \in \mathcal{T}_{\epsilon }^n (X_{\mathcal{A}}  Y_{\mathcal{D}}|x^n_{\mathcal{A}} y_{ \mathcal{S}^c}^n)}{\hat{y}_{\mathcal{S}}^n \neq y_{\mathcal{S}}^n}} \mathbb{P} \left[  \forall i \in \mathcal{S},  g_i (\hat{y}_i^n) = g_i (y_i^n) 
 \right] \nonumber \\ \nonumber
 &\stackrel{(b)} =  {\sum_{x^n_{\mathcal{A}},y^n_{\mathcal{D}}} } p(x^n_{\mathcal{A}},y^n_{\mathcal{D}}) \sum_{ \stackrel{\hat{y}_{\mathcal{S}}^n \in \mathcal{T}_{\epsilon }^n (X_{\mathcal{A}}  Y_{\mathcal{D}}|x^n_{\mathcal{A}} y_{ \mathcal{S}^c}^n)}{ \hat{y}_{\mathcal{S}}^n \neq y_{\mathcal{S}}^n}} 2^{-nR'_{\mathcal{S}}}\\ \nonumber
& \leq {\sum_{x^n_{\mathcal{A}},y^n_{\mathcal{D}}} } p(x^n_{\mathcal{A}},y^n_{\mathcal{D}}) |\mathcal{T}_{\epsilon }^n (X_{\mathcal{A}}  Y_{\mathcal{D}}|x^n_{\mathcal{A}} y_{ \mathcal{S}^c}^n)| 2^{-nR'_{\mathcal{S}}}\\ \nonumber
& \stackrel{(c)}\leq {\sum_{x^n_{\mathcal{A}},y^n_{\mathcal{D}}} } p(x^n_{\mathcal{A}},y^n_{\mathcal{D}}) 2^{  n(1 + \epsilon ) {H}(Y_{\mathcal{S}}|Y_{ \mathcal{S}^c } X_{\mathcal{A}})} 2^{-nR'_{\mathcal{S}}} \\ \nonumber
& = 2^{  n(1 + \epsilon )  {H}(Y_{\mathcal{S}}|Y_{ \mathcal{S}^c } X_{\mathcal{A}}) -nR'_{\mathcal{S}} } \\
& \leq 2^{  n(1 + \epsilon ) \max_{\mathcal{A} \in \mathbb{A}} {H}(Y_{\mathcal{S}}|Y_{ \mathcal{S}^c } X_{\mathcal{A}}) -nR'_{\mathcal{S}} }, \nonumber
\end{align}
where in $(a)$ $\hat{y}_{\mathcal{S}}^n \neq y_{\mathcal{S}}^n$ means $\hat{y}_{i}^n \neq y_{i}^n, \forall i \in \mathcal{S}$, $(b)$ holds by independence of the random binning choices across the sub-dealers, $(c)$ holds by~\cite{kramerbook}.

\section{Proof of Lemma \ref{lemappth1b}} \label{appth1b}
We first bound the second term in (\ref{twotermsa}) as follows
\begin{align*}
&\mathbb{E} \left[ \sum_{m_{\mathcal{D}},s_{\mathcal{D}},x_{\mathcal{U}}^n} \left| p^{(2)}_{M_{\mathcal{D}} S_{\mathcal{D}} X_{\mathcal{U}}^n }(m_{\mathcal{D}},s_{\mathcal{D}},x_{\mathcal{U}}^n) \right. \right.\\
& \phantom{--------}\left. \left. -  \mathbb{E} \left[ p^{(2)}_{M_{\mathcal{D}} S_{\mathcal{D}} X_{\mathcal{U}}^n }(m_{\mathcal{D}},s_{\mathcal{D}},x_{\mathcal{U}}^n) \right] \right| \right] \nonumber  \displaybreak[0]\\
& \stackrel{(a)} \leq  \sum_{m_{\mathcal{D}},s_{\mathcal{D}},x_{\mathcal{U}}^n}  2  \mathbb{E} \left[ p^{(2)}_{M_{\mathcal{D}} S_{\mathcal{D}} X_{\mathcal{U}}^n }(m_{\mathcal{D}},s_{\mathcal{D}},x_{\mathcal{U}}^n)    \right] \\ \nonumber
& = 2  \sum_{m_{\mathcal{D}},s_{\mathcal{D}},x_{\mathcal{U}}^n}    \sum_{y_{\mathcal D}^n \notin \mathcal{T}_{\epsilon}^n({Y}_{\mathcal{D}}  {X}_{\mathcal{U}}|x_{\mathcal{U}}^n)} p(y_{\mathcal{D}}^n,x_{\mathcal{U}}^n) 2^{-n(R_{\mathcal{D}}+R_{\mathcal{D}}')} \\ \nonumber 
& = 2  \mathbb{P}\left[ ( Y_{\mathcal{D}}^n, X_{\mathcal{U}}^n ) \notin \mathcal{T}_{\epsilon}^n({Y}_{\mathcal{D}} X_{\mathcal{U}}) \right]\\
& \stackrel{(b)} \leq 2  |\mathcal{Y}_{\mathcal{D}}| |\mathcal{X}_{\mathcal{U}}| e^{-n\epsilon^2 \mu_{Y_{\mathcal{D}} X_{\mathcal{U}}}} \\
& \leq 2  |\mathcal{Y}_{\mathcal{D}}| |\mathcal{X}_{\mathcal{L}}| e^{-n\epsilon^2 \mu_{Y_{\mathcal{D}} X_{\mathcal{L}}}} ,
\end{align*}
where $(a)$ holds by the triangle inequality and $(b)$ holds 
by~\cite{kramerbook}.

We now upper-bound the first term in~(\ref{twotermsa}) using Jensen's inequality by
\begin{align}
  \sum_{m_{\mathcal{D}},s_{\mathcal{D}},x_{\mathcal{U}}^n}  \sqrt{ \textup{Var}  \left(p^{(1)}_{M_{\mathcal{D}} S_{\mathcal{D}} X_{\mathcal{U}}^n }(m_{\mathcal{D}},s_{\mathcal{D}},x_{\mathcal{U}}^n) \right)}, \label{eqjensena}
\end{align}
and upper-bound the variance in \eqref{eqjensena} as follows
\begin{align*}
&\textup{Var}  \left( p^{(1)}_{M_{\mathcal{D}} S_{\mathcal{D}} X_{\mathcal{U}}^n }(m_{\mathcal{D}},s_{\mathcal{D}},x_{\mathcal{U}}^n) \right) \nonumber \\ \nonumber
& = \mathbb{E} \left[  \left( p^{(1)}_{M_{\mathcal{D}} S_{\mathcal{D}} X_{\mathcal{U}}^n }(m_{\mathcal{D}},s_{\mathcal{D}},x_{\mathcal{U}}^n) \right)^2 \right] \\
& \phantom{-} - \mathbb{E} \left[  \left( p^{(1)}_{M_{\mathcal{D}} S_{\mathcal{D}} X_{\mathcal{U}}^n }(m_{\mathcal{D}},s_{\mathcal{D}},x_{\mathcal{U}}^n) \right) \right]^2 \nonumber \\ \nonumber
& \stackrel{(a)} \leq  {\sum_{\mathcal{S} \subsetneq \mathcal{D}} \sum_{ y_{\mathcal{D}}^{n}} \!\!\! \!\!\sum_{ \substack{y_{\mathcal D}^{n\prime} \\ \text{s.t.} y_{\mathcal S}^{n\prime} \neq y_{\mathcal S}^{n} \\ y_{\mathcal{S}^c}^{n\prime} = y_{\mathcal {S}^c}^{n}}}} \mathds{1} \{  y_{\mathcal D}^n \in \mathcal{T}_{\epsilon}^n({Y}_{\mathcal{D}}  {X}_{\mathcal{U}}|x_{\mathcal{U}}^n) \}  \\ \nonumber
&  \phantom{====} \times \mathds{1} \{ y_{\mathcal D}^{n\prime} \in \mathcal{T}_{\epsilon}^n({Y}_{\mathcal{D}}  {X}_{\mathcal{U}}|x_{\mathcal{U}}^n) \}  \times p(y_{\mathcal{D}}^n,x_{\mathcal{U}}^n) p(y_{\mathcal{D}}^{n\prime},x_{\mathcal{U}}^n)    \\ \nonumber
&  \phantom{====} \times  \! \mathbb{E} \!   \left[  \textstyle\prod_{i\in \mathcal{D}} \mathds{1} \{g_i(y_i^n) = m_{i} \}  \mathds{1} \{h_i(y_i^{n}) = s_{i} \}   \right.\\ \nonumber
&  \phantom{=======} \left. \times \textstyle\prod_{i\in \mathcal{D}} \mathds{1} \{g_i(y_i^{n\prime}) = m_{i} \}   \mathds{1} \{h_i(y_i^{n\prime}) = s_{i} \}  \!\right]   \\ \nonumber
& =  \sum_{\mathcal{S} \subsetneq \mathcal{D}} \sum_{ y_{\mathcal{D}}^{n}} \!\!\! \!\! \sum_{ \substack{y_{\mathcal D}^{n\prime} \\ \text{s.t.} y_{\mathcal S}^{n\prime} \neq y_{\mathcal S}^{n} \\ y_{\mathcal{S}^c}^{n\prime} = y_{\mathcal {S}^c}^{n}}} \mathds{1} \{  y_{\mathcal D}^n \in \mathcal{T}_{\epsilon}^n({Y}_{\mathcal{D}}  {X}_{\mathcal{U}}|x_{\mathcal{U}}^n) \} \\ \nonumber
&  \phantom{====} \times \mathds{1} \{ y_{\mathcal D}^{n\prime} \in \mathcal{T}_{\epsilon}^n({Y}_{\mathcal{D}}  {X}_{\mathcal{U}}|x_{\mathcal{U}}^n) \}      \\ \nonumber
&     \phantom{====}  \times p^2(x_{\mathcal{U}}^n) p^2(y_{\mathcal{S}^c}^n|x_{\mathcal{U}}^n) p(y_{\mathcal{S}}^n|x_{\mathcal{U}}^ny_{\mathcal{S}^c}^n)  p(y_{\mathcal{S}}^{n\prime}|x_{\mathcal{U}}^n y_{\mathcal{S}^c}^{n }) \\ \nonumber
&     \phantom{====}  \times   2^{-n(2R_{\mathcal{S}}+2R'_{\mathcal{S}}+R_{\mathcal{S}^c}+R_{\mathcal{S}^c}')} \\ \nonumber
& \stackrel{(b)}\leq   \sum_{\mathcal{S} \subsetneq \mathcal{D}} p^2(x_{\mathcal{U}}^n) 2^{-n(1-\epsilon)H(Y_{\mathcal{S}^c}|X_{\mathcal{U}})}  2^{-n(2R_{\mathcal{S}}+2R'_{\mathcal{S}}+R_{\mathcal{S}^c}+R_{\mathcal{S}^c}')}    ,\numberthis
  \label{eqjensensuitea}  
\end{align*}
where in $(a)$ the notation $y_{\mathcal S}^{n\prime} \neq y_{\mathcal S}^{n}$ means $y_{i}^{n\prime} \neq y_{i}^{n}, \forall i \in \mathcal S$, and $(b)$ holds by \cite{kramerbook}.
Hence, by \eqref{eqjensena} and \eqref{eqjensensuitea}, we upper-bound the first term in (\ref{twotermsa}) by
\begin{align*}
 &   \sum_{m_{\mathcal{D}},s_{\mathcal{D}},x_{\mathcal{U}}^n}  \sqrt{ \textup{Var}  \left(p^{(1)}_{M_{\mathcal{D}} S_{\mathcal{D}} X_{\mathcal{U}}^n }(m_{\mathcal{D}},s_{\mathcal{D}},x_{\mathcal{U}}^n) \right)} \\
  &\leq \smash{\sum_{\mathcal{S} \subsetneq \mathcal{D}}   2^{n (R_{\mathcal{D}}+R'_{\mathcal{D}})}  2^{-\tfrac{n}{2}(1-\epsilon)H(Y_{\mathcal{S}^c}|X_{\mathcal{U}})}} \\
  & \phantom{----}\times 2^{-\tfrac{n}{2}(2R_{\mathcal{S}}+2R'_{\mathcal{S}}+R_{\mathcal{S}^c}+R_{\mathcal{S}^c}')} \\
    &= \sum_{\mathcal{S} \subsetneq \mathcal{D}}     2^{-\tfrac{n}{2}(1-\epsilon)H(Y_{\mathcal{S}^c}|X_{\mathcal{U}})}  2^{\tfrac{n}{2}(R_{\mathcal{S}^c}+R_{\mathcal{S}^c}')} \\
        &= \sum_{\mathcal{S} \subseteq \mathcal{D}, \mathcal{S}\neq \emptyset}     2^{-\tfrac{n}{2}(1-\epsilon)H(Y_{\mathcal{S}}|X_{\mathcal{U}})}  2^{\tfrac{n}{2}(R_{\mathcal{S}}+R_{\mathcal{S}}')} \\
                & \leq \sum_{\mathcal{S} \subseteq \mathcal{D}, \mathcal{S}\neq \emptyset}     2^{-\tfrac{n}{2}(1-\epsilon)\min_{\mathcal{U} \in \mathbb{U}}H(Y_{\mathcal{S}}|X_{\mathcal{U}})}  2^{\tfrac{n}{2}(R_{\mathcal{S}}+R_{\mathcal{S}}')}. 
\end{align*}

\section{Proof of Lemma \ref{lemloh}} \label{app_loh}
For $p_{X_{\mathcal{L}} X_{\mathcal{L}}'F_{\mathcal{L}}F_{\mathcal{L}}'} = p^2_{X_{\mathcal{L}}} p^2_{F_{\mathcal{L}}}$, we have
\begin{align*}
& \sum_{m_{\mathcal{L}},f_{\mathcal{L}}}  {p}^2_{F_{\mathcal{L}}( X_{\mathcal{L}})   F_{\mathcal{L}} }(m_{\mathcal{L}} ,f_{\mathcal{L}} )\\
& = \sum_{ f_{\mathcal{L}}} {p}^2_{F_{\mathcal{L}}}(f_{\mathcal{L}}) \sum_{m_{\mathcal{L}} } {p}^2_{F_{\mathcal{L}}( X_{\mathcal{L}})   F_{\mathcal{L}} }(m_{\mathcal{L}} |f_{\mathcal{L}} )\\
& = \sum_{ f_{\mathcal{L}}} {p}^2_{F_{\mathcal{L}}}(f_{\mathcal{L}})  \sum_{x_{\mathcal{L}},x_{\mathcal{L}}'} p_{X_{\mathcal{L}} X_{\mathcal{L}}'}(x_{\mathcal{L}},x_{\mathcal{L}}') \mathds{1} \{ f_{\mathcal{L}}(x_{\mathcal{L}}) = f_{\mathcal{L}}(x_{\mathcal{L}}') \}\\
& =  s^{-1}_{\mathcal{L}}  \sum_{x_{\mathcal{L}},x_{\mathcal{L}}'} p_{X_{\mathcal{L}} X_{\mathcal{L}}'}(x_{\mathcal{L}},x_{\mathcal{L}}') \mathbb{P} [ F_{\mathcal{L}}(x_{\mathcal{L}}) = F_{\mathcal{L}}(x_{\mathcal{L}}') ]\\
&\stackrel{(a)} = s^{-1}_{\mathcal{L}} \smash{\sum_{\mathcal{S} \subseteq \mathcal{L}}\sum_{x_{\mathcal{L}}}  {\sum_{\substack{x_{\mathcal{L}}' \\ \!\!\!\!\!\!\text{s.t.} x'_{\mathcal{S}} \neq x_{\mathcal{S}} \\  x'_{\mathcal{S}^c} = x_{\mathcal{S}^c} }}} \mathbb{P} [ F_{\mathcal{L}}(x_{\mathcal{L}}) = F_{\mathcal{L}}(x_{\mathcal{L}}')] } \\
& \phantom{-------------}\times p_{X_{\mathcal{L}}}(x_{\mathcal{L}}) p_{X_{\mathcal{L}}}(x'_{\mathcal{L}})\\~\\
& = s^{-1}_{\mathcal{L}} \smash{\sum_{\mathcal{S} \subseteq \mathcal{L}}\sum_{x_{\mathcal{L}}}  \sum_{\substack{x_{\mathcal{L}}' \\ \!\!\!\!\!\!\text{s.t.} x'_{\mathcal{S}} \neq x_{\mathcal{S}} \\  x'_{\mathcal{S}^c} = x_{\mathcal{S}^c} }} \prod_{l\in \mathcal{L}} \mathbb{P} [ F_{l}(x_{l}) = F_{l}(x_{l}')] } \\
& \phantom{-------------}\times p_{X_{\mathcal{L}}}(x_{\mathcal{L}})    p_{X_{\mathcal{L}}}( x_{\mathcal{S}}', x_{\mathcal{S}^c}) \\ ~\\
&\stackrel{(b)} \leq s^{-1}_{\mathcal{L}}  \sum_{\mathcal{S} \subseteq \mathcal{L}}\sum_{x_{\mathcal{L}}} { \sum_{\substack{x_{\mathcal{L}}' \\ \!\!\!\!\!\!\text{s.t.} x'_{\mathcal{S}} \neq x_{\mathcal{S}} \\ x'_{\mathcal{S}^c} = x_{\mathcal{S}^c} }}}  2^{-r_{\mathcal{S}}} p_{X_{\mathcal{L}}}(x_{\mathcal{L}})   p_{X_{\mathcal{L}}}( x_{\mathcal{S}}', x_{\mathcal{S}^c})  \\ 
& \stackrel{(c)} \leq s^{-1}_{\mathcal{L}}  \sum_{\mathcal{S} \subseteq \mathcal{L}}\sum_{x_{\mathcal{L}}}   2^{-r_{\mathcal{S}}}p_{X_{\mathcal{L}}}(x_{\mathcal{L}}) p_{X_{\mathcal{S}^c}}(x_{\mathcal{S}^c})  \\ 
& \stackrel{(d)} \leq s^{-1}_{\mathcal{L}}  \sum_{\mathcal{S} \subseteq \mathcal{L}}\sum_{x_{\mathcal{L}}}   2^{-r_{\mathcal{S}}} p_{X_{\mathcal{L}}}(x_{\mathcal{L}}) 2^{-H_{\infty}(p_{X_{\mathcal{S}^c}})}  \\ 
& \leq s^{-1}_{\mathcal{L}}  \sum_{\mathcal{S} \subseteq \mathcal{L}}   2^{-r_{\mathcal{S}}-H_{\infty}(p_{X_{\mathcal{S}^c}})}, \numberthis \label{eqloh}
\end{align*}
where  in $(a)$ the notation $x'_{\mathcal{S}} \neq x_{\mathcal{S}}$ means $x'_{i} \neq x_{i}, \forall i \in \mathcal S$, $(b)$ holds by the two-universality of the $F_l$'s, $l \in \mathcal{L}$, $(c)$ holds by marginalization over $X'_{\mathcal{S}}$, $(d)$~holds by definition of the min-entropy.

Next, consider  $q_Z$ defined over $\mathcal{Z}$ such that $\textup{supp}(q_Z) \subseteq \textup{supp}(p_Z)$. We have \eqref{eqprooflohl}, 
\begin{figure*}
 	\begin{align*}
      & \mathbb{V}({{p}_{F_{\mathcal{L}}( X_{\mathcal{L}}),   F_{\mathcal{L}},Z}}, p_{U_{\mathcal K}} p_{U_{\mathcal F}} p_Z) \\
      & =  \sum_{m_{\mathcal{L}},f_{\mathcal{L}},z} \left( q_Z^{1/2}(z)\right) \left(q_Z^{-1/2}(z)| {p}_{F_{\mathcal{L}}( X_{\mathcal{L}}),   F_{\mathcal{L}},Z}(m_{\mathcal{L}},f_{\mathcal{L}},z) - p_{U_{\mathcal K}} (m_{\mathcal{L}}) p_{U_{\mathcal F}} (f_{\mathcal{L}}) p_Z (z) | \right)   \\
      & \stackrel{(a)}\leq      \sqrt{ \left(\sum_{m_{\mathcal{L}},f_{\mathcal{L}},z} q_Z^{-1}(z) \left( {p}_{F_{\mathcal{L}}( X_{\mathcal{L}})   F_{\mathcal{L}}, Z}(m_{\mathcal{L}} ,f_{\mathcal{L}} ,z) - p_{U_{\mathcal K}} (m_{\mathcal{L}}) p_{U_{\mathcal F}} (f_{\mathcal{L}}) p_Z (z) \right)^2 \right) \left( \sum_{m_{\mathcal{L}},f_{\mathcal{L}},z}  q_Z(z)\right)} \\
                        & =      \sqrt{s_{\mathcal{L}} 2^{r_{\mathcal{L}}} \sum_{z}  q_Z^{-1}(z)  \sum_{m_{\mathcal{L}},f_{\mathcal{L}}}  \left( {p}_{F_{\mathcal{L}}( X_{\mathcal{L}})   F_{\mathcal{L}}, Z}(m_{\mathcal{L}} ,f_{\mathcal{L}} ,z) - \frac{p_Z (z)}{s_{\mathcal{L}} 2^{r_{\mathcal{L}}}} \right)^2 } \\
       & = \sqrt{ s_{\mathcal{L}} 2^{r_{\mathcal{L}}} \!\sum_{z}  q_Z^{-1}(z)  \!\!\sum_{m_{\mathcal{L}},f_{\mathcal{L}}}\!\! \left( {p}^2_{F_{\mathcal{L}}( X_{\mathcal{L}})   F_{\mathcal{L}}, Z}(m_{\mathcal{L}}, f_{\mathcal{L}} ,z) \!-\!2 \frac{p_Z (z)}{s_{\mathcal{L}} 2^{r_{\mathcal{L}}}}  {p}_{F_{\mathcal{L}}( X_{\mathcal{L}})   F_{\mathcal{L}}, Z}(m_{\mathcal{L}}, f_{\mathcal{L}} ,z) \!+\! \frac{p^2_Z (z)}{s^2_{\mathcal{L}} 2^{2r_{\mathcal{L}}}} \!\right) }  \\
       & = \sqrt{  \sum_{z}  q_Z^{-1}(z) \left( s_{\mathcal{L}} 2^{r_{\mathcal{L}}}\sum_{m_{\mathcal{L}},f_{\mathcal{L}}}  {p}^2_{F_{\mathcal{L}}( X_{\mathcal{L}})   F_{\mathcal{L}}, Z}(m_{\mathcal{L}} ,f_{\mathcal{L}} ,z) -p^2_Z (z) \right)  }  \\
       & = \sqrt{  \sum_{z} p^2_Z (z) q_Z^{-1}(z) \left( s_{\mathcal{L}} 2^{r_{\mathcal{L}}}\sum_{m_{\mathcal{L}},f_{\mathcal{L}}}  {p}^2_{F_{\mathcal{L}}( X_{\mathcal{L}})   F_{\mathcal{L}} | Z}(m_{\mathcal{L}} ,f_{\mathcal{L}} |z) -1 \right)  }   \\   
              & \stackrel{(b)}= \sqrt{  \sum_{z} p^2_Z (z) q_Z^{-1}(z) \left( s_{\mathcal{L}} 2^{r_{\mathcal{L}}}   \sum_{m_{\mathcal{L}},f_{\mathcal{L}}}  {p}^2_{F_{\mathcal{L}}( X^{(z)}_{\mathcal{L}})   F_{\mathcal{L}} }(m_{\mathcal{L}} ,f_{\mathcal{L}} )    -   1 \right)  }   \\  
                    & \stackrel{(c)}\leq \sqrt{  \sum_{z} p^2_Z (z) q_Z^{-1}(z) \left( 2^{r_{\mathcal{L}}}  \sum_{\mathcal{S} \subseteq \mathcal{L}}   2^{-r_{\mathcal{S}}-H_{\infty}(p_{X^{(z)}_{\mathcal{S}^c}})} -   1 \right)  }   \\  
    & = \sqrt{  \sum_{z} p_Z (z) 2^{\log \frac{p_Z (z)}{q_Z (z)}} \left(   \sum_{\mathcal{S} \subseteq \mathcal{L}}   2^{r_{\mathcal{S}^c}-H_{\infty}(p_{X^{(z)}_{\mathcal{S}^c}})} -   1 \right)  }   \\   
   & = \sqrt{  \sum_{z} p_Z (z)   \sum_{\mathcal{S} \subseteq \mathcal{L}, \mathcal{S} \neq \emptyset}   2^{r_{\mathcal{S}}-H_{\infty}(p_{X^{(z)}_{\mathcal{S}}}) +\log \frac{p_Z (z)}{q_Z (z)}}   }   \\
    &\leq \sqrt{  \sum_{z} p_Z (z)    \sum_{\mathcal{S} \subseteq \mathcal{L}, \mathcal{S} \neq \emptyset}   2^{r_{\mathcal{S}}-H_{\infty}(p_{X_{\mathcal{S}}Z} | q_Z) }   }   \\  
        &= \sqrt{  \sum_{\mathcal{S} \subseteq \mathcal{L}, \mathcal{S} \neq \emptyset}   2^{r_{\mathcal{S}}-H_{\infty}(p_{X_{\mathcal{S}}Z} | q_Z) }    } \numberthis \label{eqprooflohl}
    \end{align*}
    \hrulefill                                                                                                                
\end{figure*}
where $(a)$ holds by Cauchy-Schwarz inequality, in $(b)$ we have defined for $z \in \mathcal{Z}$, $X_{\mathcal{L}}^{(z)}$ distributed according to $p_{X_{\mathcal{L}}^{(z)}} = p_{X_{\mathcal{L}}|Z=z}$, $(c)$ holds by \eqref{eqloh}.
 \section{Proof of Lemma \ref{lems1}} \label{App_lemma2}

       \begin{proof}
For any $z^n \in \mathcal{Z}^n$ such that $q_{Z^n}(z^n)>0$, define
\begin{align*}
&\mathcal{A}(z^n) \triangleq \left\{ y^n_{\mathcal{D}} \in \mathcal{Y}^n_{\mathcal{D}} : \right.\\
& \left. \phantom{-}  -\log q_{Y^n_{\mathcal{S}}|Z^n} (y^n_{\mathcal{S}}|z^n)  \geq H(Y^n_{\mathcal{S}}|Z^n) - n \delta_{\mathcal{S}}(n), \forall \mathcal{S} \subseteq \mathcal{D} \right\}\!,
\end{align*}
and for $\mathcal{S} \subseteq \mathcal{D}$,
\begin{align*} 
&\mathcal{A}_{\mathcal{S}}(z^n) \triangleq \left\{ y^n_{\mathcal{S}} \in \mathcal{Y}^n_{\mathcal{S}} : \right.\\
& \left. \phantom{---} -\log q_{Y^n_{\mathcal{S}}|Z^n} (y^n_{\mathcal{S}}|z^n)  \geq H(Y^n_{\mathcal{S}}|Z^n) - n \delta_{\mathcal{S}}(n)  \right\}.
\end{align*}
Define for $(y^n_{\mathcal{D}},z^n) \in \mathcal{Y}^n_{\mathcal{D}} \times \mathcal{Z}^n$,
\begin{align} \label{equnn}
w_{Y^n_{\mathcal{D}}Z^n}(y^n_{\mathcal{D}},z^n) \triangleq \mathds{1} \{ y^n_{\mathcal{D}} \in \mathcal{A}(z^n)  \} q_{Y^n_{\mathcal{D}}Z^n}(y^n_{\mathcal{D}},z^n),
\end{align}
and for $\mathcal{S} \subseteq \mathcal{D}$,
\begin{align} \label{eqmarg}
w_{Y^n_{\mathcal{S}}Z^n}(y^n_{\mathcal{S}},z^n) \triangleq \sum_{y^n_{\mathcal{S}^c} \in \mathcal{Y}^n_{\mathcal{S}^c} } w_{Y^n_{\mathcal{D}}Z^n}(y^n_{\mathcal{D}},z^n).
\end{align}
We first show that $\mathbb{V}(p_{Y^n_{\mathcal{D}}Z^n},w_{Y^n_{\mathcal{D}}Z^n}) \leq \epsilon$. We have
\begin{align*}
& \mathbb{V}(q_{Y^n_{\mathcal{D}}Z^n},w_{Y^n_{\mathcal{D}}Z^n}) \\
& = \sum_{y^n_{\mathcal{D}},z^n} |q_{Y^n_{\mathcal{D}}Z^n}(y^n_{\mathcal{D}},z^n)- w_{Y^n_{\mathcal{D}}Z^n}(y^n_{\mathcal{D}},z^n)| \\
& = \sum_{y^n_{\mathcal{D}},z^n} q_{Y^n_{\mathcal{D}}Z^n}(y^n_{\mathcal{D}},z^n) \mathds{1} \{ y^n_{\mathcal{D}} \notin  \mathcal{A}(z^n) \} \\
& = \mathbb{P} \left[ Y^n_{\mathcal{D}} \notin  \mathcal{A}(Z^n) \right]\\
& \stackrel{(a)}{\leq} \sum_{ \substack{\mathcal{S} \subseteq \mathcal{D} }}\mathbb{P} \left[  Y^n_{\mathcal{S}} \notin  \mathcal{A}_{\mathcal{S}}(Z^n) \right]\\
& \stackrel{(b)}{\leq} \sum_{ \substack{\mathcal{S} \subseteq \mathcal{D} }} 2^{-\frac{n \delta_{\mathcal{S}}^2(n)}{2\log(|\mathcal{Y}_{\mathcal{S}}| +3)^2}}\\
& = \sum_{ \substack{\mathcal{S} \subseteq \mathcal{D} }} 2^{-D}\epsilon\\
& \leq \epsilon,
\end{align*}
where $(a)$ holds by the union bound, $(b)$ holds by \cite[Theorem 3.3.3]{renner2008security}.

Next, for $\mathcal{S} \subseteq \mathcal{D}$, we have
\begin{align*}
& H_{\infty}(w_{Y^n_{\mathcal{S}}Z^n}|q_{Z^n}) \\
&= - \max_{z^n \in \textup{supp}(q_{Z^n})} \max_{y^n_{\mathcal{S}} \in \mathcal{Y}^n_{\mathcal{S}}} \log \frac{w_{Y^n_{\mathcal{S}}Z^n}(y^n_{\mathcal{S}},z^n)}{q_{Z^n}(z^n)} \\
& \stackrel{(a)}{=} - \max_{z^n} \max_{y^n_{\mathcal{S}} } \log \frac{ \displaystyle\sum_{y^n_{\mathcal{S}^c} }  \mathds{1} \{ y^n_{\mathcal{D}} \in \mathcal{A}(z^n)  \} q_{Y^n_{\mathcal{D}}Z^n}(y^n_{\mathcal{D}},z^n)}{q_{Z^n}(z^n)} \\
& \stackrel{(b)}{\geq} - \max_{z^n  } \max_{y^n_{\mathcal{S}} } \log \frac{  \mathds{1} \{ y^n_{\mathcal{S}} \in \mathcal{A}_{\mathcal{S}}(z^n)  \} q_{Y^n_{\mathcal{S}}Z^n}(y^n_{\mathcal{S}},z^n)}{q_{Z^n}(z^n)}\\
& \stackrel{(c)}{\geq} H(Y^n_{\mathcal{S}}|Z^n) - n \delta_{\mathcal{S}}(n) ,
\end{align*}
where the first maximum in $(a)$ and $(b)$ is over $\textup{supp}(q_{Z^n})$, $(a)$ holds by \eqref{equnn} and \eqref{eqmarg}, $(b)$ holds because for any $y^n_{\mathcal{D}} \in \mathcal{Y}^n_{\mathcal{D}}$, $ \mathds{1} \{ y^n_{\mathcal{S}} \in \mathcal{A}_{\mathcal{S}}(z^n)  \} \geq \mathds{1} \{ y^n_{\mathcal{D}} \in \mathcal{A}(z^n)  \}$ and by marginalization over $Y^n_{\mathcal{S}^c}$, $(c)$ holds by definition of $\mathcal{A}_{\mathcal{S}}(z^n)$.
\end{proof}

 \section{Proof of Lemma \ref{lemamp}} \label{Applemamp}
Let $\mathcal{U} \in \mathbb{U}$. By Lemma \ref{lems1}, for any $\epsilon>0$, there exists a subnormalized non-negative function $w_{ Y^{nB}_{\mathcal{D}}  M^B_{\mathcal{D}} X^{nB}_{\mathcal{U}} }$  such that  \begin{align} &\label{eqepsi}\mathbb{V}(p_{ Y^{nB}_{\mathcal{D}}  M^B_{\mathcal{D}} X^{nB}_{\mathcal{U}} },w_{ Y^{nB}_{\mathcal{D}}  M^B_{\mathcal{D}} X^{nB}_{\mathcal{U}} })\leq\epsilon, \\
 &          \forall \mathcal{S}\subseteq \mathcal{D}, H_{\infty}(w_{ Y^{nB}_{\mathcal{S}}  M^B_{\mathcal{D}} X^{nB}_{\mathcal{U}} }|p_{ M^B_{\mathcal{D}} X^{nB}_{\mathcal{U}} }) \nonumber \\
           &\phantom{-----}\geq B H(Y^{n}_{\mathcal{S}}| M_{\mathcal{D}} X^{n}_{\mathcal{U}})-B \delta_{\mathcal{S}}(n,B). \label{eqepsi2}
       \end{align}
     Next, we have
\begin{align*}
&\mathbb{V} ( p_{F_{\mathcal{D}}(Y^{nB}_{\mathcal{D}})F_{\mathcal{D}} M_{\mathcal{D}}^B X^{nB}_{\mathcal{U}} }, p_{U_{\mathcal{D}}} p_{U_{\mathcal{F}}}p_{M^B_{\mathcal{D}}X^{nB}_{\mathcal{U}}  })    \\
&\stackrel{(a)}\leq \mathbb{V} ( p_{F_{\mathcal{D}}(Y^{nB}_{\mathcal{D}})F_{\mathcal{D}} M_{\mathcal{D}}^B X^{nB}_{\mathcal{U}} },w_{F_{\mathcal{D}}(Y^{nB}_{\mathcal{D}})F_{\mathcal{D}} M^B_{\mathcal{D}} X^{nB}_{\mathcal{U}} }) \\
& \phantom{--} +  \mathbb{V} (w_{F_{\mathcal{D}}(Y^{nB}_{\mathcal{D}})F_{\mathcal{D}} M^B_{\mathcal{D}} X^{nB}_{\mathcal{U}} },  p_{U_{\mathcal{D}}} p_{U_{\mathcal{F}}}w_{M^B_{\mathcal{D}}X^{nB}_{\mathcal{U}}  }) \\
& \phantom{--} +  \mathbb{V} ( p_{U_{\mathcal{D}}} p_{U_{\mathcal{F}}}w_{M^B_{\mathcal{D}}X^{nB}_{\mathcal{U}}  }, p_{U_{\mathcal{D}}} p_{U_{\mathcal{F}}}p_{M^B_{\mathcal{D}}X^{nB}_{\mathcal{U}}  })\\
& \stackrel{(b)}\leq \mathbb{V} ( p_{Y^{nB}_{\mathcal{D}}  M_{\mathcal{D}}^B X^{nB}_{\mathcal{U}} },w_{ Y^{nB}_{\mathcal{D}}  M^B_{\mathcal{D}} X^{nB}_{\mathcal{U}} })  \\
& \phantom{--}+  \mathbb{V} (w_{F_{\mathcal{D}}(Y^{nB}_{\mathcal{D}})F_{\mathcal{D}} M^B_{\mathcal{D}} X^{nB}_{\mathcal{U}} },  p_{U_{\mathcal{D}}} p_{U_{\mathcal{F}}}w_{M^B_{\mathcal{D}}X^{nB}_{\mathcal{U}}  }) \\
& \phantom{--} +  \mathbb{V} ( w_{M^B_{\mathcal{D}}X^{nB}_{\mathcal{U}}  }, p_{M^B_{\mathcal{D}}X^{nB}_{\mathcal{U}}  })\\
& \leq 2\mathbb{V} ( p_{Y^{nB}_{\mathcal{D}}  M_{\mathcal{D}}^B X^{nB}_{\mathcal{U}} },w_{ Y^{nB}_{\mathcal{D}}  M^B_{\mathcal{D}} X^{nB}_{\mathcal{U}} }) \\
& \phantom{--} +  \mathbb{V} (w_{F_{\mathcal{D}}(Y^{nB}_{\mathcal{D}})F_{\mathcal{D}} M^B_{\mathcal{D}} X^{nB}_{\mathcal{U}} },  p_{U_{\mathcal{D}}} p_{U_{\mathcal{F}}}w_{M^B_{\mathcal{D}}X^{nB}_{\mathcal{U}}  }) \\
& \stackrel{(c)} \leq  2 \epsilon+\mathbb{V} ( w_{F_{\mathcal{D}}(Y^{nB}_{\mathcal{D}})F_{\mathcal{D}} M_{\mathcal{D}} X^{nB}_{\mathcal{U}} }, p_{U_{\mathcal{D}}} p_{U_{\mathcal{F}}}w_{M_{\mathcal{D}}X^{nB}_{\mathcal{U}}  })  \\ 
& \stackrel{(d)} \leq     2 \epsilon+   \sqrt{ \sum_{ \substack{ \mathcal{S} \subseteq {\mathcal{D}} \\ \mathcal{S}\neq \emptyset }} 2^{ r_{\mathcal{S}} - H_{\infty}\left( w_{Y_{\mathcal{S}}^{nB}M^B_{\mathcal{D}}  X^{nB}_{\mathcal{U}}}| p_{M^B_{\mathcal{D}}  X^{nB}_{\mathcal{U}} } \right) }}  \numberthis \label{eqremark1} \\ 
&\stackrel{(e)} \leq      2 \epsilon+  \sqrt{ \sum_{ \substack{ \mathcal{S} \subseteq {\mathcal{D}} \\ \mathcal{S}\neq \emptyset }} 2^{ r_{\mathcal{S}} - BH \left( Y^n_{\mathcal{S}}|M_{\mathcal{D}}  X^n_{\mathcal{U}}   \right)+ B \delta_{\mathcal{S}}(n,B)} }  ,
\end{align*}
where $(a)$ holds by the triangle inequality, $(b)$ holds by the data processing inequality, $(c)$ holds by \eqref{eqepsi}, $(d)$ holds by Lemma \ref{lemloh}, $(e)$ holds by \eqref{eqepsi2}.

\section{Proof of Lemma \ref{lemamp2}} \label{Applemamp2}
Let $\mathcal{S} \subseteq \mathcal{D}$, $\mathcal{S}\neq \emptyset$. We have
\begin{align}
&H(Y_{\mathcal{S}}^n|M_{\mathcal{D}} X^n_{\mathcal{U}})  \nonumber \\  \nonumber
& = H(Y_{\mathcal{S}}^n M_{\mathcal{D}} X^n_{\mathcal{U}}) - H(M_{\mathcal{D}} X^n_{\mathcal{U}})\\ \nonumber
& =H(Y_{\mathcal{S}}^n | X^n_{\mathcal{U}}) + H( M_{\mathcal{S}^c} |Y_{\mathcal{S}}^n X^n_{\mathcal{U}})- H(M_{\mathcal{D}} | X^n_{\mathcal{U}})\\ 
& \geq H(Y_{\mathcal{S}}^n | X^n_{\mathcal{U}}) + H( M_{\mathcal{S}^c} |Y_{\mathcal{S}}^n X^n_{\mathcal{U}})- n\sum_{i \in \mathcal{D}} \sum_{j \in \llbracket 1 , 2^D+1\rrbracket} R_{i,j}, \label{eqamp1} 
\end{align}
where we have used in the last inequality that $H(M_{\mathcal{D}} | X^n_{\mathcal{U}})$ is upper bounded by the logarithm of the cardinality of the alphabet of $M_{\mathcal{D}}$.

The third term in the right-hand side of \eqref{eqamp1} is evaluated as follows.
\begin{align}
&\sum_{i \in \mathcal{D}} \sum_{j \in \llbracket 1 , 2^D+1\rrbracket} R_{i,j} \nonumber \\ \nonumber
& \stackrel{(a)}{=} \sum_{i \in \mathcal{D}} (\bar{R}_{i,2^D}+ R_{i,2^D+1}) \\ \nonumber
& \stackrel{(b)}{=}  \sum_{i \in \mathcal{D}} (H(Y_i|Y_{1:i-1} X_{\mathcal{L}})  + \epsilon H(Y_i|Y_{1:i-1} X_{\mathcal{L}}) + \epsilon) \\
& \stackrel{(c)}= H(Y_{\mathcal{D}}| X_{\mathcal{L}})  + \epsilon( H(Y_{\mathcal{D}}| X_{\mathcal{L}}) + D), \label{eqamp2} 
\end{align}
where $(a)$ and $(b)$ holds by the definitions and rates chosen in Section \ref{sec:reconcil}, $(c)$ holds by the chain rule.

Next, the second term in the right-hand side of \eqref{eqamp1} is lower bounded as follows.
\begin{align}
& H( M_{\mathcal{S}^c} |Y_{\mathcal{S}}^n X^n_{\mathcal{U}}) \nonumber \\
& \stackrel{(a)}{\geq} H( M_{\mathcal{S}^c} |Y_{\mathcal{S}}^n X^n_{\mathcal{L}}) \nonumber \\ \nonumber
& \stackrel{(b)}{\geq} \sum_{i \in \mathcal{S}^c} H( M_{i} | M_{1:i-1}Y_{\mathcal{S}}^n X^n_{\mathcal{L}})\\ \nonumber
&\stackrel{(c)}{\geq} \sum_{i \in \mathcal{S}^c} H( M_{i} | Y_{1:i-1}^nY_{\mathcal{S}}^n X^n_{\mathcal{L}})\\ \nonumber
&\stackrel{(d)}{\geq} \sum_{i \in \mathcal{S}^c} H( M_{i,1:j} | Y_{1:i-1}^nY_{\mathcal{S}}^n X^n_{\mathcal{L}})\\ \nonumber
& = \sum_{i \in \mathcal{S}^c} H( M_{i,1:j} ) - I(M_{i,1:j} ; Y_{1:i-1}^nY_{\mathcal{S}}^n X^n_{\mathcal{L}})\\ \nonumber
& \stackrel{(e)}{\geq} \sum_{i \in \mathcal{S}^c} H( M_{i,1:j} ) - o(1)\\ \nonumber
& \stackrel{(f)}{=} \smash{ \sum_{i \in \mathcal{S}^c}} n( H( Y_i | Y_{1:i-1}Y_{\mathcal{S}} X_{\mathcal{L}} )\\\nonumber
& \phantom{-----} - 3 \epsilon H( Y_i | Y_{1:i-1} X_{\mathcal{L}} ) - \epsilon) - o(1)\\
& \stackrel{(g)}{=}   n\left[ H( Y_{\mathcal{S}^c} | Y_{\mathcal{S}} X_{\mathcal{L}} ) - \delta(\epsilon)\right] - o(1), \label{eqamp3}
\end{align}
where $(a)$ and $(b)$ holds because conditioning reduces entropy, $(c)$ holds because $M_{1:i-1}$ is a function of $Y_{1:i-1}^n$, $(d)$ holds because $M_{i}$ contains $M_{i,1:j}$ for any $j \in \llbracket 1, 2^D\rrbracket$ by the construction in Section \ref{sec:rec}, $(e)$ holds by Property (ii) in Section \ref{sec:rec} and \cite[Lemma 1]{csiszar1996almost}, $(f)$ holds by the rates chosen in Section \ref{sec:reconcil} and  Property (ii) in Section~\ref{sec:rec} with \cite[Lemma~2.7]{bookCsizar}, and in $(g)$ we have defined $\delta(\epsilon)\triangleq \epsilon \left( \sum_{i \in \mathcal{S}^c} (3H( Y_i | Y_{1:i-1} X_{\mathcal{L}} ) +1) \right)$.
Hence, combining \eqref{eqamp1}, \eqref{eqamp2}, \eqref{eqamp3}, we obtain
\begin{align*}
&H(Y_{\mathcal{S}}^n|M_{\mathcal{D}} X^n_{\mathcal{U}})  \\
&  \geq  H(Y_{\mathcal{S}}^n | X^n_{\mathcal{U}}) + n\left[ H( Y_{\mathcal{S}^c} | Y_{\mathcal{S}} X_{\mathcal{L}} ) - \delta(\epsilon)\right] - o(1)\\
& \phantom{--}  - n(H(Y_{\mathcal{D}}| X_{\mathcal{L}})  + \epsilon( H(Y_{\mathcal{D}}| X_{\mathcal{L}}) + D)) \\
& = n \left[ H(Y_{\mathcal{S}}| X_{\mathcal{U}} )  -H(Y_{\mathcal{S}}| X_{\mathcal{L}}) \right]  \\
& \phantom{--}- n \delta(\epsilon) - o(1) - n \epsilon( H(Y_{\mathcal{D}}| X_{\mathcal{L}}) + D).
\end{align*}

 \section{Proof of Lemma \ref{lemsc1}} \label{Applemsc1}
 
 For any $\mathcal{T} \subsetneq \mathcal{L}$, we have
\begin{align*}
&I(S_{1},S_{2}; M_{1},M_{2}, X_{\mathcal{T}}^n)\\
& \stackrel{(a)} = I(S_{1}; M_{1},M_{2}, X_{\mathcal{T}}^n) +  I(S_{2}; M_{1},M_{2}, X_{\mathcal{T}}^n|S_{1})\\
& \stackrel{(b)} = I(S_{1}; M_{1}, X_{\mathcal{T}}^n) + I(S_{1}; M_{2}| M_{1}, X_{\mathcal{T}}^n) \\
&\phantom{-} +  I(S_{2}; M_{1},M_{2}, X_{\mathcal{T}}^n|S_{1})\\
& \stackrel{(c)} \leq I(S_{1}; M_{1}, X_{\mathcal{T}}^n) + I(S_{1}; M_{2}| M_{1}, X_{\mathcal{T}}^n) \\
&\phantom{-} +  I(S_{2}; M_{2}, X_{\mathcal{T}}^n,Y_{1}^n )\\
& \stackrel{(d)} \leq  \delta(n) + I(S_{1}; M_{2}| M_{1}, X_{\mathcal{T}}^n)\\
& \stackrel{(e)} =  \delta(n) + I(S_{1}; M''_{2}| M_{1}, X_{\mathcal{T}}^n) + I(S_{1}; M_{2}'| M''_{2} , M_{1}, X_{\mathcal{T}}^n)\\
& \stackrel{(f)} \leq  \delta(n) + I(Y^n_{1}, X_{\mathcal{L}}^n ; M''_{2}) + I(S_{1},M_{2,3}; M_{2}'| M''_{2}, M_{1}, X_{\mathcal{T}}^n)\\
& \stackrel{(g)} \leq  \delta(n)  + |M_{2}'| - H( K_2| S_1,M_{2,3},M''_{2}, M_{1}, X_{\mathcal{T}}^n)\\
&\stackrel{(h)} =  \delta(n)  + |M_{2}'| - H( K_2)\\
& \stackrel{(i)}= \delta(n)  + \delta(n_2'), \numberthis \label{eqf1}
\end{align*}
where $(a)$ and $(b)$ hold by the chain rule, $(c)$ holds because $(M_1,S_1)$ is a function of $Y_1^n$, $(d)$~holds by \eqref{eqreq1} and \eqref{eqreq2}, $(e)$ holds by the chain rule and the definition of $M_2$, $(f)$ holds because $I(S_{1}; M''_{2}| M_{1}, X_{\mathcal{T}}^n) \leq I( M_{1}, X_{\mathcal{T}}^n, S_{1}; M''_{2}) \leq  I( Y_1^n, X_{\mathcal{L}}^n; M''_{2})$, where the first inequality holds by the chain rule and the second inequality holds as in $(c)$, $(g)$ holds by~\eqref{eqreq2bbb} and by the definition of $M_2'$, $(h)$ holds by independence of the initialization phase and the successive secret distribution phase, $(i)$ holds by almost uniformity of $K_2$ in the initialization phase.
 \section{Proof of Lemma \ref{lemsc2}} \label{Applemsc2}

We have for any $\mathcal{U}, \mathcal{T} \subsetneq \mathcal{L}$,
\begin{align*}
&I(S_{1},S_{2}; I_2(\mathcal{U}), M_{1},M_{2}, X_{\mathcal{T}}^n) - \delta(n)  - \delta(n_2') \\
& \stackrel{(a)} \leq   I(S_{1},S_{2};I_2(\mathcal{U})| M_{1},M_{2}, X_{\mathcal{T}}^n) \\ 
& \leq  I(S_{1},S_{2}, M_{1},M_{2};I_2(\mathcal{U})| X_{\mathcal{T}}^n) \\ 
& \stackrel{(b)} =  I( M_{1},M_{2}; I_2(\mathcal{U})| S_{1},S_{2}, X_{\mathcal{T}}^n) \\ 
&  =  I( M_{1} ;  I_2(\mathcal{U})| S_{1},S_{2}, X_{\mathcal{T}}^n)+I( M_{2};  I_2(\mathcal{U})|M_{1}, S_{1},S_{2}, X_{\mathcal{T}}^n) \\ 
& \stackrel{(c)} =  I( M_{2}';  I_2(\mathcal{U})|M_2'',M_{1}, S_{1},S_{2}, X_{\mathcal{T}}^n) \\ 
& \leq  I( M_{2}'; M_{2,3},I_2(\mathcal{U})|M_2'',M_{1}, S_{1},S_{2}, X_{\mathcal{T}}^n) \\ 
& \stackrel{(d)} \leq |K_2| - H(K_2|M_{2,3},  I_2(\mathcal{U}),M_2'',M_{1}, S_{1},S_{2}, X_{\mathcal{T}}^n) \\ 
& \stackrel{(e)}=  |K_2| - H(K_2|  I_2(\mathcal{U}) ) \\ 
& \stackrel{(f)}=  |K_2| - H(K_2 ) + I(K_2;  I_2(\mathcal{U}) )  \\
& \stackrel{(g)}\leq  \delta(n'_2),
\end{align*}
where $(a)$ holds by the chain rule and \eqref{eqf1}, $(b)$ holds because $I(S_{1},S_{2} ; I_2(\mathcal{U})| X_{\mathcal{T}}^n) \leq I(S_{1},S_{2},X_{\mathcal{T}}^n ;I_2(\mathcal{U})) =0$, where the equality holds by independence of the initialization phase and the successive secret distribution phase, $(c)$ holds by the definition of $M_2$, the chain rule, and independence of the initialization phase and the successive secret distribution phase, $(d)$ holds by the definition of $M_2'$, $(e)$ holds by the independence of the initialization phase and the successive secret distribution phase, $(g)$ holds by the initialization~phase.

\bibliographystyle{IEEEtran}
\bibliography{polarwiretap}

\begin{thebibliography}{10}
\providecommand{\url}[1]{#1}
\csname url@samestyle\endcsname
\providecommand{\newblock}{\relax}
\providecommand{\bibinfo}[2]{#2}
\providecommand{\BIBentrySTDinterwordspacing}{\spaceskip=0pt\relax}
\providecommand{\BIBentryALTinterwordstretchfactor}{4}
\providecommand{\BIBentryALTinterwordspacing}{\spaceskip=\fontdimen2\font plus
\BIBentryALTinterwordstretchfactor\fontdimen3\font minus
  \fontdimen4\font\relax}
\providecommand{\BIBforeignlanguage}[2]{{%
\expandafter\ifx\csname l@#1\endcsname\relax
\typeout{** WARNING: IEEEtran.bst: No hyphenation pattern has been}%
\typeout{** loaded for the language `#1'. Using the pattern for}%
\typeout{** the default language instead.}%
\else
\language=\csname l@#1\endcsname
\fi
#2}}
\providecommand{\BIBdecl}{\relax}
\BIBdecl

\bibitem{chou2018secret}
R.~Chou, ``Secret sharing over a public channel from correlated random
  variables,'' in \emph{IEEE Int. Symp. Inf. Theory}, 2018, pp. 991--995.

\bibitem{shamir1979share}
A.~Shamir, ``How to share a secret,'' \emph{Communications of the ACM},
  vol.~22, no.~11, pp. 612--613, 1979.

\bibitem{blakley}
G.~Blakley, ``Safeguarding cryptographic keys,'' \emph{Proceedings of the
  National Computer Conference}, pp. 313--317, 1979.

\bibitem{stinson2005cryptography}
D.~Stinson, \emph{Cryptography: Theory and practice}.\hskip 1em plus 0.5em
  minus 0.4em\relax CRC press, 2005.

\bibitem{beimel}
A.~Beimel, ``Secret-sharing schemes: {A} survey,'' in \emph{International
  Conference on Coding and Cryptology}, Berlin, Heidelberg, 2011, pp. 11--46.

\bibitem{zou2015information}
S.~Zou, Y.~Liang, L.~Lai, and S.~Shamai, ``An information theoretic approach to
  secret sharing,'' \emph{{IEEE} {T}rans. {I}nf. {T}heory}, vol.~61, no.~6, pp.
  3121--3136, 2015.

\bibitem{Liang09}
Y.~Liang, G.~Kramer, H.~V. Poor, and S.~Shamai, ``Compound wiretap channels,''
  \emph{{EURASIP} {J}. {W}irel. {C}ommun. {N}etw.}, vol. 2009, pp. 5:1--5:12,
  2009.

\bibitem{wilson2007channel}
R.~Wilson, D.~Tse, and R.~Scholtz, ``Channel identification: Secret sharing
  using reciprocity in ultrawideband channels,'' \emph{IEEE Trans. Inf.
  Forensics and Secur.}, vol.~2, no.~3, pp. 364--375, 2007.

\bibitem{wallace2010automatic}
J.~Wallace and R.~Sharma, ``Automatic secret keys from reciprocal {MIMO}
  wireless channels: Measurement and analysis,'' \emph{IEEE Trans. Inf.
  Forensics and Secur.}, vol.~5, no.~3, pp. 381--392, 2010.

\bibitem{ye2010information}
C.~Ye, S.~Mathur, A.~Reznik, Y.~Shah, W.~Trappe, and N.~Mandayam,
  ``Information-theoretically secret key generation for fading wireless
  channels,'' \emph{IEEE Trans. Inf. Forensics and Secur.}, vol.~5, no.~2, pp.
  240--254, 2010.

\bibitem{pierrot2013experimental}
A.~Pierrot, R.~Chou, and M.~Bloch, ``Experimental aspects of secret key
  generation in indoor wireless environments,'' in \emph{IEEE 14th Workshop on
  Signal Processing Advances in Wireless Communications}, 2013, pp. 669--673.

\bibitem{Maurer93}
U.~Maurer, ``Secret {K}ey {A}greement by {P}ublic {D}iscussion from {C}ommon
  {I}nformation,'' \emph{{IEEE} {T}rans. {I}nf. {T}heory}, vol.~39, pp.
  733--742, 1993.

\bibitem{Ahlswede93}
R.~Ahlswede and I.~Csisz\'{a}r, ``Common {R}andomness in {I}nformation {T}heory
  and {C}ryptography {P}art {I}: {S}ecret {S}haring,'' \emph{{IEEE} {T}rans.
  {I}nf. {T}heory}, vol.~39, pp. 1121--1132, 1993.

\bibitem{Csiszar04}
I.~Csisz{\'a}r and P.~Narayan, ``Secrecy {C}apacities for {M}ultiple
  {T}erminals.'' \emph{{IEEE} {T}rans. {I}nf. {T}heory}, vol.~50, no.~12, pp.
  3047--3061, 2004.

\bibitem{tavangaran2016secret}
N.~Tavangaran, H.~Boche, and R.~Schaefer, ``Secret-key generation using
  compound sources and one-way public communication,'' \emph{IEEE Trans. Inf.
  Forensics and Secur.}, vol.~12, no.~1, pp. 227--241, 2016.

\bibitem{zhang2017multi}
H.~Zhang, Y.~Liang, L.~Lai, and S.~S. Shitz, ``Multi-key generation over a
  cellular model with a helper,'' \emph{{IEEE} {T}rans. {I}nf. {T}heory},
  vol.~63, no.~6, pp. 3804--3822, 2017.

\bibitem{gohari2010information}
A.~Gohari and V.~Anantharam, ``Information-theoretic key agreement of multiple
  terminals-{P}art {I},'' \emph{{IEEE} {T}rans. {I}nf. {T}heory}, vol.~56,
  no.~8, pp. 3973--3996, 2010.

\bibitem{wullschleger2007oblivious}
J.~Wullschleger, ``Oblivious-transfer amplification,'' in \emph{Proc. of the
  Annual International Conference on the Theory and Applications of
  Cryptographic Techniques}.\hskip 1em plus 0.5em minus 0.4em\relax Springer,
  2007, pp. 555--572.

\bibitem{nascimento2008oblivious}
A.~Nascimento and A.~Winter, ``On the oblivious-transfer capacity of noisy
  resources,'' \emph{{IEEE} {T}rans. {I}nf. {T}heory}, vol.~54, no.~6, pp.
  2572--2581, 2008.

\bibitem{chou2017secret}
R.~Chou and A.~Yener, ``Secret-key generation in many-to-one networks: An
  integrated game-theoretic and information-theoretic approach,'' \emph{{IEEE}
  {T}rans. {I}nf. {T}heory}, vol.~65, no.~8, pp. 5144--5159, 2019.

\bibitem{haastad1999pseudorandom}
J.~H{\aa}stad, R.~Impagliazzo, L.~A. Levin, and M.~Luby, ``A pseudorandom
  generator from any one-way function,'' \emph{SIAM Journal on Computing},
  vol.~28, no.~4, pp. 1364--1396, 1999.

\bibitem{Bennett95}
C.~Bennett, G.~Brassard, and U.~Maurer, ``Generalized privacy amplification,''
  \emph{{IEEE} {T}rans. {I}nf. {T}heory}, vol.~41, pp. 1915--1923, 1995.

\bibitem{cachin1997linking}
C.~Cachin and U.~Maurer, ``Linking information reconciliation and privacy
  amplification,'' \emph{Journal of Cryptology}, vol.~10, no.~2, pp. 97--110,
  1997.

\bibitem{Maurer00}
U.~Maurer and S.~Wolf, ``Information-{T}heoretic {K}ey {A}greement: {F}rom
  {W}eak to {S}trong {S}ecrecy for {F}ree,'' in \emph{Lecture Notes in Computer
  Science}.\hskip 1em plus 0.5em minus 0.4em\relax Springer-Verlag, 2000, pp.
  351--368.

\bibitem{rana2021information}
V.~Rana, R.~A. Chou, and H.~M. Kwon, ``Information-theoretic secret sharing
  from correlated {G}aussian random variables and public communication,''
  \emph{{IEEE} {T}rans. {I}nf. {T}heory}, vol.~68, no.~1, pp. 549--559, 2021.

\bibitem{csiszar2}
I.~{Csiszar} and P.~{Narayan}, ``Capacity of a shared secret key,'' in
  \emph{IEEE Int. Symp. Inf. Theory}, 2010, pp. 2593--2596.

\bibitem{benaloh1988generalized}
J.~Benaloh and J.~Leichter, ``Generalized secret sharing and monotone
  functions,'' in \emph{Conference on the Theory and Application of
  Cryptography}.\hskip 1em plus 0.5em minus 0.4em\relax Springer, 1988, pp.
  27--35.

\bibitem{soleymani2020distributed}
M.~Soleymani and H.~Mahdavifar, ``Distributed multi-user secret sharing,''
  \emph{{IEEE} {T}rans. {I}nf. {T}heory}, vol.~67, no.~1, pp. 164--178, 2020.

\bibitem{khalesi2021capacity}
A.~Khalesi, M.~Mirmohseni, and M.~A. Maddah-Ali, ``The capacity region of
  distributed multi-user secret sharing,'' \emph{IEEE Journal on Selected Areas
  in Information Theory}, vol.~2, no.~3, pp. 1057--1071, 2021.

\bibitem{Orlitsky01}
A.~Orlitsky and J.~Roche, ``Coding for {C}omputing,'' \emph{{IEEE} {T}rans.
  {I}nf. {T}heory}, vol.~47, no.~3, 2001.

\bibitem{bookCsizar}
I.~Csisz\'{a}r and J.~K\"{o}rner, \emph{Information Theory: Coding Theorems for
  Discrete Memoryless Systems}.\hskip 1em plus 0.5em minus 0.4em\relax
  Cambridge Univ Pr, 1981.

\bibitem{Carter79}
L.~Carter and M.~Wegman, ``{U}niversal {C}lasses of {H}ash {F}unctions,''
  \emph{{J}ournal of {C}omputer and {S}ystem {S}ciences}, vol.~18, no.~2, pp.
  143--154, 1979.

\bibitem{dodis2008fuzzy}
Y.~Dodis, R.~Ostrovsky, L.~Reyzin, and A.~Smith, ``Fuzzy extractors: How to
  generate strong keys from biometrics and other noisy data,'' \emph{SIAM
  journal on computing}, vol.~38, no.~1, pp. 97--139, 2008.

\bibitem{chou2021private}
R.~A. Chou, ``Private classical communication over quantum multiple-access
  channels,'' \emph{{IEEE} {T}rans. {I}nf. {T}heory}, vol.~68, no.~3, pp.
  1782--1794, 2021.

\bibitem{Chou21}
------, ``Pairwise oblivious transfer,'' in \emph{IEEE Information Theory
  Workshop (ITW)}, 2021.

\bibitem{chou2022bc}
R.~Chou and M.~Bloch, ``Commitment over multiple-access channels,'' in
  \emph{58th Annual Allerton Conference on Communication, Control, and
  Computing}, 2022.

\bibitem{chou2019biometric}
R.~A. Chou, ``Biometric systems with multiuser access structures,'' in
  \emph{IEEE Int. Symp. Inf. Theory}, 2019, pp. 807--811.

\bibitem{sultana2022multiple}
R.~Sultana and R.~A. Chou, ``Multiple access channel resolvability codes from
  source resolvability codes,'' \emph{{IEEE} {T}rans. {I}nf. {T}heory}, 2022.

\bibitem{renner2008security}
R.~Renner, ``Security of quantum key distribution,'' \emph{International
  Journal of Quantum Information}, vol.~6, no.~01, pp. 1--127, 2008.

\bibitem{chou2014separation}
R.~Chou and M.~Bloch, ``Separation of reliability and secrecy in rate-limited
  secret-key generation,'' \emph{{IEEE} {T}rans. {I}nf. {T}heory}, vol.~60,
  no.~8, pp. 4941--4957, 2014.

\bibitem{chou2016coding}
R.~Chou, B.~Vellambi, M.~Bloch, and J.~Kliewer, ``Coding schemes for achieving
  strong secrecy at negligible cost,'' \emph{{IEEE} {T}rans. {I}nf. {T}heory},
  vol.~63, no.~3, pp. 1858--1873, 2016.

\bibitem{chou2013data}
R.~Chou and M.~Bloch, ``Data compression with nearly uniform output,'' in
  \emph{IEEE Int. Symp. Inf. Theory}, 2013, pp. 1979--1983.

\bibitem{sultana2021}
R.~Sultana and R.~Chou, ``Low-complexity secret sharing schemes using
  correlated random variables and rate-limited public communication,'' in
  \emph{IEEE Int. Symp. Inf. Theory}, 2021.

\bibitem{chou2020}
R.~A. Chou, ``Explicit wiretap channel codes via source coding, universal
  hashing, and distribution approximation, when the channels’ statistics are
  uncertain,'' \emph{IEEE Trans. Inf. Forensics and Secur.}, vol.~18, pp.
  117--132, 2022.

\bibitem{chou2018explicit}
R.~Chou, ``Explicit codes for the wiretap channel with uncertainty on the
  eavesdropper's channel,'' in \emph{IEEE Int. Symp. Inf. Theory}, 2018, pp.
  476--480.

\bibitem{sultana2023}
R.~Sultana, V.~Rana, and R.~Chou, ``Secret sharing over a {G}aussian broadcast
  channel: {O}ptimal coding scheme design and deep learning approach at short
  blocklength,'' in \emph{IEEE Int. Symp. Inf. Theory}, 2023.

\bibitem{kramerbook}
G.~Kramer, ``Topics in multi-user information theory,'' \emph{{F}oundations and
  {T}rends in {C}ommunications and {I}nformation {T}heory}, vol.~4, pp.
  265--444, 2007.

\bibitem{csiszar1996almost}
I.~Csisz{\'a}r, ``Almost independence and secrecy capacity,'' \emph{Problemy
  Peredachi Informatsii}, vol.~32, no.~1, pp. 48--57, 1996.

\end{thebibliography}

\end{document}